\documentclass[11pt,a4paper]{article}

\usepackage{fancyhdr}
\usepackage{amsmath}
\usepackage{bm}
\usepackage{amsfonts}
\usepackage{amsthm}
\usepackage{amssymb}
\usepackage{amsbsy}
\usepackage{graphicx}
\usepackage{mathrsfs}
\usepackage[hmargin = 1in,vmargin=.75in]{geometry}
\usepackage{color}
\usepackage{enumerate}
\usepackage{algpseudocode}
\usepackage{algorithm}
\usepackage{pifont}
\usepackage{prettyref}

\font\msbm=msbm10

\numberwithin{equation}{section}

\theoremstyle{plain}
\newtheorem{theorem}{Theorem}[section]
\newtheorem{lemma}[theorem]{Lemma}
\newtheorem{corollary}[theorem]{Corollary}

\newtheorem{condition}[theorem]{Condition}

\newtheorem{remark}[theorem]{Remark}
\def\mathbb#1{\hbox{\msbm{#1}}}

\newcommand{\ba}{\boldsymbol{a}}
\newcommand{\bb}{\boldsymbol{b}}
\newcommand{\bc}{\boldsymbol{c}}

\newcommand{\be}{\boldsymbol{e}}
\newcommand{\bh}{\boldsymbol{h}}
\newcommand{\bff}{\boldsymbol{f}}
\newcommand{\bg}{\boldsymbol{g}}

\newcommand{\bn}{\boldsymbol{n}}

\newcommand{\bq}{\boldsymbol{q}}

\newcommand{\bu}{\boldsymbol{u}}
\newcommand{\bv}{\boldsymbol{v}}
\newcommand{\bw}{\boldsymbol{w}}
\newcommand{\bx}{\boldsymbol{x}}
\newcommand{\by}{\boldsymbol{y}}
\newcommand{\bz}{\boldsymbol{z}}

\newcommand{\BA}{\boldsymbol{A}}
\newcommand{\BB}{\boldsymbol{B}}
\newcommand{\BC}{\boldsymbol{C}}
\newcommand{\BD}{\boldsymbol{D}}

\newcommand{\BF}{\boldsymbol{F}}

\newcommand{\BH}{\boldsymbol{H}}

\newcommand{\BU}{\boldsymbol{U}}
\newcommand{\BV}{\boldsymbol{V}}

\newcommand{\BX}{\boldsymbol{X}}

\newcommand{\BZ}{\boldsymbol{Z}}

\newcommand{\tF}{\widetilde{F}}

\newcommand{\hbh}{\hat{\boldsymbol{h}}}

\newcommand{\hbx}{\hat{\boldsymbol{x}}}

\newcommand{\Keps}{\mathcal{N}_{\epsilon}}

\newcommand{\Kmu}{\mathcal{N}_{\mu}}
\newcommand{\Kd}{\mathcal{N}_d}
\newcommand{\KF}{\mathcal{N}_{\widetilde{F}}}
\newcommand{\Kint}{\mathcal{N}_d \cap \mathcal{N}_{\mu} \cap \mathcal{N}_{\epsilon}}

\newcommand{\bzero}{\boldsymbol{0}}

\newcommand{\A}{\mathcal{A}}

\newcommand{\PP}{\mathcal{P}}

\newcommand{\pa}{\partial}

\newcommand{\CC}{\mathbb{C}}

\newcommand{\CZ}{\mathcal{Z}}
\newcommand{\CN}{\mathcal{C}\mathcal{N}}

\newcommand{\CH}{\mathcal{H}}

\newcommand{\Dh}{\Delta\boldsymbol{h}}
\newcommand{\Dx}{\Delta\boldsymbol{x}}

\newcommand{\MS}{\mathcal{S}}

\newcommand{\MN}{\mathcal{N}}

\newcommand{\I}{\boldsymbol{I}}
\newcommand{\RR}{\mathbb{R}}
\newcommand{\ZZ}{\mathbb{Z}}
\newcommand{\lag}{\langle}
\newcommand{\rag}{\rangle}

\newcommand{\lp}{\left(} 
\newcommand{\rp}{\right)} 

\newcommand{\Tr}{\text{Tr}}
\newcommand{\eps}{\varepsilon}
\newcommand{\TB}{T^{\bot}}

\newcommand*\diff{\mathop{}\!\mathrm{d}}
\DeclareMathOperator{\blkdiag}{blkdiag}

\DeclareMathOperator{\Real}{Re}

\DeclareMathOperator{\E}{\mathbb{E}}

\DeclareMathOperator{\diag}{diag}

\newcommand{\vct}[1]{\bm{#1}}
\newcommand{\mtx}[1]{\bm{#1}}
\definecolor{xl}{RGB}{200,50,120}

\begin{document}
\title{\bf Regularized Gradient Descent: A Nonconvex Recipe \\ for Fast Joint Blind Deconvolution and Demixing \thanks{The authors acknowledge support from the NSF via grants DTRA-DMS 1322393 and DMS 1620455.}}


\author{Shuyang Ling\thanks{Courant Institute of Mathematical Sciences, New York University (Email: sling@cims.nyu.edu).} \quad Thomas Strohmer\thanks{Department of Mathematics, University of California at Davis (Email: strohmer@math.ucdavis.edu).} }


\maketitle

\begin{abstract}

We study the question of extracting a sequence of  functions $\{\boldsymbol{f}_i, \boldsymbol{g}_i\}_{i=1}^s$ from observing only the sum of their
convolutions, i.e.,  from $\boldsymbol{y} = \sum_{i=1}^s \boldsymbol{f}_i\ast \boldsymbol{g}_i$.  
While convex optimization techniques are able to solve this joint blind deconvolution-demixing problem provably and robustly under certain
conditions, for medium-size or large-size problems we need computationally faster methods without sacrificing the benefits of mathematical rigor that come with convex methods.
In this paper we present a non-convex algorithm which guarantees exact recovery under conditions that are competitive with convex optimization methods,
with the additional advantage of being computationally much more efficient.
Our two-step algorithm converges to the global minimum linearly and is also
robust in the presence of additive noise. While the derived performance bounds are suboptimal in terms of the information-theoretic limit, numerical simulations show remarkable performance even if the number of measurements is close to the number of degrees of freedom.  We discuss an application of the proposed framework in wireless communications in connection with the Internet-of-Things.
\end{abstract}

\section{Introduction}
\label{s:intro}
The goal of blind deconvolution is the task of estimating  two unknown functions from their convolution. While it is a highly ill-posed bilinear inverse problem, blind deconvolution is also an extremely important problem in signal processing~\cite{RR12}, communications engineering~\cite{WP98}, imaging processing~\cite{CE07}, audio  processing~\cite{LXQZ09}, etc.
In this paper, we deal with an even more difficult and more general variation of the blind deconvolution problem, in which we have to extract multiple convolved signals mixed together in one observation signal. This joint blind deconvolution-demixing problem arises in a range of applications such as
acoustics~\cite{LXQZ09}, dictionary learning~\cite{bristow2013fast}, and wireless communications~\cite{WP98}.

We briefly discuss one such application in more detail. Blind deconvolution/demixing problems are expected to play a vital role in the future Internet-of-Things.
The Internet-of-Things will connect billions of wireless devices, which is far more than the current wireless systems can technically and economically accommodate. One of the many challenges in the design of the Internet-of-Things will be its ability to manage the massive number of sporadic traffic
generating devices which are most of the time inactive, but regularly access the network for minor updates with no human interaction~\cite{WBSJ14}.
This means among others that the overhead caused by the exchange of certain types of information between transmitter and receiver, such as channel
estimation, assignment of data slots, etc, has to be avoided as much as possible~\cite{5Gbook,shafi20175g}.

Focusing on the underlying mathematical challenges, we consider a {\em multi-user communication} scenario where many different users/devices communicate with a common base station,  as illustrated in Figure~\ref{fig:demix}.
Suppose we have $s$ users and each of them sends a signal $\bg_i$ through an unknown channel (which differs from user to user)  to a common base station,.  We assume that the $i$-th channel, represented by its impulse response $\bff_i$, does not change during the transmission of
the signal $\bg_i$. Therefore $\bff_i$ acts as convolution operator, i.e., the signal transmitted by the $i$-th user arriving at the base station becomes $\bff_i \ast \bg_i$, 
where ``$\ast$" denotes  convolution. 
\begin{figure}[h!]
\centering
\includegraphics[width=140mm]{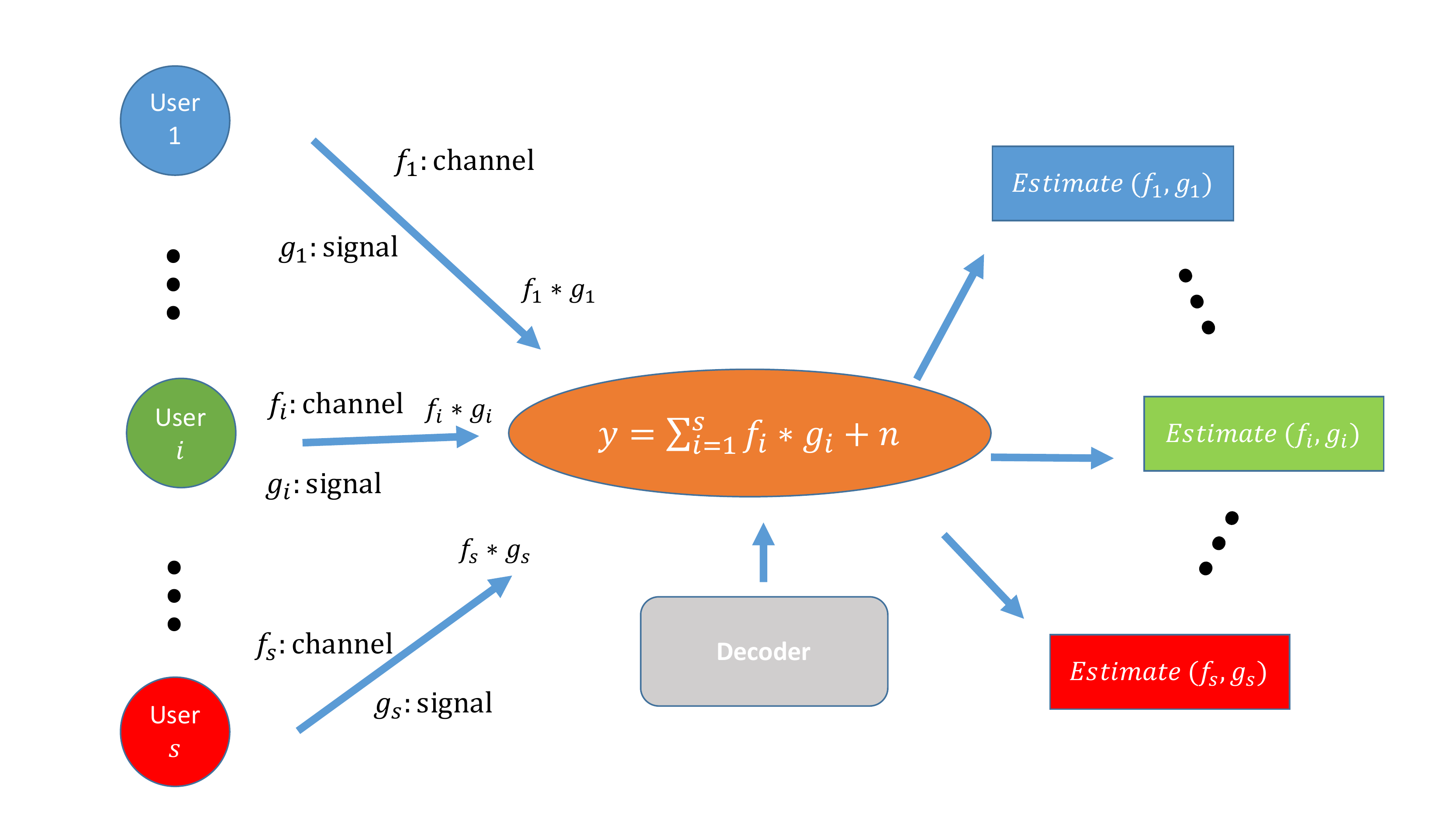}
\caption{Single-antenna multi-user communication scenario without explicit channel estimation: Each of the  $s$ users sends a signal $\bg_i$ through an unknown channel $\bff_i$  to a common base station. The base station measures  the superposition of all those signals, namely, 
$\by = \sum_{i=1}^s \bff_i\ast \bg_i $ (plus noise).  The goal is to extract all pairs of $\{(\bff_i,\bg_i)\}_{i=1}^s$ simultaneously from $\by$.}
\label{fig:demix}
\end{figure}
The antenna at the base station, instead of receiving each individual component $\bff_i\ast \bg_i$, is only 
able to record the superposition of all those signals, namely, 
\begin{equation}\label{eq:model0}
\by = \sum_{i=1}^s \bff_i\ast \bg_i +\bn,
\end{equation} where  $\bn$ represents noise. We aim to develop a fast algorithm to simultaneously extract all pairs $\{(\bff_i,\bg_i)\}_{i=1}^s$ from $\by$ (i.e., estimating the channel/impulse responses $\bff_i$ and the signals $\bg_i$ jointly) in a numerically efficient and robust way, while keeping the number of required measurements as small as possible.

\subsection{State of the art and contributions of this paper}
A thorough theoretical analysis  concerning the solvability of demixing problems via convex optimization can be found in~\cite{mccoy2013demixing}. 
There, the authors derive explicit sharp bounds and phase transitions regarding the number of measurements required to successfully demix structured signals (such as sparse signals or low-rank matrices) from a single measurement vector. In principle we could recast the blind deconvolution/demixing problem as the demixing of a sum of rank-one matrices, see~\eqref{eq:measure2}. As such, it seems to fit into the framework analyzed by McCoy and Tropp. However, the setup in~\cite{mccoy2013demixing} differs from ours in a crucial manner. McCoy and Tropp consider as measurement matrices (see the matrices $\A_i$ in~\eqref{eq:measure2}) full-rank random matrices, while in our setting the measurement matrices
are rank-one.  This difference fundamentally changes the theoretical analysis. 
The findings in~\cite{mccoy2013demixing} are therefore not applicable to the problem of joint blind deconvolution/demixing. The compressive principal component analysis in~\cite{WGMM13} is also a form of demixing problem, but its setting is only vaguely related to ours. There is a large amount of literature on demixing problems, but the vast majority does not have a ``blind deconvolution component'', therefore this body of work is only marginally related to the topic of our paper.

Blind deconvolution/demixing problems also appear in convolutional dictionary learning, see e.g.~\cite{bristow2013fast}. 
There, the aim is to factorize an ensemble of input vectors into a linear combination of overcomplete basis elements which are modeled as shift-invariant---the latter property is why the factorization turns into a convolution.  The setup is similar to~\eqref{eq:model0}, but with an additional penalty term to enforce sparsity of the convolving filters. The existing literature on convolutional dictionary learning is mainly focused on empirical results, therefore there is little overlap with our work. But it is an interesting challenge for future research to see whether the approach in this paper can be modified to provide a fast and theoretically sound solver for the sparse convolutional coding problem.

There are numerous papers concerned with blind deconvolution/demixing problems in the area of wireless communications~\cite{sudhakar2010double,Ver98,li2001direct}. But the majority of these papers assumes the availability of multiple measurement vectors, which makes the problem significantly easier. Those methods however cannot be applied to the case of a single measurement vector, which is the focus of this paper. Thus there is essentially no overlap of those papers with our work.
 
Our previous paper~\cite{LS17b} solves~\eqref{eq:model0} under subspace conditions, i.e., assuming that both $\bff_i$ and $\bg_i$ belong to  known linear subspaces. This contributes to generalizing the pioneering work by Ahmed, Recht, and Romberg~\cite{RR12} from the ``single-user" scenario to the ``multi-user" scenario. Both~\cite{RR12} and \cite{LS17b} employ a two-step convex approach: first  ``lifting"~\cite{CSV11}  is used and then the 
lifted version of the original bilinear inverse problems is relaxed into a  semi-definite program. An improvement of the theoretical bounds in~\cite{LS17b} was announced in~\cite{SJK16}.

While the convex approach is certainly effective and elegant, it can hardly handle large-scale problems. This motivates us to apply a nonconvex optimization approach~\cite{CLS14, LLSW16} to this blind-deconvolution-blind-demixing problem. The mathematical challenge, when using non-convex methods, is to derive a rigorous convergence framework with conditions that are competitive with those in a convex framework.

In the last few years several excellent articles have appeared on provably convergent nonconvex optimization applied to various problems in signal processing and machine learning, e.g., matrix completion~\cite{KMO09, KMO09b,SL16}, phase retrieval~\cite{CLS14,CC15,SQW16,CLM16}, blind deconvolution~\cite{LLB16,CJ16b,LLSW16}, dictionary learning~\cite{SQW16it}, super-resolution~\cite{EftW15} and low-rank matrix recovery~\cite{TBSR15,WCCL16}. In this paper we derive the first nonconvex optimization algorithm to solve~\eqref{eq:model0} fast and with rigorous theoretical guarantees concerning exact recovery, convergence rates, as well as robustness for noisy data. Our work can be viewed as a generalization of blind deconvolution~\cite{LLSW16} $(s=1)$ to the multi-user scenario $(s > 1)$. 

The idea behind our approach is strongly motivated by the nonconvex
optimization algorithm for phase retrieval proposed in~\cite{CLS14}. In this foundational paper, the authors use a two-step approach: (i) Construct a
good initial guess with a numerically efficient algorithm; (ii) Starting with this initial guess, prove that simple gradient descent will converge
to the true solution. Our paper follows a similar two-step scheme. However, the techniques used here are quite different from~\cite{CLS14}. 
Like the matrix completion problem~\cite{CR08}, the performance of the algorithm relies heavily and inherently on how much the ground truth signals are aligned with the design matrix. 
Due to this so-called ``incoherence" issue, we need to impose extra constraints, which results in a different construction of the so-called~\emph{basin of attraction}. Therefore, influenced by~\cite{KMO09, SL16,LLSW16}, we add penalty terms to control the incoherence and this leads to the regularized gradient descent method, which forms the core of our proposed algorithm.

To the best of our knowledge, our algorithm is the first  algorithm for the blind deconvolution/blind demixing
problem that is numerically efficient, robust against noise, and comes with rigorous recovery guarantees.

\subsection{Notation}
For a matrix $\BZ$, $\|\BZ\|$
denotes its operator norm and $\|\BZ\|_F$ is its the Frobenius norm. For a vector $\bz$, $\|\bz\|$ is its Euclidean norm and $\|\bz\|_{\infty}$ is the $\ell_{\infty}$-norm. For both
matrices and vectors, $\BZ^*$ and
$\bz^*$ denote their complex conjugate transpose. $\bar{\bz}$ is the complex conjugate of $\bz$.  We equip the matrix space $\CC^{K\times N}$ with the inner
product defined by $\lag \BU, \BV\rag : =\Tr(\BU^*\BV).$ For a given vector $\bz$, $\diag(\bz)$ represents the diagonal
matrix whose diagonal entries are  $\bz$. For any $z\in\RR$, let $z_+ = \frac{z + |z|}{2}.$

\section{Preliminaries}\label{s:model}
Obviously, without any further assumption, it is impossible to solve~\eqref{eq:model0}. Therefore, we impose the following subspace assumptions throughout our discussion~\cite{RR12,LS17b}.
\begin{itemize}
\item {\bf Channel subspace assumption:} Each finite impulse response $\bff_i\in\CC^L$ is assumed to have {\em maximum delay spread} $K$, i.e., 
\begin{equation*}
\bff_i = \begin{bmatrix}
\bh_i \\
\bzero 
\end{bmatrix}.
\end{equation*}
Here $\bh_i\in\CC^K$ is the nonzero part of $\bff_i$ and $\bff_i(n) = 0 \text{ for } n > K.$

\item {\bf Signal subspace assumption:} Let $\bg_i : = \BC_i\bar{\bx}_i$ be the outcome of the signal $\bar{\bx}_i\in\CC^N$ encoded by a matrix $\BC_i\in\CC^{L\times N}$ with $L > N$, where the encoding matrix $\BC_i$ is known and assumed to have full rank\footnote{Here we use the conjugate $\bar{\bx}_i$ instead of $\bx_i$ because it will simplify our notation in later derivations.}.
\end{itemize}

\begin{remark}
Both subspace assumptions are common in various applications. For  instance in wireless communications, the channel impulse response can always be modeled to have finite support (or maximum delay spread, as it is called in engineering jargon) due to the physical properties of wave propagation~\cite{goldsmith2005wireless}; and the signal subspace assumption is a standard feature found in many current communication systems~\cite{goldsmith2005wireless}, including CDMA where $\BC_i$ is known as spreading matrix and OFDM 
where $\BC_i$ is known as precoding matrix. 
\end{remark}

The specific choice of the encoding matrices $\BC_i$ depends on a variety of conditions. 
In this paper, we derive our theory by assuming that $\BC_i$  is a complex Gaussian random matrix, i.e., each entry in $\BC_i$ is i.i.d.\ $\mathcal{C}\mathcal{N}(0,1)$. This assumption, while sometimes imposed in the wireless communications literature, is somewhat unrealistic in practice, due to the lack of a fast algorithm to apply $\BC_i$ and due to storage requirements. In practice one would rather choose $\BC_i$ to be something like the product of a Hadamard matrix
and a  diagonal matrix with random binary entries. We hope to address such more structured encoding matrices in our future research. Our numerical simulations (see Section~\ref{s:numerics}) show no difference in the performance of our algorithm for either choice.

Under the two assumptions above, the model actually has a simpler form in the~\emph{frequency} domain. We assume throughout the paper that the convolution of finite sequences is circular convolution\footnote{This circular convolution assumption can often be reinforced directly (for example in wireless communications the use of a cyclic prefix in OFDM renders the convolution circular) or indirectly (e.g.\ via zero-padding). In the first case replacing regular convolution by circular convolution does not introduce any errors at all. In the latter case one introduces an additional approximation error in the inversion which is negligible, since it decays exponentially for impulse responses of finite length~\cite{Str00}.}.
By applying the Discrete Fourier Transform DFT) to~\eqref{eq:model0} along with the two assumptions, we have
\begin{equation*}
\frac{1}{\sqrt{L}}\BF \by = \sum_{i=1}^s\diag(\BF \bh_i)(\BF\BC_i \bar{\bx}_i) + \frac{1}{\sqrt{L}}\BF\bn
\end{equation*}
where $\BF$ is the $L\times L$ normalized unitary  DFT matrix with $\BF^*\BF = \BF\BF^* = \I_L$. The noise is assumed to be additive white complex Gaussian noise with $\bn\sim \mathcal{C}\mathcal{N}(\bzero, \sigma^2d_0^2\I_L)$ where $d_0 = \sqrt{\sum_{i=1}^s \|\bh_{i0}\|^2 \|\bx_{i0}\|^2}$, and $\{(\bh_{i0}, \bx_{i0})\}_{i=1}^s$ is the ground truth. 
We define $d_{i0} = \|\bh_{i0}\bx_{i0}^*\|_F$ and assume without loss of generality that $\|\bh_{i0}\|$ and $\|\bx_{i0}\|$ are of the same norm, i.e., $\|\bh_{i0}\| = \|\bx_{i0}\| = \sqrt{d_{i0}}$, which is due to the scaling ambiguity\footnote{Namely, if the pair $(\bh_i, \bx_i)$ is a solution,
then so is $(\alpha \bh_i, \alpha^{-1} \bx_i)$ for any $\alpha \neq 0$.  }.  
In that way, $\frac{1}{\sigma^{2}}$ actually is a measure of SNR (signal to noise ratio).

Let $\bh_i\in\CC^K$ be the first $K$ nonzero entries of $\bff_i$ and $\BB\in\CC^{L\times K}$ be a low-frequency DFT matrix (the first $K$ columns of an $L\times L$ unitary DFT matrix). Then a simple relation holds,
\begin{equation*}
\BF\bff_i = {\BB}\bh_i, \quad \BB^*\BB = \I_K.
\end{equation*}
We also denote $\BA_i := \overline{\BF\BC_i}$ and $\be := \frac{1}{\sqrt{L}}\BF\bn$. Due to the Gaussianity, $\BA_i$ also possesses complex Gaussian distribution and so does $\be.$
From now on, instead of focusing on the original model, we consider (with a slight abuse of notation) the following equivalent formulation throughout our discussion:
\begin{equation}\label{eq:measure1}
\by = \sum_{i=1}^s \diag(\BB\bh_{i})\overline{\BA_i\bx_{i}} + \be, 
\end{equation}
where $\be \sim \mathcal{C}\mathcal{N}(\bzero, \frac{\sigma^2 d_0^2}{L}\I_L)$. Our goal here is to estimate all $\{\bh_i, \bx_i\}_{i=1}^s$ from $\by,\BB$ and $\{\BA_i\}_{i=1}^s$. Obviously, this is a bilinear inverse problem, i.e., if all $\{\bh_i\}_{i=1}^s$ are given, it is a linear inverse problem (the ordinary demixing problem) to recover all $\{\bx_i\}_{i=1}^s$, and vice versa. We note  that there is a scaling ambiguity in all blind deconvolution problems  that cannot be resolved by any reconstruction method without further information. Therefore, when we talk about exact recovery in the following, then this is understood modulo such a trivial scaling ambiguity.

\vskip0.25cm
Before proceeding to our proposed algorithm we introduce some notation to facilitate a more convenient presentation of our approach. 
Let $\bb_l$ be the $l$-th column of $\BB^*$ and $\ba_{il}$ be the $l$-th column of $\BA_i^*$.  Based on our assumptions the following properties hold:
\begin{equation*}
\sum_{l=1}^L\bb_l\bb_l^* = \I_K, \quad \|\bb_l\|^2 = \frac{K}{L}, \quad \ba_{il}\sim \mathcal{C}\mathcal{N}(\bzero, \I_N).
\end{equation*}
Moreover, inspired by the well-known~\emph{lifting} idea~\cite{CSV11,RR12,CESV11,LS15},  we define the useful matrix-valued linear operator 
$\A_i : \CC^{K\times N} \to \CC^L$ and its adjoint $\A_i^*:\CC^L\rightarrow \CC^{K\times N}$  by 
\begin{equation}\label{def:Ai}
\A_i(\BZ) := \{\bb_l^*\BZ\ba_{il}\}_{l=1}^L, \quad \A^*_i(\bz) := \sum_{l=1}^L z_l \bb_l\ba_{il}^* = \BB^*\diag(\bz)\BA_i
\end{equation}
for each $1\leq i\leq s$ under canonical inner product over $\CC^{K\times N}.$ Therefore,~\eqref{eq:measure1} can be written in the following equivalent form
\begin{equation}\label{eq:measure2} 
\by = \sum_{i=1}^s \A_i(\bh_i\bx_i^*) + \be.
\end{equation}
Hence, we can think of $\by$ as the observation vector obtained from taking {\em linear} measurements with respect to a set of rank-1 matrices $\{\bh_i\bx_i^*\}_{i=1}^s.$ In fact, with a bit  of linear algebra (and ignoring the noise term for the moment), the $l$-th entry of $\by$ in~\eqref{eq:measure2} equals the inner product of two block-diagonal matrices:
\begin{equation}\label{eq:measure3}
y_l = \left\lag \underbrace{
\begin{bmatrix}
\bh_{1,0}\bx_{1,0}^*  & \bzero & \cdots & \bzero \\
\bzero  & \bh_{2,0}\bx_{2,0}^*   & \cdots & \bzero \\
\vdots & \vdots &  \ddots & \vdots \\
\bzero & \bzero &  \cdots & \bh_{s0}\bx_{s0}^*    \\
\end{bmatrix}}_{\text{defined as }\BX_0 },
\begin{bmatrix}
\bb_{l}\ba_{1l}^*  & \bzero & \cdots & \bzero \\
\bzero  & \bb_{l}\ba_{2l}^*   & \cdots & \bzero \\
\vdots & \vdots &  \ddots & \vdots \\
\bzero & \bzero &  \cdots & \bb_{l}\ba_{sl}^*    \\
\end{bmatrix} 
\right\rag + e_l,
\end{equation}
where $y_l = \sum_{i=1}^s \bb_l^*\bh_{i0}\bx_{i0}^*\ba_{il} + e_l, 1\leq l\leq L$ and $\BX_0$ is defined as the ground truth matrix.
In other words, we aim to recover such a block-diagonal matrix $\BX_0$
from $L$ linear measurements 
with block structure if $\be = \bzero.$ 

By stacking all $\{\bh_{i}\}_{i=1}^s$ (and $\{\bx_i\}_{i=1}^s, \{\bh_{i0}\}_{i=1}^s,\{\bx_{i0}\}_{i=1}^s$) into a long column, we let
\begin{equation}\label{def:hx}
\bh := 
\begin{bmatrix}
\bh_1 \\
\vdots\\
\bh_s
\end{bmatrix}, \quad
\bh_0 := 
\begin{bmatrix}
\bh_{1,0} \\
\vdots\\
\bh_{s0}
\end{bmatrix}\in\CC^{Ks}
,\quad
\bx := 
\begin{bmatrix}
\bx_1 \\
\vdots\\
\bx_s
\end{bmatrix},\quad
\bx_0 := 
\begin{bmatrix}
\bx_{1,0} \\
\vdots\\
\bx_{s0}
\end{bmatrix}
\in\CC^{Ns}.
\end{equation}
We define $\CH$ as a bilinear operator which maps a pair $(\bh, \bx)\in\CC^{Ks}\times \CC^{Ns}$ into a block diagonal matrix in $\CC^{Ks\times Ns}$, i.e., 
\begin{equation}\label{def:H}
\CH(\bh, \bx) := 
\begin{bmatrix}
\bh_1\bx_1^*  & \bzero & \cdots & \bzero \\
\bzero  & \bh_2\bx_2^*   & \cdots & \bzero \\
\vdots & \vdots &  \ddots & \vdots \\
\bzero & \bzero &  \cdots & \bh_s\bx_s^*    \\
\end{bmatrix}\in\CC^{Ks\times Ns}.
\end{equation}
Let $\BX := \CH(\bh, \bx)$ and $\BX_0 := \CH(\bh_0, \bx_0)$ where $\BX_0$ is the ground truth as illustrated in~\eqref{eq:measure3}. Define $\A(\BZ):\CC^{Ks\times Ns}\rightarrow \CC^L$ as
\begin{equation}\label{def:A}
\A(\BZ) := \sum_{i=1}^s \A_i(\BZ_i),
\end{equation}
where $\BZ = \text{blkdiag}(\BZ_1,\cdots,\BZ_s)$ and $\text{blkdiag}$ is the standard MATLAB function to construct block diagonal matrix. Therefore, 
$\A(\CH(\bh, \bx)) = \sum_{i=1}^s \A_i(\bh_i\bx_i^*)$ and $\by = \A(\CH(\bh_0, \bx_0)) + \be.$
The adjoint operator $\A^*$ is defined naturally as
\begin{equation}\label{def:A-adj}
\A^*(\bz) : = 
\begin{bmatrix}
\A_1^*(\bz)  & \bzero & \cdots & \bzero \\
\bzero  & \A_2^*(\bz)  & \cdots & \bzero \\
\vdots & \vdots &  \ddots & \vdots \\
\bzero & \bzero &  \cdots & \A_s^*(\bz)   \\
\end{bmatrix}\in\CC^{Ks\times Ns},
\end{equation}
which is a linear map from $\CC^L$ to $\CC^{Ks\times Ns}.$
To measure the approximation error of $\BX_0$ given by $\BX$, we define $\delta(\bh,\bx)$ as the global relative error: 
\begin{equation}\label{def:delta}
\delta(\bh,\bx) := \frac{\|\BX - \BX_0\|_F}{\|\BX_0\|_F} = \frac{\sqrt{\sum_{i=1}^s \|\bh_i\bx_i^* - \bh_{i0}\bx_{i0}^*\|_F^2}}{d_0} = \sqrt{\frac{\sum_{i=1}^s \delta_i^2 d_{i0}^2}{ \sum_{i=1}^s d_{i0}^2}},
\end{equation}
where $\delta_i : = \delta_i(\bh_i,\bx_i)$ is the relative error within each component:
\begin{equation*}
\delta_i(\bh_i,\bx_i) := \frac{\|\bh_i\bx_i^* - \bh_{i0}\bx_{i0}^*\|_F}{d_{i0}}.
\end{equation*}
Note that $\delta$ and $\delta_i$ are functions of $(\bh,\bx)$ and $(\bh_i,\bx_i)$ respectively and in most cases, we just simply use $\delta$ and $\delta_i$ if no possibility of confusion exists.

\subsection{Convex versus nonconvex approaches}
As indicated in~\eqref{eq:measure3}, joint blind deconvolution-demixing can be recast as the task to recover a rank-$s$ block-diagonal matrix from linear measurements. In general, such a low-rank matrix recovery problem is NP-hard. In order to take advantage of the low-rank property of the ground truth, it is natural to adopt convex relaxation by solving a convenient nuclear norm minimization program, i.e., 
\begin{equation}\label{eq:convex}
\min \sum_{i=1}^s \|\BZ_i\|_*, \quad s.t. \quad\sum_{i=1}^s \A_i(\BZ_i) = \by.
\end{equation}

The question of when the solution of~\eqref{eq:convex} yields exact recovery is first answered in our previous work~\cite{LS17b}. Late,~\cite{SJK16,JungKS17} have improved this result to the near-optimal bound $L\geq C_0s(K + N)$ up to some $\log$-factors where the main theoretical result  is informally summarized in the following theorem.
\begin{theorem}[\bf Theorem 1.1 in~\cite{JungKS17}]\label{thm:convex}
Suppose that $\BA_i$ are $L\times N$ i.i.d. complex Gaussian matrices and $\BB$ is an $L\times K$ partial DFT matrix with $\BB^*\BB = \I_K$. Then solving~\eqref{eq:convex} gives exact recovery if the number of measurements $L$ yields
\begin{equation*}
L \geq C_{\gamma} s(K+N)\log^3L 
\end{equation*}
with probability at least $1 - L^{-\gamma}$ where $C_{\gamma}$ is an absolute scalar only depending on $\gamma$  linearly.
\end{theorem}


While the SDP relaxation is definitely effective and has theoretic performance guarantees, the computational costs for solving an SDP already become too expensive for moderate size problems, let alone for large scale problems.
Therefore, we try to look for a more efficient nonconvex approach such as gradient descent, which hopefully is also reinforced by theory.
It seems quite natural  to achieve the goal by minimizing the following~\emph{nonlinear} least squares objective function with respect to $(\bh, \bx)$
\begin{equation}\label{def:F}
 F(\bh, \bx) : = \|\A(\CH(\bh,\bx)) - \by\|^2 = \left\|\sum_{i=1}^s\A_i(\bh_i\bx_i^*) - \by\right\|^2.
\end{equation}
In particular, if $\be = \bzero,$ we write
\begin{equation}\label{def:F0}
F_0(\bh, \bx) :  = \left\|\sum_{i=1}^s\A_i(\bh_i\bx_i^* - \bh_{i0}\bx_{i0}^*)\right\|^2.
\end{equation}
As also pointed out in~\cite{LLSW16}, this is a highly nonconvex optimization problem. Many of the commonly used algorithms, such as gradient descent or alternating minimization, may not necessarily yield convergence to the global minimum, so that we cannot always hope to obtain the desired solution. Often, those simple algorithms might get  stuck in local minima.

\subsection{The basin of attraction}
Motivated by several excellent recent papers of nonconvex optimization on various signal processing and machine learning problem, we propose our two-step algorithm: (i)~Compute an initial guess carefully; (ii)~Apply gradient descent to the objective function, starting with the carefully chosen initial guess. One difficulty of understanding nonconvex optimization consists in how to construct the so-called~\emph{basin of attraction}, i.e., if the starting point is inside this basin of attraction, the iterates will always stay inside the region and converge to the global minimum. The construction of the basin of attraction varies for different problems~\cite{CLS14,CLM16,SL16}. For this problem, similar to~\cite{LLSW16}, the construction follows from the following three observations. Each of these observations suggests the definition of a certain {\em neighborhood} and the basin of attraction is then defined as the intersection of these three neighborhood sets $\Kint.$

\begin{enumerate}[1.]
\item {\bf Ambiguity of solution}: in fact, we can only recover $(\bh_i,\bx_i)$ up to a scalar since $(\alpha\bh_i,\alpha^{-1}\bx_i)$ and $(\bh_i,\bx_i)$ are both solutions for $\alpha\neq 0$.
From a numerical perspective, we want to avoid the scenario when $\|\bh_i\|\rightarrow 0$ and $\|\bx_i\|\rightarrow\infty$ while $\|\bh_i\|\|\bx_i\|$ is fixed, which potentially leads to numerical instability. To balance both the norm of $\|\bh_i\|$ and $\|\bx_i\|$ for all $1\leq i\leq s$, we define
\begin{equation*}\label{def:Kd}
\Kd := \{ \{(\bh_i, \bx_i)\}_{i=1}^s: \| \bh_i\| \leq 2\sqrt{d_{i0}},  \| \bx_i\| \leq 2\sqrt{d_{i0}}, 1\leq i\leq s \},
\end{equation*}
which is a convex set. 

\item {\bf Incoherence}: the performance depends on how large/small the incoherence $\mu^2_h$ is, where $\mu_h^2$ is defined by
\begin{equation*}
\mu^2_h : = \max_{1\leq i\leq s} \frac{L\|\BB\bh_{i0}\|^2_{\infty}}{\|\bh_{i0}\|^2}.
\end{equation*}
The idea is that:~\emph{the smaller the $\mu^2_h$ is, the better the performance is.} Let us consider an extreme case: if $\BB\bh_{i0}$ is highly sparse or spiky, we lose much information on those zero/small entries and cannot hope to get satisfactory recovered signals. In other words, we need the ground truth $\bh_{i0}$ has ``spectral flatness" and $\bh_{i0}$ is not highly localized on the Fourier domain.

A similar quantity is also introduced in the matrix completion problem~\cite{CR08,SL16}. The larger $\mu^2_h$ is, the more $\bh_{i0}$ is aligned with one particular row of $\BB.$
To control the incoherence between $\bb_{l}$ and $\bh_i$, we define the second neighborhood, 
\begin{equation}\label{def:Kmu}
\Kmu :=  \{ \{\bh_i\}_{i=1}^s : \sqrt{L} \|\BB\bh_i\|_{\infty} \leq 4\sqrt{d_{i0}}\mu,  1\leq i\leq s\}, 
\end{equation}
where $\mu$ is a parameter and $\mu \geq \mu_h$. Note that $\Kmu$ is also a convex set. 

\item {\bf Close to the ground truth}: we also want to construct an initial guess such that it is close to the ground truth, i.e., 
\begin{equation}\label{def:Keps}
\Keps  :=  \left\{ \{(\bh_i, \bx_i)\}_{i=1}^s: \delta_i = \frac{\|\bh_i\bx_i^* - \bh_{i0}\bx_{i0}^*\|_F}{d_{i0}} \leq \eps, 1\leq i\leq s   \right\}
\end{equation}
where $\eps$ is a predetermined parameter in $(0, \frac{1}{15}]$.
\end{enumerate}
\begin{remark}
To ensure $\delta_{i} \leq \eps$, it suffices to ensure $\delta \leq \frac{\eps}{\sqrt{s}\kappa}$ where $\kappa := \frac{\max d_{i0}}{\min d_{i0}} \geq 1$. This is because
\begin{equation*}
\frac{1}{s\kappa^2}\sum_{i=1}^s \delta_i^2 \leq \delta^2 \leq  \frac{\eps^2}{s\kappa^2} 
\end{equation*}
which implies $\max_{1\leq i\leq s}\delta_i \leq \eps.$
\end{remark}
\begin{remark}
When we say $(\bh, \bx)\in\mathcal{N}_d, \Kmu$ or $\Keps$, it means for all $i=1,\dots,s$ we have $(\bh_i,\bx_i) \in\Kd$, $\Kmu$ or $\Keps$ respectively. In particular, $(\bh_0, \bx_0) \in\Kint$ where $\bh_0$ and $\bx_0$ are defined in~\eqref{def:hx}.
\end{remark}

\subsection{Objective function and Wirtinger derivative}
To implement the first two observations, we introduce the regularizer $G(\bh, \bx)$, defined as the sum of $s$ components
\begin{eqnarray}\label{def:G}
G(\bh, \bx):= \sum_{i=1}^s G_i(\bh_i,\bx_i) .
\end{eqnarray}
For each component $G_i(\bh_i,\bx_i)$, we let $\rho \geq d^2 + 2\|\be\|^2$, $0.9 d_0 \leq d \leq 1.1d_0$, $0.9d_{i0} \leq d_i \leq 1.1d_{i0}$ for all $1\leq i\leq s$ and
\begin{align}
\label{def:Gi}
G_i  :=  \rho \Big[ \underbrace{G_0\left(\frac{\|\bh_i\|^2}{2d_i}\right) + G_0\left(\frac{\|\bx_i\|^2}{2d_i}\right)}_{ \Kd } + \underbrace{\sum_{l=1}^LG_0\left(\frac{L |\bb_l^*\bh_i|^2}{8d_i\mu^2 }\right)}_{\Kmu} \Big],
\end{align}
where $G_0(z) = \max\{z-1, 0\}^2$.
Here both $d$ and $\{d_i\}_{i=1}^s$ are data-driven and well approximated by our spectral initialization procedure; and $\mu^2$ is a tuning parameter which could be estimated if we assume a specific statistical model for the channel (for example, in the widely used Rayleigh fading model, the channel coefficients are assumed to be complex Gaussian).
The idea behind $G_i$ is quite straightforward though the formulation is complicated. For each  $G_i$ in~\eqref{def:Gi}, the first two terms try to force the iterates to lie in $\Kd$ and the third term tries to encourage the iterates to lie in $\Kmu.$ What about the neighborhood $\Keps$? A proper choice of the initialization followed by gradient descent which keeps the objective function decreasing will ensure that the iterates stay in $\Keps$.

\vskip0.25cm
Finally, we consider the objective function as the sum of nonlinear least squares objective function $F(\bh,\bx)$ in~\eqref{def:F} and the regularizer $G(\bh,\bx)$, 
\begin{equation}\label{def:FG}
\tF(\bh, \bx) := F(\bh,\bx) + G(\bh, \bx).
\end{equation}

Note that the input of the function $\tF(\bh,\bx)$ consists of complex variables but the output is real-valued. As a result, the following simple relations hold
\begin{equation*}
\frac{\pa \tF}{\pa \bar{\bh}_i} = \overline{\frac{\pa \tF}{\pa \bh_i} }, \quad \frac{\pa \tF}{\pa \bar{\bx}_i} = \overline{\frac{\pa \tF}{\pa \bx_i} }.
\end{equation*}
Similar properties also apply to both $F(\bh,\bx)$ and $G(\bh,\bx)$.

Therefore, to minimize this function, it suffices to consider only the gradient of $\tF$ with respect to $\bar{\bh}_i$ and $\bar{\bx}_i$, which is also called Wirtinger derivative~\cite{CLS14}.
The Wirtinger derivatives of $F(\bh,\bx)$ and $G(\bh,\bx)$ w.r.t. $\bar{\bh}_i$ and $\bar{\bx}_i$ can be easily computed as follows 
\begin{align}
\nabla F_{\bh_i} & = \A_i^*\left(\A(\BX)- \by\right)\bx_i = \A_i^*\left(\A(\BX-\BX_0)- \be \right)\bx_i, \label{eq:WFh} \\
\nabla F_{\bx_i} & = \left(\A_i^*\left(\A(\BX) - \by\right)\right)^*\bh_i = \left(\A_i^*\left(\A(\BX-\BX_0) - \be\right)\right)^*\bh_i, \label{eq:WFx}\\
\nabla G_{\bh_i}
& = \frac{\rho}{2d_i}\Big[G'_0\left(\frac{\|\bh_i\|^2}{2d_i}\right) \bh_i 
+ \frac{L}{4\mu^2} \sum_{l=1}^L G'_0\left(\frac{L|\bb_l^*\bh_i|^2}{8d_i\mu^2}\right) \bb_l\bb_l^*\bh_i \Big], \label{eq:WGh} \\ 
\nabla G_{\bx_i} & = \frac{\rho}{2d_i} G'_0\left( \frac{\|\bx_i\|^2}{2d_i}\right) \bx_i, \label{eq:WGx}
\end{align}
where  $\A(\BX)  = \sum_{i=1}^s \A_i(\bh_i\bx_i^*)$ and $\A^*$ is defined in~\eqref{def:A-adj}.
In short, we denote
\begin{equation}\label{def:grad}
\nabla\tF_{\bh} : = \nabla F_{\bh} + \nabla G_{\bh}, \quad
\nabla F_{\bh} : = 
\begin{bmatrix}
\nabla F_{\bh_1} \\
\vdots \\
\nabla F_{\bh_s}
\end{bmatrix},
\quad 
\nabla G_{\bh} : = 
\begin{bmatrix}
\nabla G_{\bh_1} \\
\vdots \\
\nabla G_{\bh_s}
\end{bmatrix}.
\end{equation}
Similar definitions hold for $\nabla\tF_{\bx},\nabla F_{\bx}$ and $G_{\bx}$. It is easy to see that $\nabla F_{\bh} = \A^*(\A(\BX) - \by)\bx$ and $\nabla F_{\bx} = (\A^*(\A(\BX) - \by))^*\bh$.

\section{Algorithm and Theory}
\label{s:thm}
\subsection{Two-step algorithm}
As mentioned before, the first step is to find a good initial guess $(\bu^{(0)}, \bv^{(0)})\in\CC^{Ks}\times\CC^{Ns}$ such that it is inside the basin of attraction. 
The initialization follows from this key fact: 
\begin{equation*}
\E(\A_i^*(\by)) = \E\left(\A_i^*\left(\sum_{j=1}^s\A_j(\bh_{j0}\bx_{j0}^* )+\be\right)\right) = \bh_{i0}\bx_{i0}^*,
\end{equation*}
where we use $\BB^*\BB = \sum_{l=1}^L\bb_l\bb_l^* = \I_K$, $\E(\ba_{il}\ba_{il}^*) = \I_N$ and
\begin{align*}
\E(\A_i^*\A_i(\bh_{i0}\bx_{i0}^*)) & = \sum_{l=1}^L \bb_l\bb_l^*\bh_{i0}\bx_{i0}^* \E(\ba_{il}\ba_{il}^*)
= \bh_{i0}\bx_{i0}^*, \\
\E(\A_j^*\A_i(\bh_{i0}\bx_{i0}^*)) & 
= \sum_{l=1}^L \bb_l\bb_l^*\bh_{i0}\bx_{i0}^* \E(\ba_{il}\ba_{jl}^*)
= \bzero,\quad \forall j\neq i. 
\end{align*}
Therefore, it is natural to extract the leading singular value and associated left and right singular vectors from each $\A_i^*(\by)$ and use them as (a hopefully good) approximation to $(d_{i0}, \bh_{i0}, \bx_{i0}).$ This idea leads to Algorithm~\ref{Initial}, the theoretic guarantees of which are given in Section~\ref{s:init}. The second step of the algorithm is just to apply gradient descent to $\tF$ with the initial guess $\{(\bu^{(0)}_{i}, \bv^{(0)}_i, d_i)\}_{i=1}^s$ or $(\bu^{(0)}, \bv^{(0)},\{ d_i\}_{i=1}^s)$, where $\bu^{(0)}$ stems from stacking all $\bu^{(0)}_{i}$ into one long vector\footnote{It is clear that instead of gradient descent one could also use a second-order method to achieve faster convergence at the tradeoff of increased computational cost per iteration. The theoretical convergence analysis for a second-order method will require a very different approach from the one developed in this paper.}.
\begin{algorithm}[h!]
\caption{Initialization via spectral method and projection}
\label{Initial}
\begin{algorithmic}[1]
\For{ $i = 1, 2, \dots, s$}
\State Compute $\A_i^*(\by).$
\State Find the leading singular value, left and right singular vectors of $\A_i^*(\by)$, denoted by $(d_i, \hat{\bh}_{i0}, \hat{\bx}_{i0})$.
\State Solve the following optimization problem for $1\leq i\leq s$:
\begin{equation*}
\bu^{(0)}_{i} := \text{argmin}_{\bz\in\CC^K} \|\bz - \sqrt{d_{i}}\hbh_{i0}\|^2 \text{ s.t. }\sqrt{L}\|\BB\bz\|_{\infty} \leq 2\sqrt{d_{i}}\mu.
\end{equation*}
\State Set $\bv^{(0)}_i = \sqrt{d_i}\hat{\bx}_{i0}$. 
\EndFor
\State Output: $\{(\bu^{(0)}_{i}, \bv^{(0)}_i, d_i)\}_{i=1}^s$ or $(\bu^{(0)}, \bv^{(0)}, \{d_i\}_{i=1}^s)$.
\end{algorithmic}
\end{algorithm}
\begin{algorithm}[h!]
\caption{Wirtinger gradient descent with constant stepsize $\eta$}
\label{AGD}
\begin{algorithmic}[1]
\State {\bf Initialization:} obtain $(\bu^{(0)}, \bv^{(0)}, \{d_i\}_{i=1}^s)$ via Algorithm~\ref{Initial}.
\For{ $t = 1, 2, \dots, $}
\For{ $i = 1, 2, \dots, s$}
\State $\bu^{(t)}_{i} = \bu^{(t-1)}_{i} - \eta \nabla \tF_{\bh_i}(\bu^{(t-1)}_{i}, \bv^{(t-1)}_{i})$, 
\State $\bv^{(t)}_{i} = \bv^{(t-1)}_{i} - \eta \nabla \tF_{\bx_i}(\bu^{(t-1)}_{i}, \bv^{(t-1)}_{i})$,
\EndFor
\EndFor
\end{algorithmic}
\end{algorithm}
\begin{remark}
For Algorithm~\ref{AGD}, we can rewrite each iteration into
\begin{equation*}
\bu^{(t)} = \bu^{(t-1)} - \eta\nabla \tF_{\bh}(\bu^{(t-1)}, \bv^{(t-1)}), \quad\bv^{(t)} = \bv^{(t-1)} - \eta\nabla \tF_{\bx}(\bu^{(t-1)}, \bv^{(t-1)}), 
\end{equation*}
where $\nabla\tF_{\bh}$ and $\nabla\tF_{\bx}$ are in~\eqref{def:grad}, and 
\begin{equation*}
\bu^{(t)} : = 
\begin{bmatrix}
\bu_1^{(t)} \\
\vdots \\
\bu_s^{(t)}
\end{bmatrix}, \quad
\bv^{(t)} : = 
\begin{bmatrix}
\bv_1^{(t)} \\
\vdots \\
\bv_s^{(t)}
\end{bmatrix}.
\end{equation*}

\end{remark}

\subsection{Main results}
Our main findings are summarized as follows: Theorem~\ref{thm:init} shows that the initial guess given by Algorithm~\ref{Initial} indeed belongs to the basin of attraction. Moreover, $d_i$ also serves as a good approximation of $d_{i0}$ for each $i$. Theorem~\ref{thm:main} demonstrates that the regularized Wirtinger gradient descent will guarantee the linear convergence of the iterates and the recovery is exact in the noisefree case and stable in the presence of noise. 
\begin{theorem}\label{thm:init}
The initialization obtained via Algorithm~\ref{Initial} satisfies
\begin{align}
 (\bu^{(0)}, \bv^{(0)}) \in \frac{1}{\sqrt{3}}\Kd\bigcap \frac{1}{\sqrt{3}} \MN_{\mu}\bigcap \MN_{\frac{2\eps}{5\sqrt{s}\kappa}} \label{eq:init-val1} 
\end{align}
and  
\begin{equation}\label{eq:d}
0.9d_{i0} \leq d_i\leq 1.1d_{i0} ,\quad 0.9d_0 \leq d\leq 1.1d_0, 
\end{equation}
holds with probability at least
$1 - L^{-\gamma+1}$ if the number of measurements satisfies
\begin{equation}
\label{Lbound}
L \geq C_{\gamma+\log(s)}(\mu_h^2 + \sigma^2)s^2 \kappa^4  \max\{K,N\}\log^2 L/\eps^2.
\end{equation}
Here $\eps$ is any predetermined constant in $(0, \frac{1}{15}]$, and $C_{\gamma}$ is a constant only linearly depending on $\gamma$ with $\gamma \geq 1$. 

\end{theorem}

\begin{theorem}\label{thm:main}
Starting with the initial value $\bz^{(0)}:= (\bu^{(0)}, \bv^{(0)})$ satisfying~\eqref{eq:init-val1},
the Algorithm~\ref{AGD} creates a sequence of iterates $(\bu^{(t)}, \bv^{(t)})$ which converges to the global minimum linearly, 
\begin{align}\label{eq:errorbound}
\|\CH(\bu^{(t)}, \bv^{(t)}) - \CH(\bh_0, \bx_0) \|_F 
  & \leq \frac{\eps d_0}{\sqrt{2s\kappa^2}}(1 - \eta\omega)^{t/2} + 60\sqrt{s} \|\A^*(\be)\|
\end{align}
with probability at least $1 - L^{-\gamma+1}$ where $\eta\omega = \mathcal{O}((s\kappa d_0(K+N)\log^2L)^{-1})$ and 
\begin{equation*}
 \|\A^*(\be)\| \leq C_0 \sigma  d_0\sqrt{\frac{\gamma s(K + N)(\log^2L)}{L}}
\end{equation*}
 if the number of measurements $L$ satisfies
\begin{equation}
\label{boundforL}
L \geq C_{\gamma+\log (s)}(\mu^2 + \sigma^2)s^2 \kappa^4  \max\{K,N\}\log^2 L/\eps^2.
\end{equation}

\end{theorem}

\begin{remark}
Our previous work~\cite{LS17b} shows that  the convex approach via semidefinite programming (see~\eqref{eq:convex}) requires $L \geq C_0s^2(K + \mu^2_h N)\log^3(L)$ to ensure exact recovery. Later,~\cite{JungKS17} improves this result to the near-optimal bound $L\geq C_0s(K + \mu^2_h N)$ up to some $\log$-factors.
The difference between nonconvex and convex methods lies in the appearance of the condition number $\kappa$ in~\eqref{boundforL}. This is not just an artifact of the proof---empirically we also observe that the value of $\kappa$ affects the convergence rate of our nonconvex algorithm, see Figure~\ref{fig:snr-kappa}.
\end{remark}

\begin{remark}
Our theory suggests $s^2$-dependence for the number of measurements $L$, although numerically $L$ in fact depends on $s$ linearly, as shown in Section~\ref{s:numerics}. The reason for $s^2$-dependence will be addressed in details in Section~\ref{s:outline}.
\end{remark}

\begin{remark}
In the theoretical analysis, we assume that $\BA_i$ (or equivalently $\BC_i$) is a Gaussian random matrix. Numerical simulations suggest that this assumption is clearly not necessary. For example, $\BC_i$ may be chosen to be a Hadamard-type matrix which is more appropriate and favorable for communications. 
\end{remark}
\begin{remark}
If $\be = \bzero,$~\eqref{eq:errorbound} shows that $(\bu^{(t)}, \bv^{(t)})$  converges to the ground truth at a linear rate. On the other hand, if  noise exists, $(\bu^{(t)}, \bv^{(t)})$ is guaranteed to converge to a point within a small neighborhood of $(\bh_0,\bx_0).$ More importantly, if the number of measurements $L$ gets larger, $\|\A^*(\be)\|$ decays at the rate of $\mathcal{O}(L^{-1/2})$.  
\end{remark}

\section{Numerical simulations}\label{s:numerics}

In this section we present a range of numerical simulations to illustrate and complement different aspects of our theoretical framework. We will empirically
analyze the number of measurements needed for perfect joint deconvolution/demixing to see how this compares to our theoretical bounds. We will also study the robustness for noisy data. In our simulations we use Gaussian encoding matrices, as in our theorems. But we also try more realistic structured
encoding matrices, that are more reminiscent of what one might come across in wireless communications.

While Theorem~\ref{thm:main} says that the number of measurements $L$ depends {\em quadratically} on the number of sources $s$, numerical simulations suggest near-optimal performance.
Figure~\ref{fig:L-vs-s-gaussian} demonstrates that $L$ actually depends linearly on $s$, i.e., the boundary between success (white) and failure (black) is approximately a linear function of $s$. In the experiment,  $K = N = 50$ are fixed, all $\BA_i$ are complex Gaussians and all $(\bh_i,\bx_i)$ are standard complex Gaussian vectors.  For each pair of $(L,s)$, 25 experiments are performed and we treat the recovery as a success if $\frac{\|\hat{\BX} - \BX_0\|_F}{\|\BX_0\|_F} \leq 10^{-3}.$ For our algorithm, we use backtracking  to determine the stepsize and the iteration stops either if  $\|\A(\CH(\bh^{(t+1)}, \bx^{(t+1)}) - \CH(\bh^{(t)}, \bx^{(t)})) \| < 10^{-6}\|\by\|$ or if the number of iterations reaches 500.
The backtracking is based on the Armijo-Goldstein condition~\cite{luenberger2015linear}. The initial stepsize is chosen to be $\eta = \frac{1}{K+N}$. If $\tF(\bz^{(t)} - \eta \nabla \tF(\bz^{(t)})) > \tF(\bz^{(t)})$, we just divide $\eta$ by two and use a smaller stepsize.  

We see from Figure~\ref{fig:L-vs-s-gaussian} that the number of measurements for the proposed algorithm to succeed not only seems to depend linearly on the number of sensors,  but it is  actually rather close to the information-theoretic limit $s(K+N)$. Indeed, the green dashed line in Figure~\ref{fig:L-vs-s-gaussian}, which represents the empirical boundary for the phase transition between success and failure corresponds to a line with slope about $\frac{3}{2} s(K+N)$.  It is interesting to compare this empirical performance to the sharp theoretical phase transition bounds one would obtain via convex optimization~\cite{CRP12,mccoy2013demixing}. Considering the convex approach based on lifting in~\cite{LS17b}, we can adapt the theoretical framework in~\cite{CRP12} to the blind deconvolution/demixing setting, but with one modification. The bounds in~\cite{CRP12} rely on Gaussian widths of tangent cones related to the measurement matrices $\A_i$. Since simply analytic formulas for these expressions seem to be out of reach for the structured rank-one measurement matrices used in our paper, we instead compute the bounds for full-rank Gaussian random matrices, which yields a sharp bound of about  $3s(K+N)$ 
(the corresponding bounds for rank-one sensing matrices will likely have a constant larger than 3).  Note that these sharp theoretical bounds predict quite accurately the empirical behavior of convex methods. Thus our empirical bound for using a non-convex
methods compares rather favorably with that of the convex approach.

\begin{figure}[h!]
\centering
\includegraphics[width=180mm]{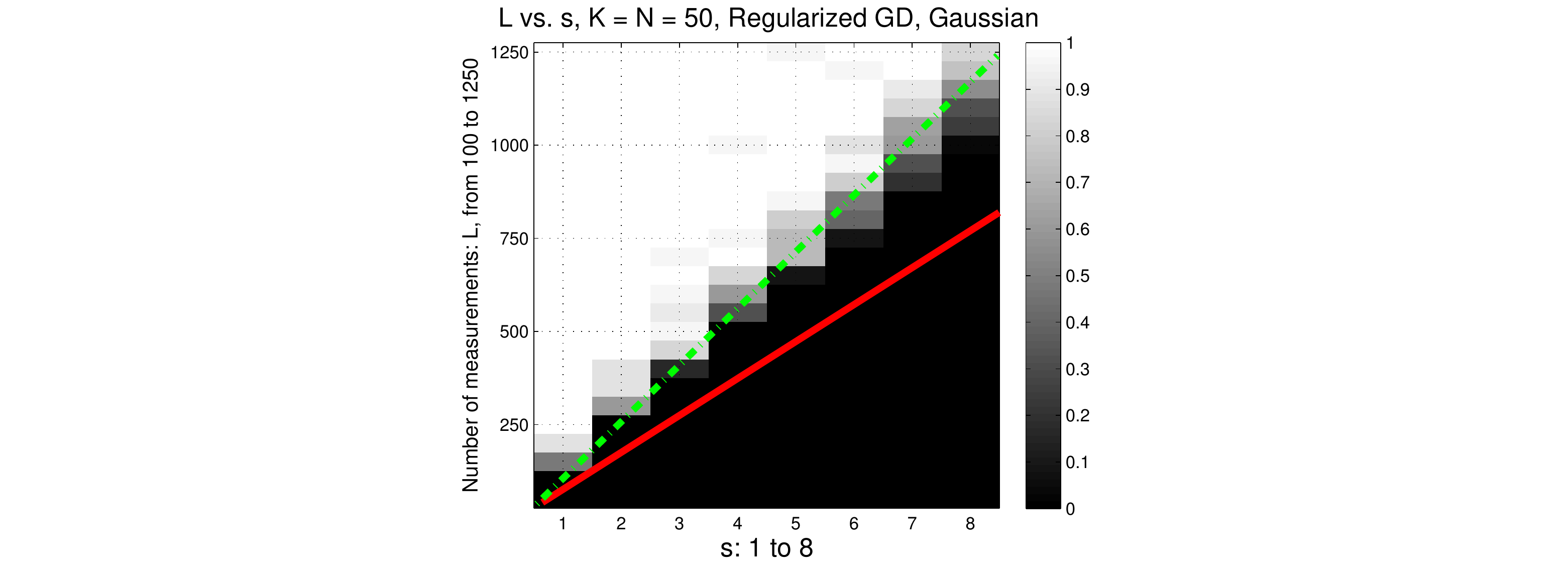}
\caption{Phase transition plot for empirical recovery performance under different choices of $(L,s)$ where $K = N =50$ are fixed. Black region: failure; white region: success. The red solid line depicts the number of degrees of freedom and the green dashed line shows the empirical phase transition bound for Algorithm~\ref{AGD}.}
\label{fig:L-vs-s-gaussian}
\end{figure}

Similar conclusions can be drawn from Figure~\ref{fig:L-vs-s-hadamard}; there all $\BA_i$ are in the form of $\BA_i = \BF \BD_i \BH$ where $\BF$ is the unitary $L\times L$ DFT matrix, all $\BD_i$ are independent diagonal binary $\pm 1$ matrices and $\BH$ is an $L\times N$ fixed partial deterministic Hadamard matrix. 
The purpose of using $\BD_i$ is to enhance the incoherence between each channel so that our algorithm is able to tell apart each individual signal and channel. As before we assume Gaussian channels, i.e., $\bh_i\sim \CN(\bzero, \I_K)$
Therefore, our approach does not only work for Gaussian encoding matrices $\BA_i$ but also for the matrices that are interesting to real-world applications, although no satisfactory theory has been derived yet for that case. Moreover, due to the structure of $\BA_i$ and $\BB$, fast transform algorithms are available, potentially allowing for  real-time deployment. 
\begin{figure}[h!]
\centering
\includegraphics[width=75mm]{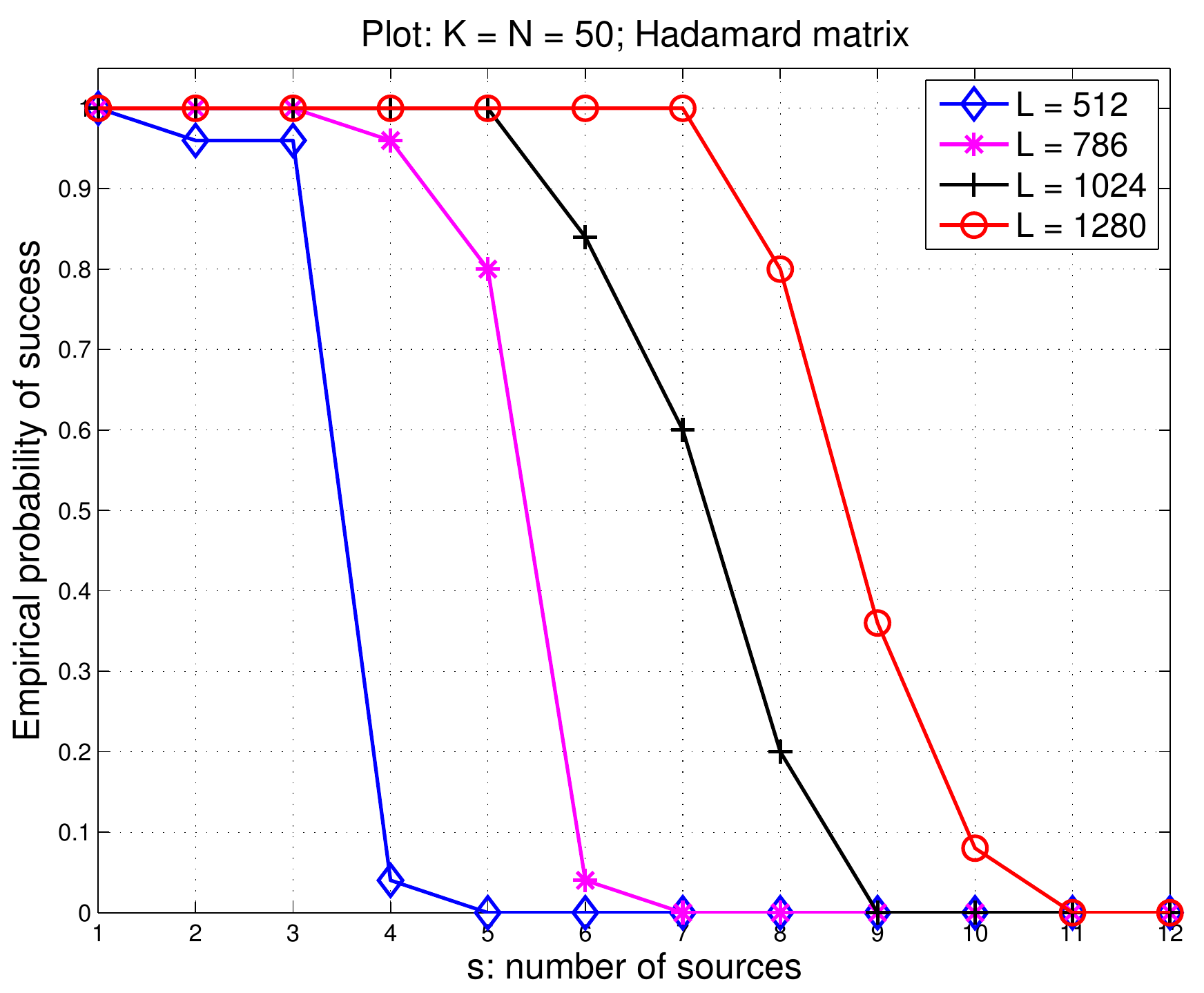}
\caption{Empirical probability of successful recovery for different pairs of $(L,s)$ when $K = N =50$ are fixed. }
\label{fig:L-vs-s-hadamard}
\end{figure}

Figure~\ref{fig:snr-gaussian} shows the robustness of our algorithm under different levels of noise. We also run 25 samples for each level of SNR and different $L$ and then compute the average relative error. It is easily seen that the relative error scales linearly with the SNR and one unit of increase in SNR (in dB) results in one unit of decrease in the relative error. 
\begin{figure}[h!]
\centering
\begin{minipage}{0.48\textwidth}
\includegraphics[width=75mm]{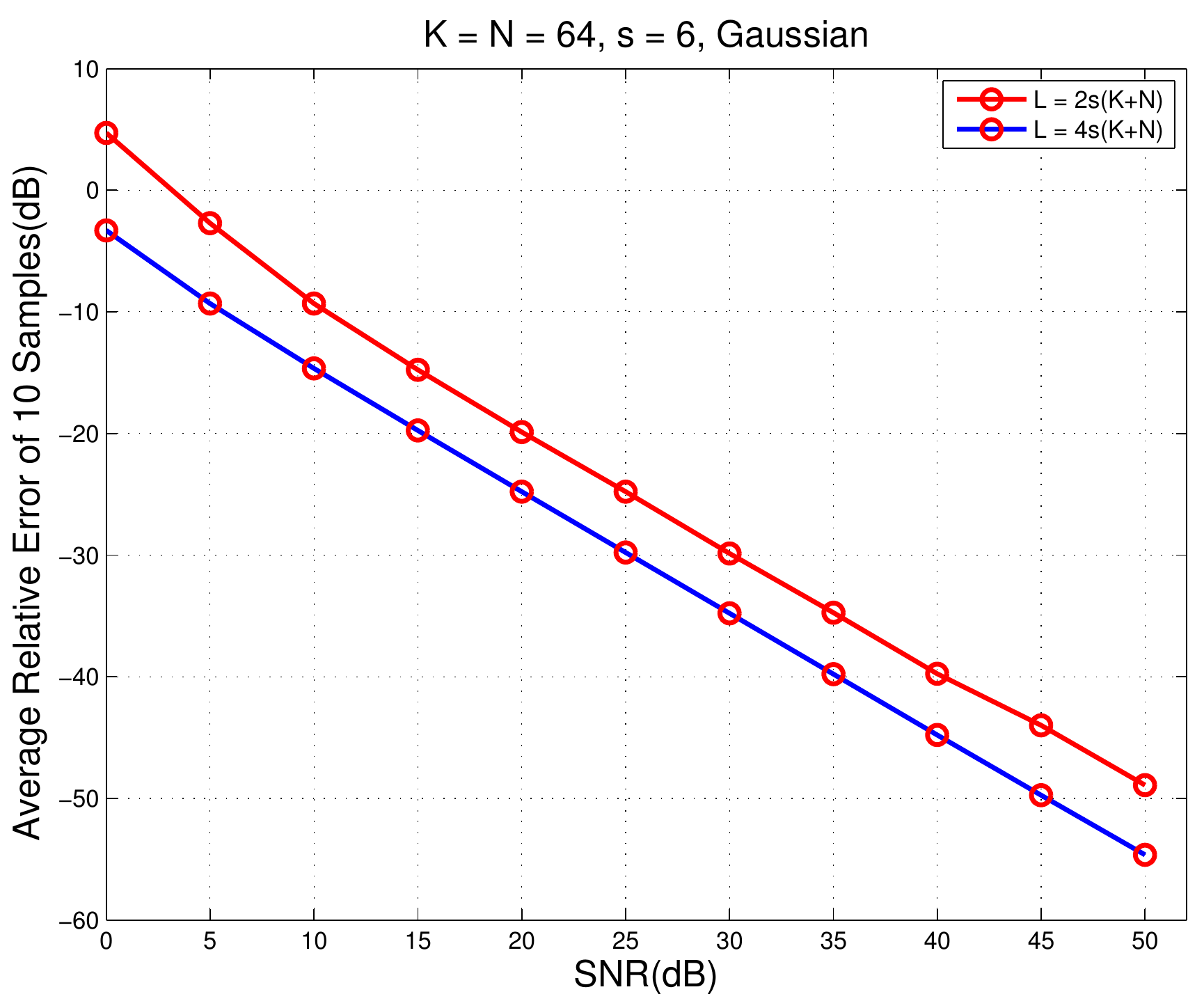}
\end{minipage}
\hfill
\begin{minipage}{0.48\textwidth}
\includegraphics[width=75mm]{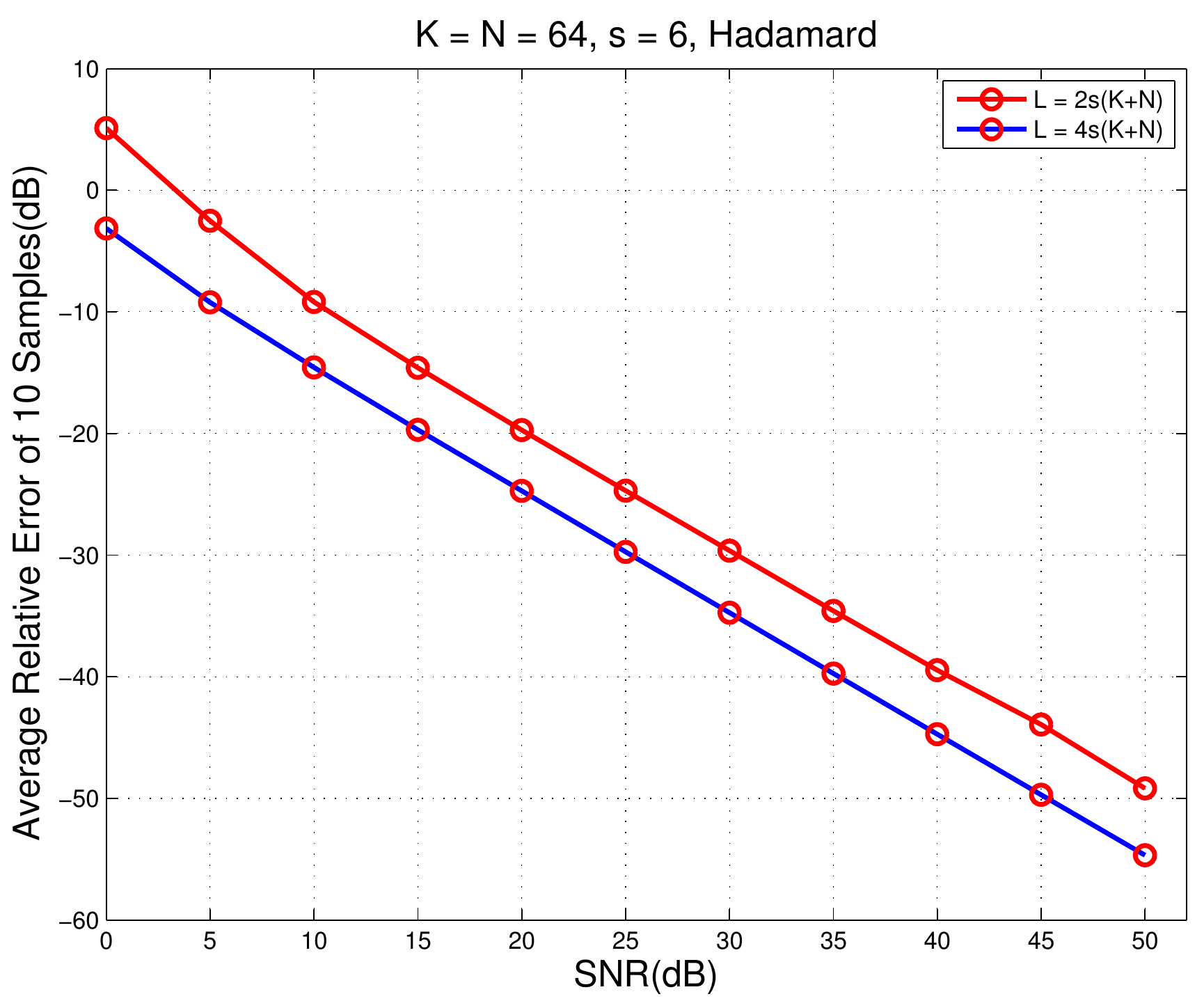}
\end{minipage}
\caption{Relative error vs. SNR (dB): SNR = $20\log_{10}\left(\frac{\|\by\|}{\|\be\|}\right)$.}
\label{fig:snr-gaussian}
\end{figure}

Theorem~\ref{thm:main} suggests that the performance and convergence rate actually depend on the condition number of $\BX_0 = \CH(\bh_0,\bx_0)$, i.e., 
on $\kappa = \frac{\max d_{i0}}{\min d_{i0}}$ where $d_{i0} = \|\bh_{i0}\|\|\bx_{i0}\|$. Next we demonstrate that this dependence on the condition number is not an artifact of the proof, but is indeed also observed empirically. In this experiment, we let $s=2$ and set for the first component $d_{1,0} = 1$ and for the second one $d_{2,0} = \kappa$ for $\kappa \in \{1,2,5\}$. Here, $\kappa = 1$ means that the received signals of both sensors have equal power, whereas $\kappa=5$ means that the signal received from the second sensor is considerably stronger. The initial stepsize is chosen as $\eta=1$, followed by the backtracking scheme. Figure~\ref{fig:snr-kappa} shows how the relative error decays with respect to the number of iterations $t$ under different condition number $\kappa$ and $L.$

The larger $\kappa$ is, the slower the convergence rate is, as we see from Figure~\ref{fig:snr-kappa}. This may result from two reasons: our spectral initialization may not be able to give a good initial guess for those weak components; moreover, during the gradient descent procedure, the gradient directions for the weak components could be totally dominated/polluted by the strong components. Currently, we still have no effective way of how to deal with this issue of slow convergence when $\kappa$ is not small. We have to leave this topic for future investigations.

\begin{figure}[h!]
\centering
\begin{minipage}{0.48\textwidth}
\includegraphics[width=75mm]{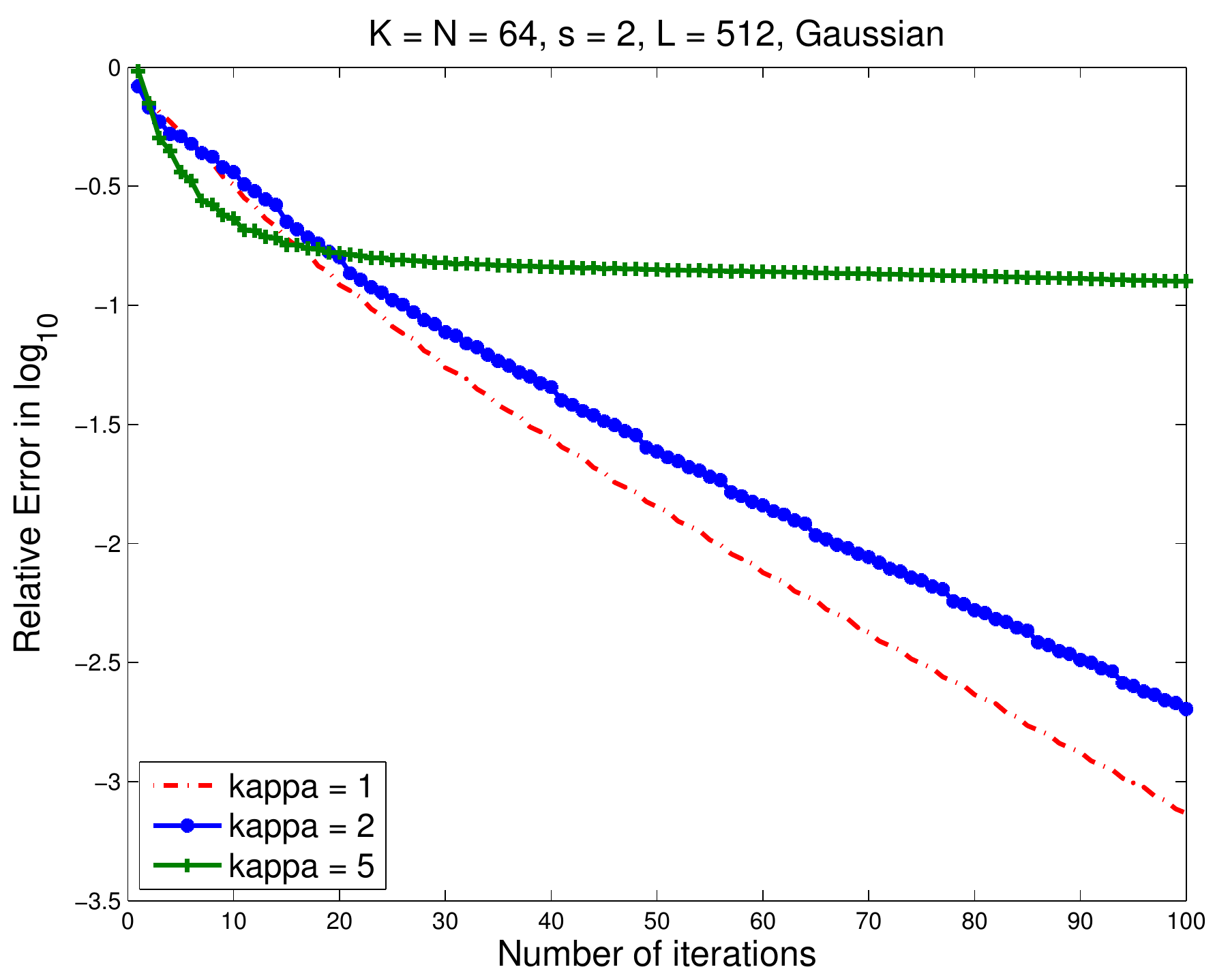}
\end{minipage}
\hfill
\begin{minipage}{0.48\textwidth}
\includegraphics[width=75mm]{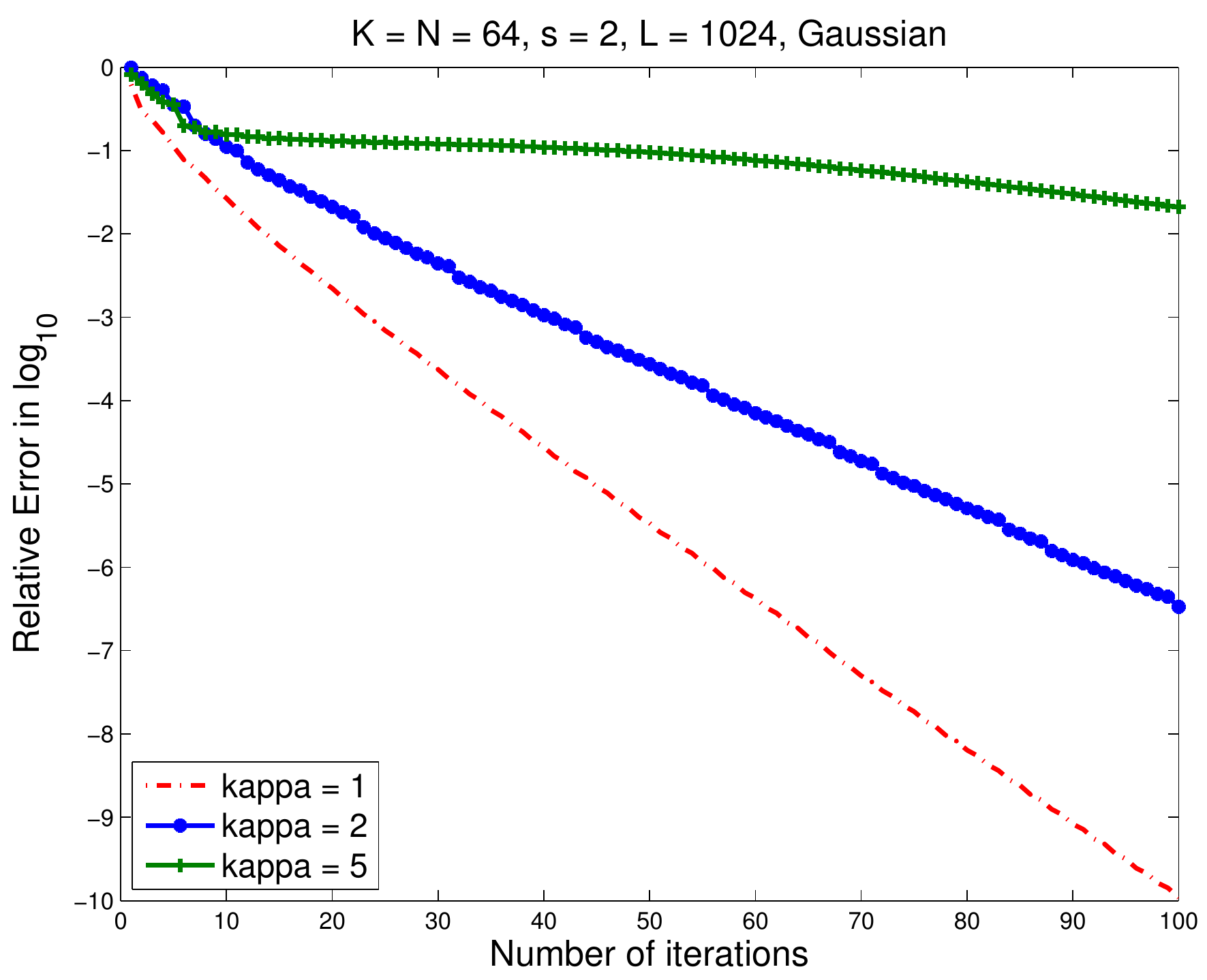}
\end{minipage}
\caption{Relative error vs. number of iterations $t$.}
\label{fig:snr-kappa}
\end{figure}

\section{Convergence analysis}
\label{s:converge}
Our convergence analysis relies on the following four conditions where the first three of them are local properties.  We will also briefly discuss how they contribute to the proof of our main theorem. Note that our previous work~\cite{LLSW16} on blind deconvolution is actually a special case $(s=1)$ of~\eqref{eq:measure1}. The proof of Theorem~\ref{thm:main} follows in part the main ideas in~\cite{LLSW16}. The readers may find the technical parts of~\cite{LLSW16} and this manuscript share many similarities.  However, there are also important differences. After all, we are now dealing with a more complicated problem where the ground truth matrix $\BX_0$ and measurement matrices are both rank-$s$ block-diagonal matrices, as shown in~\eqref{eq:measure3}, instead of rank-1 matrices in~\cite{LLSW16}. The key is to understand the properties of the linear operator $\A$ applying to different types of~\emph{block-diagonal} matrices. Therefore, many technical details are much more involved while on the other hand, some of results in~\cite{LLSW16} can be used directly. 
During the presentation, we will clearly point out both the similarities to and differences from~\cite{LLSW16}.

\subsection{Four key conditions}
\begin{condition}{\bf Local regularity condition:}\label{cond:reg}
Let $\bz : = (\bh, \bx)\in\CC^{s(K+N)}$ and $\nabla\tF(\bz) := \begin{bmatrix} \nabla\tF_{\bh}(\bz) \\ \nabla\tF_{\bx}(\bz) \end{bmatrix} \in\CC^{s(K+N)}$, then
\begin{equation}\label{def:reg}
\|\nabla \tF(\bz)\|^2 \geq \omega [\tF(\bz) - c]_+
\end{equation}
for $\bz \in\Kint$ where $\omega = \frac{d_0}{7000}$ and $c = \|\be\|^2 + 2000s\|\A^*(\be)\|^2.$
\end{condition}
We will prove Condition~\ref{cond:reg} in Section~\ref{s:LRC}. Condition~\ref{cond:reg} states that $\tF(\bz) = 0$ if $\|\nabla \tF(\bz)\| = 0$ and $\be = 0$, i.e., all the stationary points inside the basin of attraction are global minima.

\begin{condition}{\bf Local smoothness condition:}\label{cond:smooth}
Let $\bz = (\bh, \bx)$ and $\bw = (\bu, \bv)$ and there holds
\begin{equation}\label{def:CL}
\|\tF(\bz + \bw) - \tF(\bz)\| \leq C_L \|\bw\|
\end{equation}
for $\bz + \bw$ and $\bz$ inside $\Kint$ where $C_L \approx \mathcal{O}(d_0s\kappa(1 + \sigma^2)(K + N)\log^2 L )$ is the Lipschitz constant of $\tF$ over $\Kint$. The convergence rate  is governed by $C_L$.
\end{condition}
The proof of Condition~\ref{cond:smooth} can be found in Section~\ref{s:smooth}. 

\begin{condition}{\bf Local restricted isometry property:}\label{cond:rip}
Denote $\BX = \CH(\bh, \bx)$ and $\BX_0 = \CH(\bh_0, \bx_0)$. There holds
\begin{equation}\label{eq:rip}
\frac{2}{3} \|\BX - \BX_0\|_F^2
\leq \left\| \A(\BX - \BX_0) \right\|^2
\leq \frac{3}{2} \|\BX - \BX_0\|_F^2
\end{equation}
uniformly all for $(\bh, \bx)\in\Kint$.
\end{condition}
Condition~\ref{cond:rip} will be proven in Section~\ref{s:LRIP}. It says that the convergence of the objective function implies the convergence of the iterates. 
\begin{remark}[\bf Necessity of inter-user incoherence]\label{rem:inc}
Although Condition~\ref{cond:rip} is seemingly the same as the one in our previous work~\cite{LLSW16}, it is indeed very different. Recall that $\A$ is a linear operator acting on block-diagonal matrices and its output is the sum of $s$ different components involving $\A_i$. Therefore, the proof of Condition~\ref{cond:rip} heavily depends on the inter-user incoherence whereas this notion of incoherence is not needed at all for the single-user scenario. At the beginning of Section~\ref{s:model}, we discuss the choice of $\BC_i$ (or $\BA_i$). In order to distinguish one user from another, it is essential to use sufficiently different\footnote{Suppose all $\BC_i$ are the same, there is no hope to recover all pairs of $\{(\bh_i,\bx_i)\}_{i=1}^s$ simultaneously.}  encoding matrices $\BC_i$ (or $\BA_i$). Here  the independence and Gaussianity of all $\BC_i$ (or $\BA_i$) guarantee that  $\|\PP_{T_i}\A_i^*\A_j\PP_{T_j}\|$ is sufficiently small for all $i\neq j$ where $T_i$ is defined in~\eqref{def:T}. It is a key element to ensure the validity of Condition~\ref{cond:rip} which is also an important component to prove Condition~\ref{cond:reg}. On the other hand, due to the recent progress on this joint deconvolution and demixing problem, one is also able to prove a local restricted isometry property with tools such as bounding the suprema of chaos processes~\cite{JungKS17} by assuming $\{\BA_i\}_{i=1}^s$ as Gaussian matrices.
\end{remark}

\begin{condition}{\bf Robustness condition:}\label{cond:robust}
Let  $\eps \leq \frac{1}{15}$ be a predetermined constant. We have
\begin{equation}\label{eq:robust}
\|\A^*(\be)\| = \max_{1\leq i\leq s}\|\A_i^*(\be)\| \leq \frac{\eps d_0}{10\sqrt{2}s \kappa},
\end{equation}
where $\be\sim \mathcal{C}\mathcal{N}(0, \frac{\sigma^2d_0^2}{L})$ if $L \geq C_{\gamma}\kappa^2s^2(K + N)/\eps^2.$
\end{condition}
We will prove Condition~\ref{cond:robust} in Section~\ref{s:init}. We now extract one useful result based on Conditions~\ref{cond:rip} and~\ref{cond:robust}.
From these two conditions, we are able to produce a good approximation of $F(\bh, \bx)$ for all $(\bh, \bx)\in\Kint$ in terms of $\delta$ in~\eqref{def:delta}. For $(\bh, \bx)\in\Kint$,  the following inequality holds
\begin{equation}\label{eq:LUBD}
\frac{2}{3}\delta^2d_0^2 -\frac{\eps\delta d_0^2}{5\sqrt{s}\kappa} + \|\be\|^2 \leq F(\bh, \bx) \leq \frac{3}{2}\delta^2d_0^2 + \frac{\eps\delta d_0^2}{5\sqrt{s}\kappa} + \|\be\|^2.
\end{equation}
Note that~\eqref{eq:LUBD}  simply follows from
\begin{equation*}
F(\bh, \bx) = \| \A(\BX - \BX_0) \|_F^2 - 2\Real(\lag \BX- \BX_0, \A^*(\be)\rag) + \|\be\|^2.
\end{equation*}
Note that~\eqref{eq:rip} implies $\frac{2}{3}\delta^2d_0^2\leq \|\A(\BX-\BX_0)\|_F^2\leq \frac{3}{2}\delta^2d_0^2$. Thus it suffices to estimate the cross-term, 
\begin{align}
|\Real(\lag \BX- \BX_0, \A^*(\be)\rag)| 
& \leq  \|\A^*(\be)\| \|\BX - \BX_0\|_* = \|\A^*(\be)\| \sum_{i=1}^s\|\bh_i\bx_i^* - \bh_{i0}\bx_{i0}^*\|_* \nonumber \\
& \leq  \sqrt{2}\|\A^*(\be)\| \sum_{i=1}^s\|\bh_i\bx_i^* - \bh_{i0}\bx_{i0}^*\|_F \nonumber \\ 
& \leq  \sqrt{2s} \|\A^*(\be)\| \|\BX - \BX_0\|_F \leq  \frac{\eps \delta d_0^2}{10\sqrt{s}\kappa} \label{eq:cross}
\end{align}
where $\|\cdot\|_*$ and $\|\cdot\|$ are a pair of dual norms and $\|\A^*(\be)\|$ comes from~\eqref{eq:robust}.

\vskip0.5cm
\subsection{Outline of the convergence analysis}
\label{s:outline}
For the ease of proof, we introduce another neighborhood:
\begin{equation*}
\KF = \left\{ (\bh,\bx) : \tF(\bh, \bx) \leq  \frac{\eps^2 d_0^2}{3s\kappa^2} + \|\be\|^2\right\}.
\end{equation*}
Moreover, another reason to consider $\KF$ is based on the fact that gradient descent~\emph{only} allows one to make the objective function decrease if the step size is chosen appropriately. 
In other words, all the iterates $\bz^{(t)}$ generated by gradient descent are inside $\KF$ as long as $\bz^{(0)}\in \KF.$ 

On the other hand, it is crucial to note that the decrease of the objective function does not necessarily imply the decrease of the relative error of the iterates. Therefore, we want to construct an initial guess in $\Keps\cap \KF$ so that $\bz^{(0)}$ is sufficiently close to the ground truth and then analyze the behavior of $\bz^{(t)}$.

\vskip0.25cm

In the rest of this section, we basically try to prove the following relation:
\begin{equation*}
\underbrace{ \frac{1}{\sqrt{3}}\Kd\cap \frac{1}{\sqrt{3}} \MN_{\mu} \cap \MN_{\frac{2\eps}{5\sqrt{s}\kappa}}}_{\text{Initial guess}}
 \subset \underbrace{\Keps\cap \KF}_{ \{\bz^{(t)}\}_{t\geq 0} \text{ in } \Keps\cap \KF } 
 \subset \underbrace{\Kint}_{\text{Key conditions hold over }\Kint }.
\end{equation*}

Now we give a more detailed explanation of the relation above, which constitutes the main structure of the proof:
\begin{enumerate}
\item We will show $\frac{1}{\sqrt{3}}\Kd\cap \frac{1}{\sqrt{3}} \MN_{\mu} \cap \MN_{\frac{2\eps}{5\sqrt{s}\kappa}} \subset \Keps\cap \KF$ in the proof of Theorem~\ref{thm:main} in Section~\ref{s:mainthm}, which is quite straightforward.

\item Lemma~\ref{lem:betamu} explains why it holds that $\Keps\cap \KF\subset \Kint$ and where the $s^2$-bottleneck comes from. 
\item Lemma~\ref{lem:line_section} implicitly shows that the iterates $\bz^{(t)}$ will remain in $\Keps\cap\KF$ if the initial guess $\bz^{(0)}$ is inside $\Keps\cap \KF$ and $\tF(\bz^{(t)})$ is monotonically decreasing (simply by induction). 
Lemma~\ref{lem:induction} makes this observation explicit by showing  that $\bz^{(t)}\in \Keps \cap \KF$ implies $\bz^{(t+1)} : = \bz^{(t)} - \eta\nabla \tF(\bz^{(t)})\in \Keps\cap\KF$ if the stepsize $\eta$ obeys $\eta \leq \frac{1}{C_L}$.
Moreover, Lemma~\ref{lem:induction} guarantees sufficient decrease of $\tF(\bz^{(t)})$ in each iteration, which paves the road towards the proof of linear convergence of $\tF(\bz^{(t)})$ and thus $\bz^{(t)}.$
\end{enumerate}

\vskip0.25cm
Remember that $\Kd$ and $\Kmu$ are both convex sets, and the purpose of introducing regularizers $G_i(\bh_i, \bx_i)$ is to approximately project the iterates onto $\Kd\cap\Kmu.$ Moreover, we hope that once the iterates are inside $\Keps$ and inside a sublevel subset $\KF$, they will never escape from $\KF\cap\Keps$. 
Those ideas are fully reflected in the following lemma. 

\begin{lemma}
\label{lem:betamu}
Assume $0.9d_{i0}\leq d_i\leq 1.1d_{i0}$ and $0.9d_{0}\leq d\leq 1.1d_0$. There holds $\KF \subset \Kd \cap \Kmu$; moreover, under Conditions~\ref{cond:rip} and~\ref{cond:robust}, we have $\KF \cap \Keps\subset \Kd \cap \Kmu\cap\MN_{\frac{9}{10}\epsilon}$.
\end{lemma}

\begin{proof}
If $(\vct{h}, \vct{x}) \notin \Kd \cap \Kmu$, by the definition of $G$ in~\eqref{def:G}, at least one component in $G$ exceeds $\rho G_0\left(\frac{2d_{i0}}{d_i}\right)$. We have 
\begin{eqnarray*}
\tF(\bh, \bx) & \geq & \rho G_0\left(\frac{2d_{i0}}{d_i}\right) \geq  (d^2 + 2\|\be\|^2) \left( \frac{2d_{i0}}{d_i} - 1\right)^2 \\
& \geq & (2/1.1 - 1)^2 (d^2 + 2\|\be\|^2) \\
& \geq & \frac{1}{2} d_{0}^2 + \|\be\|^2 > \frac{\eps^2 d_0^2}{3s \kappa^2} + \|\be\|^2,
\end{eqnarray*}
where $\rho \geq d^2 + 2\|\be\|^2$, $0.9d_0 \leq d \leq 1.1d_0$ and $0.9d_{i0} \leq d_i\leq 1.1d_{i0}.$
This implies $(\bh, \bx) \notin \KF$ and hence $\KF \subset \Kd \cap \Kmu$. \\

Note that $(\bh, \bx)\in \Kint$ if $(\bh, \bx) \in \KF \cap \Keps$. Applying~\eqref{eq:LUBD} gives
\begin{equation*}
\frac{2}{3}\delta^2d_0^2 -\frac{\eps\delta d_0^2}{5\sqrt{s}\kappa} + \|\be\|^2 \leq F(\bh, \bx)\leq \tF(\bh, \bx)\leq\frac{\eps^2 d_0^2}{3s\kappa^2} +  \|\be\|^2
\end{equation*}
which implies that $\delta \leq \frac{9}{10}\frac{\eps}{\sqrt{s}\kappa}.$ 
By definition of $\delta$ in~\eqref{def:delta}, there holds
\begin{equation}\label{eq:bottleneck}
\frac{81\eps^2}{100s\kappa^2} \geq \delta^2 = \frac{\sum_{i=1}^s \delta_i^2d_{i0}^2}{\sum_{i=1}^s d_{i0}^2} \geq \frac{\sum_{i=1}^s \delta_i^2}{s\kappa^2} \geq \frac{1}{s\kappa^2} \max_{1\leq i\leq s}\delta_i^2,
\end{equation}
which gives $\delta_i \leq \frac{9}{10}\eps$ and $(\bh, \bx)\in \MN_{\frac{9}{10}\eps}.$
\end{proof}
\begin{remark}
The $s^2$-bottleneck comes from~\eqref{eq:bottleneck}. If $\delta \leq \eps$ is small, we cannot guarantee that each $\delta_i$ is also smaller than $\eps$. Just consider the simplest case when all $d_{i0}$ are the same: then $d_0^2 = \sum_{i=1}^s d_{i0}^2 = s d_{i0}^2$ and there holds
\begin{equation*}
\eps^2\geq \delta^2 = \frac{1}{s}\sum_{i=1}^s \delta_i^2.
\end{equation*}
Obviously, we cannot conclude that $\max \delta_i \leq \eps$ but only say that  $\delta_i \leq \sqrt{s}\eps.$ This is why we require $\delta ={\cal O}(\frac{\eps}{\sqrt{s}})$ to ensure $\delta_i \leq \eps$, which gives $s^2$-dependence in $L.$

\end{remark}


\begin{lemma}
\label{lem:line_section}
Denote $\vct{z}_1 = (\bh_1, \bx_1)$ and $\vct{z}_2 = (\bh_2, \bx_2)$. Let $\vct{z}(\lambda):=(1-\lambda)\vct{z}_1 + \lambda \vct{z}_2$. If $\vct{z}_1 \in \Keps$ and $\vct{z}(\lambda) \in \KF$ for all $\lambda \in [0, 1]$, we have $\vct{z}_2 \in \Keps$.
\end{lemma}
\begin{proof}
Note that for $\bz_1\in\Keps\cap \KF$, we have $\bz_1\in \Kd\cap\Kmu\cap\MN_{\frac{9}{10}\eps}$ which follows from the second part of Lemma~\ref{lem:betamu}. Now we prove $\bz_2\in\Keps$ by contradiction.  Let us suppose that  $\bz_2 \notin \Keps$ and $\bz_1 \in \Keps$. There exists $\vct{z}(\lambda_0):=(\vct{h}(\lambda_0), \vct{x}(\lambda_0)) \in \Keps$ for some $\lambda_0 \in [0, 1]$ such that $\max_{1\leq i\leq s}\frac{\|\bh_i\bx_i^* - \bh_{i0}\bx_{i0}^*\|_F}{d_{i0}} = \epsilon$. Therefore, $\vct{z}(\lambda_0) \in \KF\cap\Keps$ and \prettyref{lem:betamu} implies $\max_{1\leq i\leq s}\frac{\|\bh_i\bx_i^* - \bh_{i0}\bx_{i0}^*\|_F}{d_{i0}} \leq \frac{9}{10}\epsilon$, which contradicts $\max_{1\leq i\leq s}\frac{\|\bh_i\bx_i^* - \bh_{i0}\bx_{i0}^*\|_F}{d_{i0}} = \epsilon$.
\end{proof}

\begin{lemma}
\label{lem:induction}
Let the stepsize $\eta \leq \frac{1}{C_L}$, $\bz^{(t)} : = (\bu^{(t)}, \bv^{(t)})\in\CC^{s(K + N)}$ and $C_L$ be the Lipschitz constant of $\nabla\tF(\bz)$ over $\Kint$ in~\eqref{def:CL}. If $\bz^{(t)}\in \Keps \cap \KF$, we have $\bz^{(t+1)} \in \Keps \cap \KF$ and
\begin{equation}
\label{eq:decreasing2}
\tF(\bz^{(t+1)}) \leq  \tF(\bz^{(t)}) - \eta \|\nabla \tF(\bz^{(t)})\|^2
\end{equation}
where $\bz^{(t+1)} = \bz^{(t)} - \eta\nabla\tF(\bz^{(t)}).$
\end{lemma}
\begin{remark}
This lemma tells us that once $\bz^{(t)}\in\Keps\cap\KF$, the next iterate $\bz^{(t+1)} = \bz^{(t)} - \eta \nabla \tF(\bz^{(t)})$ is also inside $\Keps\cap\KF$ as long as the stepsize $\eta \leq \frac{1}{C_L}$. In other words, $\Keps\cap\KF$ is in fact a stronger version of the basin of attraction. Moreover, the objective function will decay sufficiently in each step as long as we can control the lower bound of the $\nabla \tF$, which is guaranteed by the Local Regularity Condition~\ref{cond:rip}.

\end{remark}
\begin{proof}
Let $\phi(\tau) := \tF(\bz^{(t)} - \tau \nabla \tF(\bz^{(t)}))$, $\phi(0) = \tF(\bz^{(t)})$ and consider the following quantity:
\begin{equation*}
\tau_{\max}: = \max \{\mu: \phi(\tau) \leq \tF(\bz^{(t)}), 0\leq\tau \leq \mu \},
\end{equation*}
where $\tau_{\max}$ is the largest stepsize such that the objective function $\tF(\bz)$ evaluated at any point over the whole line segment $\{\bz^{(t)} -\tau \tF(\bz^{(t)}), 0\leq \tau\leq \tau_{\max}\}$ is not greater than $\tF(\bz^{(t)})$. Now we will show $\tau_{\max} \geq \frac{1}{C_L}$. Obviously, if $\|\nabla\tF(\bz^{(t)})\| = 0$, it holds automatically. 

Consider $\|\nabla\tF(\bz^{(t)})\|\neq 0$ and assume $\tau_{\max} < \frac{1}{C_L}$. 
First note that, 
\begin{equation*}
\frac{\diff}{\diff \tau} \phi(\tau) < 0 \Longrightarrow\tau_{\max} > 0.
\end{equation*}
By the definition of $\tau_{\max}$, there holds $\phi(\tau_{\max}) = \phi(0)$ since $\phi(\tau)$ is a continuous function w.r.t. $\tau$. Lemma~\ref{lem:line_section} implies
\begin{equation*}
\{ \bz^{(t)} - \tau \nabla\tF(\bz^{(t)}), 0\leq \tau \leq \tau_{\max} \} \subseteq \Keps\cap\KF.
\end{equation*}
 Now we apply Lemma~\ref{lem:DSL}, the modified descent lemma, and obtain
\begin{equation*}
\tF(\bz^{(t)} - \tau_{\max}\nabla\tF(\bz^{(t)})) \leq \tF(\bz^{(t)}) - (2\tau_{\max} - C_L\tau_{\max}^2)\|\tF(\bz^{(t)})\|^2 \leq \tF(\bz^{(t)}) - \tau_{\max}\|\tF(\bz^{(t)})\|^2
\end{equation*}
where $C_L\tau_{\max} \leq 1.$ In other words, $\phi(\tau_{\max}) \leq \tF(\bz^{(t)} - \tau_{\max}\nabla\tF(\bz^{(t)})) < \tF(\bz^{(t)}) = \phi(0)$ contradicts $\phi(\tau_{\max}) = \phi(0)$. 

Therefore, we conclude that $\tau_{\max} \geq \frac{1}{C_L}$. For any $\eta \leq \frac{1}{C_L}$, Lemma~\ref{lem:line_section} implies
\begin{equation*}
\{ \bz^{(t)} - \tau \nabla\tF(\bz^{(t)}), 0\leq \tau \leq \eta \} \subseteq \Keps\cap\KF
\end{equation*}
and applying Lemma~\ref{lem:DSL} gives
\begin{equation*}
\tF(\bz^{(t)} - \eta \nabla\tF(\bz^{(t)})) \leq \tF(\bz^{(t)}) - (2\eta - C_L\eta^2)\|\tF(\bz^{(t)})\|^2 \leq \tF(\bz^{(t)}) - \eta\|\tF(\bz^{(t)})\|^2.
\end{equation*}
\end{proof}

\subsection{Proof of Theorem~\ref{thm:main}}
\label{s:mainthm}
Combining all the considerations above, we now prove Theorem~\ref{thm:main} to conclude this section.
\begin{proof}
The proof consists of three parts: 
\paragraph{Part I: Proof of $\bz^{(0)} : = (\bu^{(0)}, \bv^{(0)}) \in \Keps\cap\KF$.} From the assumption of Theorem~\ref{thm:main},
\begin{equation*}
\bz^{(0)} \in \frac{1}{\sqrt{3}}\Kd \bigcap \frac{1}{\sqrt{3}}\Kmu\cap \MN_{\frac{2\eps}{5\sqrt{s}\kappa}}.
\end{equation*}
First we show $G(\bu^{(0)}, \bv^{(0)}) = 0$: for $0\leq i\leq s$ and the definition of $\Kd$ and $\Kmu$,
\begin{equation*}
\frac{\|\bu^{(0)}_i\|^2}{2d_i} \leq \frac{2d_{i0}}{3d_i} < 1, \quad \frac{L|\bb_l^* \bu^{(0)}_i|^2}{8d_i\mu^2} \leq \frac{L}{8d_i\mu^2} \cdot\frac{16d_{i0}\mu^2}{3L} \leq \frac{2d_{i0}}{3d_i} < 1,
\end{equation*}
where $\|\bu^{(0)}_i\| \leq \frac{2\sqrt{d_{i0}}}{\sqrt{3}}$, $\sqrt{L}\|\BB\bu^{(0)}_i\|_{\infty} \leq \frac{4 \sqrt{d_{i0}}\mu}{\sqrt{3}}$ and $\frac{9}{10}d_{i0} \leq d_i\leq \frac{11}{10}d_{i0}.$ Therefore 
$$G_0\left( \frac{\|\bu^{(0)}_i\|^2}{2d_i}\right) = G_0\left( \frac{\|\bv^{(0)}_i\|^2}{2d_i}\right) = G_0\left(\frac{L|\bb_l^*\bu_i^{(0)}|^2}{8d_i\mu^2}\right) = 0$$ 
for all $1\leq l\leq L$ and $G(\bu^{(0)}, \bv^{(0)}) = 0.$ 

For $\bz^{(0)} = (\bu^{(0)}, \bv^{(0)})\in \MN_{\frac{2\eps}{5\sqrt{s}\kappa}}$, we have 
$\delta(\bz^{(0)}) := \frac{\sqrt{\sum_{i=1}^s \delta_i^2d_{i0}^2 }}{d_0} \leq \frac{2\eps}{5\sqrt{s}\kappa}.$
By~\eqref{eq:LUBD}, there holds $\delta(\bz^{(0)}) \leq \frac{2\eps}{5\sqrt{s}\kappa}$ and $G(\bu^{(0)}, \bv^{(0)}) = 0$,
\begin{equation*}
\tF(\bu^{(0)}, \bv^{(0)}) = F(\bu^{(0)}, \bv^{(0)}) \leq \|\be\|^2 + \frac{3}{2}\delta^2(\bz^{(0)})d_0^2 + \frac{\eps \delta(\bz^{(0)}) d_0^2}{5\sqrt{s}\kappa} \leq \|\be\|^2 + \frac{\eps^2 d_0^2}{3s\kappa^2}
\end{equation*}
and hence $\bz^{(0)} = (\bu^{(0)}, \bv^{(0)})\in \Keps\bigcap \KF.$

\paragraph{Part II: The linear convergence of the objective function $\tF(\bz^{(t)})$.}

Denote $\bz^{(t)} : = (\bu^{(t)}, \bv^{(t)}).$ Note that $\bz^{(0)}\in\Keps\cap\KF$,  Lemma~\ref{lem:induction} implies $\bz^{(t)}\in \Keps\cap\KF$ for all $t\geq 0$ by induction if  $\eta \leq \frac{1}{C_L}$. Moreover, combining Condition~\ref{cond:reg} with Lemma~\ref{lem:induction} leads to 
\begin{equation*}
\tF(\bz^{(t )}) \leq \tF(\bz^{(t-1)}) - \eta\omega \left[  \tF(\bz^{(t-1)})  - c \right]_+, \quad t\geq 1
\end{equation*}
with $c = \|\be\|^2 + a\|\A^*(\be)\|^2$ and $a = 2000s$. 
Therefore, by induction, we have
\begin{equation*}
\left[ \tF(\bz^{(t)}) - c\right]_+ \leq (1 - \eta\omega)  \left[  \tF(\bz^{(t-1)})  - c \right]_+ \leq  \left(1 - \eta\omega\right)^t \left[ \tF(\bz^{(0)}) - c\right]_+ \leq \frac{\eps^2 d_0^2}{3s\kappa^2} (1 - \eta\omega)^{t}  
\end{equation*}
where $\tF(\bz^{(0)}) \leq \frac{\eps^2d_0^2}{3s\kappa^2} + \|\be\|^2$ and $\left[ \tF(\bz^{(0)}) - c \right]_+ \leq   \left[ \frac{1}{3s\kappa^2}\eps^2 d_0^2 - a\|\A^*(\be)\|^2 \right]_+ \leq \frac{\eps^2 d_0^2}{3s\kappa^2}.$
Now we conclude that  $\left[ \tF(\bz^{(t)}) - c\right]_+$ converges to $0$ linearly. 

\paragraph{Part III: The linear convergence of the iterates $(\bu^{(t)}, \bv^{(t)})$.}
Denote 
\begin{equation*}
\delta(\bz^{(t)}) : = \frac{\|\CH(\bu^{(t)}, \bv^{(t)}) - \CH(\bh_0,\bx_0)\|_F}{d_0}.
\end{equation*}
Note that $\bz^{(t)}\in \Keps\cap\KF\subseteq \Kint$ and over $\Kint$, there holds
$F_0(\bz^{(t)}) \geq \frac{2}{3}\delta^2(\bz^{(t)})d_0^2$
which follows from Local RIP Condition in~\eqref{eq:rip} and $F_0(\bz^{(t)})$ defined in~\eqref{def:F0}. Moreover
\begin{eqnarray*}
\tF(\bz^{(t)}) - \|\be\|^2 & \geq & F_0(\bz^{(t)}) - 2\Real\left(\lag \A^*(\be), \CH(\bu^{(0)}, \bv^{(0)}) - \CH(\bh_0, \bx_0) \rag\right)  \\
& \geq & \frac{2}{3} \delta^2(\bz^{(t)})d_0^2 - 2\sqrt{2s}\|\A^*(\be)\| \delta(\bz^{(t)})d_0
\end{eqnarray*}
where $G(\bz^{(t)}) \geq 0$ and the second inequality follows from~\eqref{eq:cross}.
There holds
\begin{equation*}
\frac{2}{3} \delta^2(\bz^{(t)})d_0^2 - 2\sqrt{2s}\|\A^*(\be)\| \delta(\bz^{(t)})d_0  - a\|\A^*(\be)\|^2   \leq \left[ \tF(\bz^{(t)}) - c \right]_+ \leq \frac{ \eps^2 d_0^2}{3s\kappa^2}(1 - \eta\omega)^t 
\end{equation*}
and equivalently, 
\begin{equation*}
\left|\delta(\bz^{(t)})d_0 - \frac{3\sqrt{2}}{2} \|\A^*(\be)\| \right|^2 \leq  \frac{\eps^2 d_0^2}{2s\kappa^2} (1 - \eta\omega)^t  + \left(\frac{3}{2}a + \frac{9}{2}\right)\|\A^*(\be)\|^2.
\end{equation*}
Solving the inequality above for $\delta(\bz^{(t)})$, we have 
\begin{eqnarray}
\delta(\bz^{(t)}) d_0 
& \leq &  \frac{\eps d_0}{\sqrt{2s\kappa^2}}(1 - \eta\omega)^{t/2}   +\left(\frac{3\sqrt{2}}{2} + \sqrt{\frac{3}{2}a + \frac{9}{2}} \right)\|\A^*(\be)\| \nonumber \\
& \leq &  \frac{\eps d_0}{\sqrt{2s\kappa^2}}(1 - \eta\omega)^{t/2} + 60\sqrt{s} \|\A^*(\be)\| \label{eq:main-res-2}
\end{eqnarray}
where $a = 2000s.$
Let $d^{(t)} : = \sqrt{\sum_{i=1}^s \|\bu_i^{(t)}\|^2\|\bv_i^{(t)}\|^2 }$ for $t\in\ZZ_{\geq 0}.$
By~\eqref{eq:main-res-2} and triangle inequality, we immediately obtain
$|d^{(t)} - d_0| \leq  \frac{\eps d_0}{\sqrt{2s\kappa^2}}(1 - \eta\omega)^{t/2} + 60\sqrt{s} \|\A^*(\be)\|.$
\end{proof}

\section{Proof of the four conditions}
This section is devoted to proving the four key conditions introduced in Section~\ref{s:converge}. The~\emph{local smoothness condition} and the~\emph{robustness condition} are relatively less challenging to deal with. The more difficult part is to show the~\emph{local regularity condition} and the~\emph{local isometry property}.
The key to solve those problems is to understand how the vector-valued linear operator $\A$ in~\eqref{def:A} behaves on block-diagonal matrices, such as $\CH(\bh,\bx)$, $\CH(\bh_0,\bx_0)$ and $\CH(\bh,\bx) - \CH(\bh_0,\bx_0).$ In particular, when $s=1$, all those matrices become rank-1 matrices, which have been well discussed in our previous work~\cite{LLSW16}.

\vskip0.25cm
First of all, we define the linear subspace $T_i\subset\CC^{K\times N}$ along with its orthogonal complement for $1\leq i\leq s$ as
\begin{eqnarray}\label{def:T}
\begin{split}
T_i & :=  \{ \BZ_i\in\CC^{K\times N} : \BZ_i = \bh_{i0}\bv_i^* + \bu_i\bx_{i0}^*, \quad \bu_i\in\CC^K,\bv_i\in\CC^N \}, \\ 
\TB_i & :=  \left\{ \left(\I_K - \frac{\bh_{i0}\bh_{i0}^*}{d_{i0}}\right) \BZ_i \left(\I_N - \frac{\bx_{i0}\bx_{i0}^*}{d_{i0}}\right) :\BZ_i\in\CC^{K\times N} \right\}
\end{split}
\end{eqnarray}
where $\|\bh_{i0}\| = \|\bx_{i0}\| = \sqrt{d_{i0}}.$
In particular, $\bh_{i0}\bx_{i0}^* \in T_i$ for all $1\leq i\leq s$.

\vskip0.25cm
The proof also requires us to consider block-diagonal matrices whose $i$-th block belongs to $T_i$ (or $\TB_i$). Let $\BZ = \blkdiag(\BZ_1,\cdots,\BZ_s)\in\CC^{Ks\times Ns}$ be a block-diagonal matrix and say $\BZ\in T$ if
\begin{equation*}
T := \{ \text{blkdiag}(\{\BZ_i\}_{i=1}^s) | \BZ_i\in T_i \}
\end{equation*}
and $\BZ\in\TB$ if
\begin{equation*}
\TB := \{ \text{blkdiag}(\{\BZ_i\}_{i=1}^s) | \BZ_i\in \TB_i \}
\end{equation*}
where both $T$ and $\TB$ are subsets in $\CC^{Ks\times Ns}$ and $\CH(\bh_0,\bx_0)\in T.$

\vskip0.25cm
Now we take a closer look at a special case of block-diagonal matrices, i.e., $\CH(\bh, \bx)$ and calculate its projection onto $T$ and $\TB$ respectively and it suffices to consider $\PP_{T_i}(\bh_i\bx_i^*)$ and $\PP_{\TB_i}(\bh_i\bx_i^*)$.
For each block $\bh_i\bx_i^*$ and $1\leq i\leq s$, there are unique orthogonal decompositions
\begin{equation}\label{eq:orth}
\bh_i := \alpha_{i1} \bh_{i0} + \tilde{\bh}_i, \quad \bx := \alpha_{i2} \bx_{i0} + \tilde{\bx}_i,
\end{equation}
where $\bh_{i0} \perp \tilde{\bh}_i$ and $\bx_{i0} \perp \tilde{\bx}_i$. It is important to note that $\alpha_{i1} = \alpha_{i1}(\bh_i) = \frac{\lag \bh_{i0}, \bh_i\rag}{d_{i0}}$ and $\alpha_{i2} = \alpha_{i2}(\bx_i) =  \frac{\lag \bx_{i0}, \bx_i\rag}{d_{i0}}$  and thus $\alpha_{i1}$ and $\alpha_{i2}$ are functions of $\bh_i$ and $\bx_i$ respectively.
Immediately, we have the following matrix orthogonal decomposition for $\bh_i\bx_i^*$ onto $T_i$ and $\TB_i$,
\begin{equation}\label{eq:decomposition}
\bh_i\bx_i^* - \bh_{i0} \bx_{i0}^* = \underbrace{(\alpha_{i1} \overline{\alpha_{i2}} - 1)\bh_{i0}\bx_{i0}^* + \overline{\alpha_{i2}} \tilde{\bh}_i \bx_{i0}^* + \alpha_{i1} \bh_{i0} \tilde{\bx}_i^*}_{\text{belong to } T_i}  + \underbrace{\tilde{\bh}_i \tilde{\bx}_i^*}_{\text{belongs to }\TB_i}
\end{equation}
where the first three components are in $T_i$ while $\tilde{\bh}_i\tilde{\bx}_i^*\in T^{\bot}_i$.

\subsection{Key lemmata}

From the decomposition in~\eqref{eq:orth} and~\eqref{eq:decomposition}, we want to analyze how $\|\tilde{\bh}_i\|$, $\|\tilde{\bx}_i\|$, $\alpha_{i1}$ and $\alpha_{i2}$ depend on $\delta_i = \frac{\|\bh_i\bx_i^* - \bh_{i0}\bx_{i0}^*\|_F}{d_{i0}}$ if $\delta_i < 1$. The following lemma answers this question, which can be viewed as an application of singular value/vector perturbation theory~\cite{Wedin72} applied to rank-1 matrices. From the lemma below, we can see that if $\bh_i\bx_i^*$ is close to $\bh_{i0}\bx_{i0}^*$, then $\PP_{\TB_i}(\bh_i\bx_i^*)$ is in fact very small (of order ${\cal O}(\delta_i^2 d_{i0})$).
\begin{lemma}{\bf (Lemma 5.9 in~\cite{LLSW16})}
\label{lem:orth_decomp}
Recall that $\|\bh_{i0}\| = \|\bx_{i0}\| = \sqrt{d_{i0}}$. If $\delta_i := \frac{\|\bh_i\bx_i^* - \bh_{i0} \bx_{i0}^*\|_F}{d_{i0}}<1$, we have the following useful bounds
\[
|\alpha_{i1}|\leq \frac{\|\bh_i\|}{\|\bh_{i0}\|}, \quad |\alpha_{i1}\overline{\alpha_{i2}} - 1|\leq \delta_i,
\]
and
\[
\|\tilde{\bh}_{i}\| \leq \frac{\delta_i}{1 - \delta_i}\|\bh_i\|,\quad \|\tilde{\bx}_i\| \leq \frac{\delta_i}{1 - \delta_i}\|\bx_i\|,\quad \|\tilde{\bh}_i\| \|\tilde{\bx}_i\| \leq \frac{\delta_i^2}{2(1 - \delta_i)} d_{i0}.
\]
Moreover, if $\|\bh_i\| \leq 2\sqrt{d_{i0}}$ and $\sqrt{L}\|\mtx{B} \bh_i\|_\infty \leq 4\mu \sqrt{d_{i0}}$, i.e., $\bh_i\in\Kd\bigcap\Kmu$, we have $\sqrt{L}\|\mtx{B}_i \tilde{\bh}_i\|_\infty \leq 6 \mu \sqrt{d_{i0}}$.
\end{lemma}


Now we start to focus on several results related to the linear operator $\A$.
\begin{lemma}{\bf (Operator norm of $\A$).}\label{lem:A-UPBD}
For $\A$ defined in~\eqref{def:A}, there holds
\begin{equation}\label{eq:A-UPBD}
\|\A\| \leq \sqrt{s(N\log(NL/2) + (\gamma+\log s)\log L)}
\end{equation}
with probability at least $1 - L^{-\gamma}.$
\end{lemma}
\begin{proof}
Note that $\A_i(\BZ_i) : = \{\bb_l^*\BZ_i\ba_{il}\}_{l=1}^L$ in~\eqref{def:Ai}. Lemma 1 in~\cite{RR12} implies
\begin{equation*}
\|\A_i\| \leq \sqrt{N\log(NL/2) + \gamma'\log L}
\end{equation*}
with probability at least $1 - L^{-\gamma'}.$ By taking the union bound over $1\leq i\leq s$, 
\begin{equation*}
\max\|\A_i\| \leq \sqrt{N\log(NL/2) + (\gamma+ \log s)\log L}
\end{equation*}
with probability at least $1 - sL^{-\gamma-\log s} \geq 1 - L^{-\gamma}.$

\vskip0.25cm

For $\A$ defined in~\eqref{def:A}, applying the triangle inequality gives
\begin{align*}
\|\A(\BZ)\| & = \left\|\sum_{i=1}^s \A_i(\BZ_i)\right\| \leq \sum_{i=1}^s \|\A_i\|\|\BZ_i\|_F  \leq \max_{1\leq i\leq s} \|\A_i\| \sqrt{s \sum_{i=1}^s \|\BZ_i\|_F^2} = \sqrt{s}\max_{1\leq i\leq s} \|\A_i\|  \|\BZ\|_F
\end{align*}
where $\BZ = \blkdiag(\BZ_1,\cdots, \BZ_s)\in\CC^{Ks\times Ns}.$
Therefore,
\begin{equation*}
\|\A\| \leq \sqrt{s}\max_{1\leq i\leq s}\|\A_i\| \leq \sqrt{ s(N\log(NL/2) + (\gamma+\log s)\log L)}
\end{equation*}
with probability at least $1 - L^{-\gamma}$.
\end{proof}

\vskip0.25cm

\begin{lemma}{\bf (Restricted isometry property for $\A$ on $T$).}
\label{lem:ripu}
The linear operator $\A$ restricted on $T$ is well-conditioned, i.e., 
\begin{equation}\label{eq:RIP-AT}
\|\PP_T\A^*\A\PP_T - \PP_T\| \leq \frac{1}{10}
\end{equation}
where $\PP_T$ is the projection operator from $\CC^{Ks\times Ns}$ onto $T$, given $L \geq C_{\gamma}s^2 \max\{K, \mu_h^2 N\}\log^2L$ with probability at least $1 - L^{-\gamma}.$
\end{lemma}
\begin{remark}
Here $\A\PP_{T}$ and $\PP_T\A^*$ are defined as
\begin{equation*}
\A\PP_T(\BZ) = \sum_{i=1}^s \A_i(\PP_{T_i}(\BZ_i)), \quad\PP_T\A^*(\bz) = \blkdiag( \PP_{T_1}(\A_1^*(\bz)), \cdots, \PP_{T_s}(\A_s^*(\bz)) )
\end{equation*}
respectively where $\BZ$ is a block-diagonal matrix and $\bz\in\CC^L.$
\end{remark}
As shown in the remark above, the proof of Lemma~\ref{lem:ripu} depends on the properties of both $\PP_{T_i}\A_i^*\A_i\PP_{T_i}$ and $\PP_{T_i}\A_i^*\A_j\PP_{T_j}$ for $i\neq j$. Fortunately, we have already proven related results in~\cite{LS17b} which are written as follows:
\begin{lemma}[\bf Inter-user incoherence, Corollary 5.3 and 5.8 in~\cite{LS17b}]
There hold
\begin{equation}\label{eq:AT1}
\|\PP_{T_i} \A_i^*\A_j\PP_{T_j}\|  \leq \frac{1}{10s},  \quad \forall i\neq j; \qquad \|\PP_{T_i} \A_i^*\A_i\PP_{T_i} - \PP_{T_i}\|  \leq \frac{1}{10s}, \quad\forall 1\leq i\leq s
\end{equation}
with probability at least $1 - L^{-\gamma+1}$ if $L\geq C_{\gamma}s^2\max\{K, \mu^2_hN\}\log^2L\log(s+1).$ 
\end{lemma}
Note that $\|\PP_{T_i} \A_i^*\A_j\PP_{T_j}\|  \leq \frac{1}{10s}$ holds because of independence between each individual random Gaussian matrix $\BA_i$.
In particular, if $s=1$, the inter-user incoherence $\|\PP_{T_i} \A_i^*\A_j\PP_{T_j}\|  \leq \frac{1}{10s}$ is not needed at all. With~\eqref{eq:AT1}, it is easy to prove Lemma~\ref{lem:ripu}. 
\begin{proof}[\bf Proof of Lemma~\ref{lem:ripu}]
For any block diagonal matrix $\BZ = \blkdiag(\BZ_1, \cdots,\BZ_s)\in\CC^{Ks\times Ns}$ and $\BZ_i\in\CC^{K\times N}$, 
\begin{align}
\lag \BZ, \PP_T\A^*\A\PP_T(\BZ) - \PP_T(\BZ)\rag 
& = \sum_{1\leq i,j\leq s} \lag \A_i\PP_{T_i}(\BZ_i), \A_j\PP_{T_j}(\BZ_j)\rag  - \|\PP_T(\BZ)\|_F^2 \nonumber \\
& =   \sum_{i=1}^s \lag \BZ_i, \PP_{T_i}\A_i^*\A_i\PP_{T_i}(\BZ_i) - \PP_{T_i}(\BZ_i)\rag +  \sum_{i\neq j} \lag \A_i\PP_{T_i}(\BZ_i), \A_j\PP_{T_j}(\BZ_j)\rag. \label{eq:TAAT2}
\end{align}
Using~\eqref{eq:AT1}, the following two inequalities hold,
\begin{align*}
|\lag \BZ_i, \PP_{T_i}\A_i^*\A_i\PP_{T_i}(\BZ_i) - \PP_{T_i}(\BZ_i)\rag| & \leq  \|\PP_{T_i}\A_i^*\A_i\PP_{T_i}  - \PP_{T_i} \| \|\BZ_i\|_F^2 \leq \frac{\|\BZ_i\|^2_F}{10s}, \\
|\lag \A_i\PP_{T_i}(\BZ_i), \A_j\PP_{T_j}(\BZ_j)\rag| & \leq \|\PP_{T_i}\A_i^*\A_j\PP_{T_j} \| \|\BZ_i\|_F\|\BZ_j\|_F \leq \frac{\|\BZ_i\|_F\|\BZ_j\|_F}{10s}. 
\end{align*}
After substituting both estimates into~\eqref{eq:TAAT2}, we have 
\begin{equation*}
|\lag \BZ, \PP_T\A^*\A\PP_T(\BZ) - \PP_T(\BZ)\rag| \leq \sum_{1\leq i, j\leq s} \frac{ \|\BZ_i\|_F\|\BZ_j\|_F }{10s} \leq \frac{1}{10s}\left(\sum_{i=1}^s \|\BZ_i\|_F\right)^2 \leq \frac{\|\BZ\|_F^2}{10}.
\end{equation*}
\end{proof}

Finally, we show how $\A$ behaves when applied to block-diagonal matrices $\BX = \CH(\bh,\bx)$. In particular, the calculations will be much simplified for the case $s=1$. 
\begin{lemma}{(\bf $\A$ restricted on block-diagonal matrices with rank-1 blocks).}
\label{lem:key}

\noindent Consider $\BX = \CH(\bh, \bx)$ and
\begin{equation}\label{eq:sigmamax}
\sigma^2_{\max}(\bh, \bx) := \max_{1\leq l\leq L} \sum_{i=1}^s |\bb^*_l\bh_i|^2 \|\bx_i\|^2.
\end{equation}
Conditioned on~\eqref{eq:A-UPBD}, we have
\begin{equation}\label{eq:AX-rank}
\|\A(\BX)\|^2  \leq \frac{4}{3} \|\BX\|_F^2+ 2  \sqrt{2s\|\BX\|_F^2 \sigma^2_{\max}(\bh, \bx)(K+N)\log L}  + 8s\sigma^2_{\max}(\bh, \bx)(K+N) \log L,
\end{equation}
uniformly for any $\bh\in\CC^{Ks}$ and $\bx\in\CC^{Ns}$ with probability  at least $1  - \frac{1}{\gamma}\exp(-s(K+N))$ if $L\geq C_{\gamma}s(K+N)\log L$. Here $ \|\BX\|_F^2= \|\CH(\bh, \bx)\|_F^2 = \sum_{i=1}^s \|\bh_i\|^2\|\bx_i\|^2.$
\end{lemma}

\begin{remark}
Here are a few more explanations and facts about $\sigma^2_{\max}(\bh,\bx)$.
Note that $\|\A(\BX)\|^2$ is the sum of $L$ sub-exponential~\footnote{For the definition and properties of sub-exponential random variables, the readers can find all relevant information in~\cite{Ver10}.} random variables, i.e., 
\begin{equation}\label{eq:AX}
\|\A(\BX)\|^2 = \sum_{l=1}^L \left|\sum_{i=1}^s \bb_l^*\bh_i \bx_i^*\ba_{il}\right|^2.
\end{equation}
Here $\sigma^2_{\max}(\bh, \bx)$ corresponds to the largest expectation of all those components in $\|\A(\BX)\|^2$. 

For $\sigma^2_{\max}(\bh, \bx)$, without loss of generality, we assume $\|\bx_i\| = 1$ for $1\leq i\leq s$ and let $\bh\in\CC^{Ks}$ be a unit vector, i.e., $\|\bh\|^2 = \sum_{i=1}^s \|\bh_i\|^2= 1$.  The bound
\begin{equation}\label{eq:sigma-lb}
\frac{1}{L} \leq \sigma^2_{\max}(\bh, \bx) \leq \frac{K}{L}
\end{equation}
follows from $L \sigma^2_{\max}(\bh, \bx) \geq \sum_{l=1}^L \sum_{i=1}^s |\bb_l^*\bh_i|^2 = \|\bh\|^2=1.$

Moreover, $\sigma_{\max}^2(\bh,\bx)$ and $\sigma_{\max}(\bh,\bx)$ are both Lipschitz functions w.r.t. $\bh.$ Now we want to determine their Lipschitz constants. First note that for $\|\bx_i\| = 1$, $\sigma_{\max}(\bh,\bx)$ equals
\begin{equation*}
\sigma_{\max}(\bh, \bx) =  \max_{1\leq l\leq L} \|(\I_s\otimes \bb_l^*)\bh\|
\end{equation*}
where $\otimes$ denotes Kronecker product. 
 Let $\bu\in\CC^{Ks}$ be another unit vector and we have
\begin{align}
|\sigma_{\max}(\bh, \bx) - \sigma_{\max}(\bu, \bx)| 
& = \left| \max_{1\leq l\leq L} \|(\I_s\otimes \bb_l^*)\bh  - \max_{1\leq l\leq L} \|(\I_s\otimes \bb_l^*)\bu\|  \right| \nonumber \\
& = \max_{1\leq l\leq L} \left| \|(\I_s\otimes \bb_l^*) \bh\| - \|(\I_s\otimes \bb_l^*) \bu\| \right| \nonumber \\
& \leq \max_{1\leq l\leq L} \|(\I_s\otimes \bb_l^*) (\bh - \bu)\| \leq \|\bh-\bu\| \label{eq:Lips-sigma}
\end{align}
where $\|\I_s\otimes \bb_l^*\| = \|\bb_l\| \sqrt{\frac{K}{L}} < 1.$ For $\sigma^2_{\max}(\bh,\bx),$
\begin{align}
|\sigma^2_{\max}(\bh, \bx) - \sigma^2_{\max}(\bu, \bx)| 
& \leq  (\sigma_{\max}(\bh, \bx) + \sigma_{\max}(\bu, \bx)) \cdot |\sigma_{\max}(\bh, \bx) - \sigma_{\max}(\bu, \bx)| \nonumber \\
& \leq \frac{2K}{L}\|\bh-\bu\| \leq  2\|\bh-\bu\|. \label{eq:Lips-sigmasq}
\end{align}
\end{remark}

\begin{proof}[{\bf Proof of Lemma~\ref{lem:key}}]
Without loss of generality, let  $\|\bx_i\| = 1$ and $\sum_{i=1}^s \|\bh_i\|^2 = 1$. It suffices to prove $f(\bh, \bx) \leq \frac{4}{3}$ for all $(\bh, \bx)\in\CC^{Ks}\times\CC^{Ns}$ in~\eqref{def:hx}  where $f(\bh, \bx)$ is defined as 
\begin{equation*}
f(\bh, \bx) := \|\A(\BX)\|^2 - 2  \sqrt{2s \sigma^2_{\max}(\bh, \bx)(K+N)\log L} - 8s\sigma^2_{\max}(\bh, \bx)(K+N) \log L.
\end{equation*}
\paragraph{Part I: Bounds of $\|\A(\BX)\|^2$ for any fixed $(\bh,\bx)$.}
From~\eqref{eq:AX}, we already know that $Y = \|\A(\BX)\|_F^2 = \sum_{i=1}^{2L} c_i\xi_i^2$ where $\{\xi_i\}$ are i.i.d. $\chi^2_1$ random variables and $\bc = (c_1, \cdots, c_{2L})^T\in \RR^{2L}$. More precisely, we can determine $\{c_i\}_{i=1}^{2L}$ as 
\begin{equation*}
\left| \sum_{i=1}^s \bb_l^*\bh_i\bx^*\ba_{il}\right|^2 = c_{2l-1} \xi_{2l-1}^2 + c_{2l}\xi_{2l}^2,\quad c_{2l-1} = c_{2l} = \frac{1}{2}\sum_{i=1}^s |\bb_l^*\bh_i|^2
\end{equation*}
because $\sum_{i=1}^s \bb^*_l\bh_i  \bx_i^*\ba_{il} \sim \mathcal{C}\mathcal{N}\left(0, \sum_{i=1}^s |\bb^*_l \bh_i|^2\right)$.

By the Bernstein inequality, there holds
\begin{equation}\label{ineq:bern}
\mathbb{P}(Y - \mathbb{E}(Y) \geq t) \leq \exp\left(- \frac{t^2}{8\|\vct{c}\|^2}\right) \vee \exp\left(- \frac{t}{8\|\vct{c}\|_\infty}\right)
\end{equation}
where $\E(Y) = \|\BX\|_F^2 = 1.$ In order to apply the Bernstein inequality, we need to estimate $\|\bc\|^2$ and $\|\bc\|_{\infty}$ as follows,
\begin{align*}
\|\bc\|_{\infty} & =  \frac{1}{2}\max_{1\leq l\leq L}\sum_{i=1}^s|\bb^*_l\bh_i|^2 = \frac{1}{2} \sigma^2_{\max}(\bh, \bx), \\
\|\bc\|_2^2 & =  \frac{1}{2}\sum_{l=1}^L \left|\sum_{i=1}^s|\bb^*_l\bh_i|^2 \right|^2 
 \leq \frac{1}{2}\left(  \sum_{i=1}^s\sum_{l=1}^L|\bb^*_l\bh_i|^2 \right)\max_{1\leq l\leq L}\sum_{i=1}^s|\bb^*_l\bh_i|^2  \leq \frac{1}{2} \sigma^2_{\max}(\bh, \bx).
\end{align*}
Applying~\eqref{ineq:bern} gives
\begin{equation*}
\mathbb{P}( \|\A(\BX)\|^2 \geq 1 + t)\leq  \exp\left(- \frac{t^2}{4 \sigma^2_{\max}(\bh, \bx)}\right) \vee \exp\left(- \frac{t}{4\sigma^2_{\max}(\bh, \bx)}\right).
\end{equation*}
In particular, by setting 
\begin{equation*}
t = g(\bh,\bx):= 2 \sqrt{2 s\sigma^2_{\max}(\bh, \bx)(K+N)\log L} + 8s\sigma^2_{\max}(\bh, \bx)(K + N)\log L,
\end{equation*}
we have
\begin{equation*}
\mathbb{P}\left(\|\A(\BX)\|^2 \geq 1 + g(\bh,\bx)\right) \leq e^{ - 2 s(K+N)(\log L)}.
\end{equation*}
So far, we have shown that $f(\bh, \bx) \leq 1$ with probability at least $1 - e^{- 2 s(K+N)(\log L)}$ for a fixed pair of $(\bh, \bx).$ 

\paragraph{Part II: Covering argument.}
Now we will use a covering argument to extend this result for all $(\bh, \bx)$ and thus prove that $f(\bh, \bx)\leq \frac{4}{3}$ uniformly for all $(\bh, \bx)$.

We start with defining $\mathcal{K}$ and $\mathcal{N}_i$ as $\epsilon_0$-nets of $\MS^{Ks-1}$ and $\MS^{N-1}$ for $\bh$ and $\bx_i,1\leq i\leq s$, respectively. The bounds $|\mathcal{K}|\leq (1+\frac{2}{\epsilon_0})^{2sK}$ and $|\mathcal{N}_i|\leq (1+\frac{2}{\epsilon_0})^{2N}$ follow from the covering numbers of the sphere (Lemma 5.2 in~\cite{Ver10}). Here we let $\mathcal{N} := \mathcal{N}_1\times \cdots \times  \mathcal{N}_s.$
By taking the union bound over $\mathcal{K}\times \mathcal{N},$ we have that $f(\bh, \bx)\leq 1$ holds uniformly for all $(\bh, \bx) \in \mathcal{K} \times \mathcal{N}$ with probability at least 
\begin{equation*}
1- \left(1+ 2/\epsilon_0\right)^{2s(K + N)} e^{ - 2s(K+N)\log L } = 1- e^{-2s(K + N)\left(\log L - \log \left(1 + 2/\eps_0\right)\right)}.
\end{equation*}
For any $(\bh, \bx) \in \MS^{Ks-1}\times \underbrace{\MS^{N-1}\times \cdots \times \MS^{N-1}}_{s \text{ times}}$, we can find a point $(\bu, \bv) \in \mathcal{K} \times \mathcal{N}$ satisfying $\| \bh - \bu \| \leq \eps_0$ and $\|\bx_i - \bv_i\| \leq \eps_0$ for all $1\leq i\leq s$. 
Conditioned on~\eqref{eq:A-UPBD}, we know that 
\begin{equation*}
\|\A\|^2\leq s(N\log(NL/2) + (\gamma + \log s)\log L) \leq s(N + \gamma + \log s)\log L.
\end{equation*}

Now we aim to evaluate $|f(\bh,\bx) - f(\bu,\bv)|$. 
First we consider $|f(\bu, \bx) - f(\bu, \bv)|$. Since $\sigma^2_{\max}(\bu, \bx) = \sigma^2_{\max}(\bu,\bv)$ if $\|\bx_i\| = \|\bv_i\| = \|\bu\|=1$ for $1\leq i\leq s$, we have 
\begin{eqnarray*}
|f(\bu, \bx) - f(\bu, \bv)| & = &  \left|\left\| \A(\CH(\bu, \bx))\right\|_F^2 - \left\| \A(\CH(\bu,\bv)) \right\|_F^2 \right| \\
& \leq & \left\| \A(\CH(\bu, \bx - \bv))\right\| \cdot \left\| \A(\CH(\bu, \bx + \bv))\right\| \\
& \leq & \|\A\|^2 \sqrt{\sum_{i=1}^s \|\bu_i\|^2\|\bx_i - \bv_i\|^2} \sqrt{\sum_{i=1}^s \|\bu_i\|^2\|\bx_i + \bv_i\|^2} \\
& \leq & 2\|\A\|^2 \eps_0 \leq 2s(N + \gamma + \log s)(\log L)\eps_0
\end{eqnarray*}
where the first inequality is due to $||z_1|^2 - |z_2|^2| \leq |z_1 - z_2||z_1 + z_2|$ for any $z_1, z_2 \in \mathbb{C}$.

We proceed to estimate  $|f(\bh, \bx) - f(\bu, \bx)|$ by using~\eqref{eq:Lips-sigmasq} and~\eqref{eq:Lips-sigma}, 
\begin{eqnarray*}
 | f(\bh, \bx) - f(\bu, \bx)| & \leq  & J_1 + J_2 + J_3 \\
& \leq &  (2\|\A\|^2 + 2\sqrt{2s(K+N)\log L}+ 16s(K+N) \log L) \eps_0\\
& \leq  & 25s(K +N + \gamma + \log s)(\log L) \eps_0
\end{eqnarray*}
where~\eqref{eq:Lips-sigmasq} and~\eqref{eq:Lips-sigma} give
\begin{eqnarray*}
J_1 & = &\left| \|\A(\CH(\bh,\bx))\|_F^2 - \|\A(\CH(\bu,\bx))\|_F^2\right| \leq \left\|  \A( \CH(\bh - \bu,\bx) )\right\| \left\| \A( \CH(\bh + \bu,\bx) )\right\| \leq 2\|\A\|^2 \eps_0, \\
J_2 & = & 2  \sqrt{2s(K+N)\log L}\cdot |\sigma_{\max}(\bh, \bx) - \sigma_{\max}(\bu, \bx)| \leq 2  \sqrt{2s(K+N)\log L} \eps_0, \\
J_3 & = & 8s(K+N) (\log L) \cdot |\sigma^2_{\max}(\bh, \bx) - \sigma^2_{\max}(\bu, \bx)|  \leq 16s(K+N)(\log L) \eps_0.
\end{eqnarray*}

\vskip0.25cm
Therefore, if $\epsilon_0 = \frac{1}{81s(N + K + \gamma + \log s)\log L}$, there holds
\begin{equation*}
f(\bh,\bx) \leq f(\bu,\bv) + \underbrace{|f(\bu, \bx) - f(\bu, \bv)| + |f(\bh,\bx) -f(\bu,\bx) |}_{\leq 27s(K+N+\gamma + \log s)(\log L)\eps_0\leq \frac{1}{3}} \leq \frac{4}{3}
\end{equation*}
for all $(\bh,\bx)$ uniformly with probability  at least $1- e^{-2s(K + N)\left(\log L - \log \left(1 + 2/\eps_0\right)\right)}.$
By letting $L \geq C_{\gamma}s(K+N)\log L$ with $C_{\gamma}$ reasonably large and $\gamma \geq 1$, we have $\log L - \log\left(1 + 2/\eps_0\right) \geq \frac{1}{2}(1 + \log(\gamma))$ and
with probability at least $1 - \frac{1}{\gamma}\exp(-s(K+N))$.
\end{proof}

\subsection{Proof of the local restricted isometry property}
\label{s:LRIP}
\begin{lemma}
\label{lem:rip}
Conditioned on~\eqref{eq:RIP-AT} and~\eqref{eq:AX-rank}, the following RIP type of property holds:
\begin{equation*}
\frac{2}{3} \|\BX - \BX_0\|_F^2 \leq \|\A(\BX - \BX_0)\|^2 \leq \frac{3}{2}\|\BX-\BX_0\|_F^2
\end{equation*}
uniformly for all $(\bh,\bx)\in \Kint$ with $\mu \geq \mu_h$ and $\epsilon \leq \frac{1}{15}$ if $L \geq C_{\gamma}\mu^2 s(K+N)\log^2 L$ for some numerical constant $C_{\gamma}$.
\end{lemma}

\begin{proof}
The main idea of the proof follows two steps: decompose $\BX-\BX_0$ onto $T$ and $\TB$, then apply~\eqref{eq:RIP-AT} and~\eqref{eq:AX-rank} to $\PP_T(\BX-\BX_0)$ and $\PP_{\TB}(\BX-\BX_0)$ respectively.

\vskip0.25cm
For any $\BX =\CH(\bh,\bx)\in\Keps$ with $\delta_i \leq \eps\leq \frac{1}{15}$, we can decompose $\BX - \BX_0$ as the sum of two block diagonal matrices $\BU = \blkdiag(\BU_i, 1\leq i\leq s)$ and $\BV = \blkdiag(\BV_i, 1\leq i\leq s)$ where each pair of $(\BU_i, \BV_i)$ corresponds to the orthogonal decomposition of $\bh_i\bx_i^* - \bh_{i0}\bx_{i0}^*$, 
\begin{equation}\label{eq:UV}
\bh_i\bx^*_i - \bh_{i0}\bx_{i0}^* :=
\underbrace{(\alpha_{i1} \overline{\alpha_{i2}} - 1)\bh_{i0}\bx_{i0}^* + \overline{\alpha_{i2}} \tilde{\bh}_i \bx_{i0}^* + \alpha_{i1} \bh_{i0}\tilde{\bx}_i^*}_{\BU_i\in T_i}  + \underbrace{ \tilde{\bh}_i \tilde{\bx}_i^*}_{\BV_i \in  T_i^\perp}
\end{equation}
which has been briefly discussed in~\eqref{eq:orth} and~\eqref{eq:decomposition}.
Note that
$\A(\BX - \BX_0) = \A(\BU + \BV)$ and 
\begin{equation*}
\|\A(\BU)\| - \|\A(\BV)\| \leq \|\A(\BU + \BV)\| \leq \|\A(\BU)\| + \|\A(\BV)\|.
\end{equation*}
Therefore, it suffices to have a two-side bound for $\|\A(\BU)\|$ and an upper bound for $\|\A(\BV)\|$ where $\BU\in T$ and $\BV\in \TB$ in order to establish the local isometry property.

\paragraph{Estimation of $\|\A(\BU)\|$:} For $\|\A(\BU)\|$, we know from~\prettyref{lem:ripu} that 
\begin{equation}\label{eq:AU1}
\sqrt{\frac{9}{10}}\|\BU\|_F\leq \|\A(\BU)\| \leq \sqrt{\frac{11}{10}}\|\BU\|_F
\end{equation}
and hence we only need to compute $\|\BU\|_F.$
By \prettyref{lem:orth_decomp}, there also hold $\|\BV_i\|_F \leq \frac{\delta_i^2}{2(1 - \delta_i)} d_{i0}$ and $\delta_i - \|\BV_i\|_F\leq \|\BU_i\|_F\leq \delta_i + \|\BV_i\|_F$, i.e., 
\begin{equation*}
\left(\delta_i - \frac{\delta_i^2}{2(1 - \delta_i)}\right)d_{i0} \leq \|\BU_i\|_F \leq \left(\delta_i + \frac{\delta_i^2}{2(1 - \delta_i)}\right)d_{i0}, \quad 1\leq i\leq s.
\end{equation*}
With $\|\BU\|_F^2 = \sum_{i=1}^s \|\BU_i\|_F^2$, 
it is easy to get $\delta d_0\left(1 - \frac{\eps}{2(1-\eps)}\right)
\leq \|\BU\|_F \leq \delta d_0 \left(1 + \frac{\eps}{2(1-\eps)}\right)$.
Combined with~\eqref{eq:AU1}, we get
\begin{equation}
\label{eq:AU}
\sqrt{\frac{9}{10}}\left(1 - \frac{\eps}{2(1-\eps)}\right)\delta d_0 \leq \|\A(\BU) \| \leq \sqrt{\frac{11}{10}}\left(1 + \frac{\eps}{2(1-\eps)}\right)\delta d_0.
\end{equation}
\paragraph{Estimation of $\|\A(\BV)\|$:}
Note that $\BV$ is a block-diagonal matrix with rank-1 block. So  applying \prettyref{lem:key} gives us
\begin{align}
\|\A(\BV)\|^2 
&\leq \frac{4}{3} \|\BV\|_F^2+ 2  \sqrt{2s\|\BV\|_F^2 \sigma^2_{\max}(\tilde{\bh}, \tilde{\bx})(K+N)\log L}  + 8s\sigma^2_{\max}(\tilde{\bh}, \tilde{\bx})(K+N) \log L \label{ineq:AV} 
\end{align}
where $\BV = \CH(\tilde{\bh}, \tilde{\bx})$ and $\tilde{\bh} = \begin{bmatrix}
\tilde{\bh}_1 \\
\vdots \\
\tilde{\bh}_s
\end{bmatrix}.
$
It suffices to get an estimation of $\|\BV\|_F$ and $\sigma^2_{\max}(\tilde{\bh},\tilde{\bx})$ to bound $\|\A(\BV)\|$ in~\eqref{ineq:AV}.

Lemma~\ref{lem:orth_decomp} says that $\|\tilde{\bh}_i\| \|\tilde{\bx}_i\| \leq \frac{\delta_i^2}{2(1 - \delta_i)} d_{i0} \leq \frac{\eps}{2(1-\eps)} \delta_i d_{i0}$ if $\eps < 1$. Moreover,
\begin{equation}\label{eq:orth-1}
\|\tilde{\bx}_i\| \leq \frac{\delta_i}{1 - \delta_i}\|\bx_i\| \leq \frac{2\delta_{i}}{1 - \delta_i} \sqrt{d_{i0}}, \quad \sqrt{L}\|\mtx{B} \tilde{\bh}_i \|_\infty \leq 6 \mu \sqrt{d_{i0}}, \quad 1\leq i\leq s
\end{equation}
if $(\bh,\bx)$ belongs to $\Kint.$
For $\|\BV\|_F$,
\begin{equation*}
\|\BV\|_F = \sqrt{\sum_{i=1}^s \|\BV_i\|_F^2} = \sqrt{\sum_{i=1}^s \|\tilde{\bh}_i\|^2 \|\tilde{\bx}_i\|^2} \leq \frac{\eps\delta d_0}{2(1-\eps)}.
\end{equation*}
Now we aim to get an upper bound for $\sigma^2_{\max}(\tilde{\bh}, \tilde{\bx})$ by using~\eqref{eq:orth-1},
\begin{equation*}
\sigma_{\max}^2(\tilde{\bh}, \tilde{\bx}) = \max_{1\leq l\leq L}\sum_{i=1}^s |\bb^*_l\tilde{\bh}_i|^2 \|\tilde{\bx}_i\|^2 \leq C_0\frac{\mu^2 \sum_{i=1}^s \delta_i^2 d_{i0}^2}{L} = C_0\frac{\mu^2\delta^2 d_0^2}{L}.
\end{equation*}
By substituting the estimations of $\|\BV\|_F$ and $\sigma^2_{\max}(\tilde{\bh}, \tilde{\bx})$ into~\eqref{ineq:AV}
\begin{equation}\label{eq:AV}
\|\A(\BV)\|^2 \leq \frac{\eps^2 \delta^2d_0^2}{3(1-\eps)^2} + \frac{\sqrt{2}\eps \delta^2 d_0^2}{1-\eps} \sqrt{\frac{C_0\mu^2 s (K+N)\log L}{L}} + \frac{8C_0 \mu^2 \delta^2d_0^2s(K+N)\log L}{L}.
\end{equation}
By letting $L \geq C_{\gamma }\mu^2 s(K + N)\log^2 L$ with $C_{\gamma} $ sufficiently large and combining \prettyref{eq:AV} and \prettyref{eq:AU}, we have
\begin{equation*}
\sqrt{\frac{2}{3}}\delta d_0 \leq \|\A(\BU)\| - \|\A(\BV)\|  \leq \|\A(\BU+\BV)\| \leq \|\A(\BU)\| + \|\A(\BV)\|  \leq \sqrt{\frac{3}{2}}\delta d_0,
\end{equation*}
which gives $\frac{2}{3}\|\BX - \BX_0\|_F^2 \leq \|\A(\BX - \BX_0)\|^2 \leq \frac{3}{2}\|\BX - \BX_0\|_F^2.$
\end{proof}

\subsection{Proof of the local regularity condition} 
\label{s:LRC}
We first introduce a few notations: for all $(\bh, \bx) \in \Kd \cap \Keps$, consider $\alpha_{i1}, \alpha_{i2}, \tilde{\bh}_i$ and $\tilde{\bx}_i$ defined in \prettyref{eq:orth} and define
\begin{equation*}
\Dh_i = \bh_i - \alpha_i \bh_{i0}, \quad \Dx_i = \bx_i - \overline{\alpha}_i^{-1}\bx_{i0}
\end{equation*}
where 
\begin{equation*}
\alpha_i (\bh_i, \bx_i)= 
\begin{cases} 
(1 - \delta_0)\alpha_{i1}, & \text{~if~} \|\bh_i\|_2 \geq \|\bx_i\|_2 \\ 
\frac{1}{(1 - \delta_0)\overline{\alpha_{i2}}}, & \text{~if~} \|\bh_i\|_2 < \|\bx_i\|_2\end{cases}
\end{equation*}
with 
\begin{equation}\label{def:delta0}
\delta_0 := \frac{\delta}{10}.
\end{equation} 

The function $\alpha_i(\bh_i,\bx_i)$ is defined for each block of $\BX = \CH(\bh, \bx).$
The particular form of  $\alpha_i(\bh, \bx)$ serves primarily for proving the  Lemma~\ref{lem:regG}, i.e., local regularity condition of $G(\bh, \bx)$. We also define
\begin{equation*}
\Dh : = 
\begin{bmatrix}
\bh_1 - \alpha_1 \bh_{1,0} \\
\vdots \\
\bh_s - \alpha_s \bh_{s0}
\end{bmatrix}\in\CC^{Ks}, \quad
\Dx : = 
\begin{bmatrix}
\bx_1 - \alpha_1 \bx_{1,0} \\
\vdots \\
\bx_s - \alpha_s \bx_{s0}
\end{bmatrix}\in\CC^{Ns}.
\end{equation*}
The following lemma gives bounds of $\Dx_i$ and $\Dh_i$.
\begin{lemma}
\label{lem:DxDh}
For all $(\bh,\bx) \in \Kd \cap \Keps$ with $\epsilon \leq \frac{1}{15}$, there hold 
\begin{align*}
\max\{ \|\Dh_i\|_2^2, \|\Dx_i\|_2^2\} & \leq (7.5\delta_i^2 + 2.88\delta_0^2) d_{i0}, \\
\|\Dh_i\|_2^2 \|\Dx_i\|_2^2 & \leq \frac{1}{26}(\delta_i^2 + \delta_0^2) d_{i0}^2.
\end{align*}
Moreover, if we assume $(\bh_i, \bx_i) \in \Kmu$ additionally, we have $ \sqrt{L}\|\BB(\Dh_i)\|_\infty \leq 6\mu\sqrt{d_{i0}}$.
\end{lemma}

\begin{proof}
We only consider $\|\bh_i\|_2 \geq \|\bx_i\|_2$ and $\alpha_i = (1 - \delta_0) \alpha_{1i}$, and the other case is exactly the same due to the symmetry. For both $\Dh_i$ and $\Dx_i$, by definition, 
\begin{align}
\Dh_i & = \bh_i - \alpha_{i}\bh_{i0} =  \delta_0 \alpha_{i1} \bh_{i0} + \tilde{\bh}_i\label{def:Dhi},\\
\Dx_i & = \bx_i - \frac{1}{(1 - \delta_0)\overline{\alpha_i}_1} \bx_{i0} = \left(\alpha_{i2} - \frac{1}{(1 - \delta_0)\overline{\alpha}_{i1}}\right)\bx_{i0} + \tilde{\bx}_i \label{def:Dxi},
\end{align}
where $\bh_i = \alpha_{i1}\bh_{i0} + \tilde{\bh}_i$ and $\bx_i = \alpha_{i2}\bx_{i0} + \tilde{\bx}_i$ come from the orthogonal decomposition in~\eqref{eq:orth}.

\vskip0.25cm
We start with estimating $\|\Dh_i\|^2.$ Note that $\|\bh_i\|_2^2 \leq 4d_{i0}$ and $\|\alpha_{i1} \bh_{i0}\|_2^2\leq \|\bh_i\|_2^2$ since $(\bh, \bx)\in\Kd\cap\Kmu$. By \prettyref{lem:orth_decomp}, we have 
\begin{equation}\label{eq:Dh-est}
\|\Dh_i\|_2^2 = \|\tilde{\bh}_i\|_2^2 + \delta_0^2\|\alpha_{i1} \bh_{i0}\|_2^2 \leq \left(\left(\frac{\delta_i}{1-\delta_i}\right)^2 + \delta_0^2\right)\|\bh_i\|_2^2 \leq ( 4.6\delta_i^2 + 4\delta_0^2) d_{i0}.
\end{equation}

Then we calculate $\|\Dx_i\|$: from~\eqref{def:Dxi}, we have
\begin{equation*}
\|\Dx_i\|^2 = \left|\alpha_{i2} - \frac{1}{(1 - \delta_0)\overline{\alpha}_{i1}} \right|^2d_{i0} + \|\tilde{\bx}_i\|^2 \leq \left|\alpha_{i2} - \frac{1}{(1 - \delta_0)\overline{\alpha}_{i1}} \right|^2d_{i0} + \frac{4\delta_i^2 d_{i0}}{(1 - \delta_i)^2},
\end{equation*}
where \prettyref{lem:orth_decomp} gives $\|\tilde{\bx}_i\|_2 \leq \frac{\delta_i}{1-\delta_i}\|\bx_i\|_2 \leq \frac{2\delta_i}{1-\delta_i} \sqrt{d_{i0}}$ for $(\bh,\bx)\in\Kd\cap\Keps$.

So it suffices to estimate $\left| \alpha_{i2} - \frac{1}{(1 - \delta_0)\overline{\alpha}_{i1}} \right|$, which satisfies
\begin{equation}\label{eq:coef-up}
\left|\alpha_{i2} - \frac{1}{(1 - \delta_0)\overline{\alpha_{i1}}}\right| 
=
\frac{1}{|\alpha_{i1}|} \left| \overline{\alpha_{i1}} \alpha_{i2}- 1  
-  \frac{\delta_0}{1 - \delta_0} \right|
\leq 
\frac{1}{|\alpha_{i1}|} \left( \left|(\overline{\alpha_{i1}} \alpha_{i2}- 1)\right| 
+  \frac{\delta_0}{1 - \delta_0} \right).
\end{equation}
Lemma~\ref{lem:orth_decomp} implies that $| \overline{\alpha_{i1}} \alpha_{i2}- 1| \leq \delta_i$, and~\eqref{eq:orth} gives
\begin{equation}\label{eq:lb-alpha}
|\alpha_{i1}|^2 = \frac{1}{d_{i0}}(\|\bh_i\|^2 - \| \tilde{\bh}_i \|^2) \geq  \frac{1}{d_{i0}}\left(1 - \frac{\delta_i^2}{(1-\delta_i)^2} \right)\|\bh_i\|^2 \geq \left(1 - \frac{\delta_i^2}{(1-\delta_i)^2} \right)(1-\eps)
\end{equation}
where $\|\tilde{\bh}_i\| \leq \frac{\delta_i}{1-\delta_i}\|\bh_i\|$ and $\|\bh_i\|^2 \geq \|\bh_i\|\|\bx_i\| \geq (1-\eps)d_{i0}$ if $\|\bh_i\|\geq \|\bx_i\|.$ Substituting~\eqref{eq:lb-alpha} into~\eqref{eq:coef-up} gives
\begin{eqnarray*}
\left|\alpha_{i2} - \frac{1}{(1 - \delta_0)\overline{\alpha_{i1}}}\right| 
& \leq & \frac{1}{\sqrt{1-\eps}} \left(1 - \frac{\delta_i^2}{(1-\delta_i)^2} \right)^{-1/2}\left(\delta_i + \frac{\delta_0}{1-\delta_0}\right) \leq 1.2(\delta_i + \delta_0).
\end{eqnarray*}
Then we have
\begin{eqnarray}
\|\Dx_i\|_2^2 
& \leq & \left(1.44(\delta_i+\delta_0)^2+ \frac{4\delta^2_i}{(1 - \delta_i)^2}\right)  d_{i0} \leq (7.5\delta_i^2 + 2.88\delta_0^2)d_{i0}.\label{eq:Dx-est}
\end{eqnarray}

Finally, we try to bound $\|\Dh_i\|^2\|\Dx_i\|^2.$
\prettyref{lem:orth_decomp} gives $\|\tilde{\bh}_i\|_2 \|\tilde{\bx}_i\|_2 \leq \frac{\delta_i^2d_{i0}}{2(1 - \delta_i)}$ and  $|\alpha_{i1}| \leq 2$.
Combining them along with~\eqref{def:Dhi},~\eqref{def:Dxi},~\eqref{eq:Dh-est} and~\eqref{eq:Dx-est}, we have
\begin{align*}
\|\Dh_i\|_2^2 \|\Dx_i\|_2^2 &\leq \|\tilde{\bh}_i\|_2^2\|\tilde{\bx}_i\|_2^2 + \delta_0^2 |\alpha_{i1}|^2 \|\bh_{i0}\|_2^2 \|\Dx_i\|_2^2 + \left|\alpha_{i2} - \frac{1}{(1 - \delta_0)\overline{\alpha}_{i1}}\right|^2 \|\bx_{i0}\|_2^2 \|\Dh_i\|_2^2 
\\
& \leq \left(\frac{\delta_i^4}{4(1 - \delta_i)^2}  + 4\delta_0^2 (7.5\delta_i^2 + 2.88\delta_0^2) + 1.44(\delta_i + \delta_0)^2 (4.6\delta_i^2 + 4\delta_0^2 )\right) d_{i0}^2 \\
& \leq \frac{(\delta_i^2 + \delta_0^2)d_{i0}^2}{26}. 
\end{align*}

By symmetry, similar results hold for the case $\|\bh_i\|_2 < \|\bx_i\|_2$ and $\max\{\|\Dh_i\|, \|\Dx_i\|\} \leq (7.5\delta_i^2 + 2.88\delta_0^2)d_{i0}.$
\vskip0.25cm
Next, under the additional assumption $(\bh, \bx) \in \Kmu$, we now prove $\sqrt{L}\|\mtx{B}(\Dh_i)\|_\infty \leq 6\mu\sqrt{d_{i0}}$:\\
Case 1: $\|\bh_i\|_2 \geq \|\bx_i\|_2$ and $\alpha_i = (1 - \delta_0) \alpha_{i1}$. By \prettyref{lem:orth_decomp} gives $|\alpha_{i1}| \leq 2$, which implies
\begin{align*}
\sqrt{L}\|\mtx{B}(\Dh_i)\|_\infty &\leq \sqrt{L}\|\mtx{B}\bh_i \|_\infty + (1 - \delta_0) |\alpha_{i1}|\sqrt{L}\|\mtx{B}\bh_{i0}\|_\infty 
\\
&\leq 4\mu\sqrt{d_{i0}} + 2(1 - \delta_0)\mu_h \sqrt{d_{i0}} \leq 6\mu\sqrt{d_{i0}}.
\end{align*}
Case 2: $\|\bh_i\|_2 < \|\bx_i\|_2$ and $\alpha_i = \frac{1}{(1-\delta_0)\overline{\alpha_{i2}}}$. Using the same argument as~\eqref{eq:lb-alpha} gives
\begin{equation*}
|\alpha_{i2}|^2\geq \left(1 - \frac{\delta_i^2}{(1-\delta_i)^2} \right)(1-\eps).
\end{equation*}
Therefore,
\begin{align*}
\sqrt{L}\|\mtx{B}(\Dh_i)\|_\infty &\leq \sqrt{L}\|\mtx{B}\bh_i\|_\infty + \frac{1}{(1 - \delta_0) |\overline{\alpha_{i2}}|} \sqrt{L}\|\mtx{B}\bh_0\|_\infty
\\
&\leq 4\mu\sqrt{d_0} + \left(1 - \frac{\delta_i^2}{(1-\delta_i)^2} \right)^{-1/2} \frac{\mu_h \sqrt{d}_0}{(1-\delta_0)\sqrt{1-\eps}} \leq 6 \mu\sqrt{d_0}.
\end{align*}
\end{proof}

\begin{lemma}{\bf (Local Regularity for $F(\bh,\bx)$)}\label{lem:regF}
Conditioned on~\eqref{eq:rip} and~\eqref{eq:AX-rank},  the  following inequality holds
\begin{equation*}
 \Real\lp\lag \nabla F_{\bh}, \Dh \rag + \lag \nabla F_{\bx}, \Dx\rag\rp 
\geq \frac{\delta^2 d_0^2}{8} - 2\sqrt{s}\delta d_0 \|\A^*(\be)\|,
\end{equation*}
uniformly for any $(\bh, \bx) \in \Kint$ with $\epsilon \leq \frac{1}{15}$ if
$L \geq C\mu^2 s(K+N)\log^2 L$  for some numerical constant $C$.
\end{lemma}

\begin{proof}
First note that for 
\begin{equation*}
I_0 =  \lag \nabla F_{\bh}, \Dh \rag + \overline{\lag \nabla F_{\bx}, \Dx\rag }  = \sum_{i=1}^s \lag \nabla F_{\bh_i}, \Dh_i \rag + \overline{\lag \nabla F_{\bx_i}, \Dx_i \rag}.
\end{equation*}
For each component, recall that~\eqref{eq:WFh} and~\eqref{eq:WFx}, we have
\begin{align*}
\lag \nabla F_{\bh_i}, \Dh_i \rag + \overline{\lag \nabla F_{\bx_i}, \Dx_i \rag}
& = \lag \A_i^*(\A(\BX - \BX_0) - \be)\bx_i, \Dh_i \rag  + \overline{\lag (\A_i^*(\A(\BX - \BX_0) - \be))^*\bh_i, \Dx_i \rag}  \\
& = \left\lag \A(\BX - \BX_0)  - \be, \A_i((\Dh_i)\bx_i^* + \bh_i (\Dx_i)^*) \right\rag.
\end{align*}
Define $\BU_i$ and $\BV_i$ as
\begin{equation}
\label{eq:UV2}
\BU_i := \alpha_i\bh_{i0}(\Dx_i)^* + \overline{\alpha_i}^{-1}(\Dh_i)\bx_{i0}^* \in T_i, \quad \BV_i := \Dh_i(\Dx_i)^*.
\end{equation}
Here $\BV_i$ does not necessarily belong to $\TB_i.$ From the way of how $\Dh_i$, $\Dx_i$, $\BU_i$ and $\BV_i$ are constructed, two simple relations hold:
\begin{eqnarray*}
\bh_i\bx_i^* - \bh_{i0}\bx_{i0}^* & = &  \BU_i + \BV_i, \\
(\Dh_i)\bx_i^* + \bh_i (\Dx_i)^* & = & \BU_i + 2\BV_i.
\end{eqnarray*}
Define $\BU : = \blkdiag(\BU_1, \cdots, \BU_s)$ and $\BV : = \blkdiag(\BV_1, \cdots, \BV_s)$. $I_0$ can be simplified to
\begin{align*}
I_0 & = \sum_{i=1}^s \lag \nabla F_{\bh_i}, \Dh_i \rag + \overline{\lag \nabla F_{\bx_i}, \Dx_i \rag} = \sum_{i=1}^s \lag \A(\BU+\BV)-  \be, \A_i(\BU_i + 2\BV_i)\rag \\
& = \underbrace{\lag \A(\BU+\BV), \A(\BU + 2\BV)\rag}_{I_{01}} - \underbrace{\lag  \be, \A(\BU + 2\BV)\rag}_{I_{02}}.
\end{align*}
Now we will give a lower bound for $\Real(I_{01})$ and an upper bound for $\Real(I_{02})$ so that the lower bound of $\Real(I_0)$ is obtained. 
By the Cauchy-Schwarz inequality, $\Real(I_{01})$ has the lower bound
\begin{equation}\label{eq:I_0}
\Real(I_{01}) \geq  (\|\A(\BU)\| - \|\A(\BV)\|) (\|\A(\BU)\| - 2\|\A(\BV)\|). 
\end{equation}
In the following, we will give an upper bound for $\|\A(\BV)\|$ and a lower bound for $\|\A(\BU)\|$.

\paragraph{Upper bound for $\|\A(\BV)\|$:} Note that $\BV$ is a block-diagonal matrix with rank-1 blocks, and applying \prettyref{lem:key} results in
\begin{equation*}
\|\A(\BV)\|^2 \leq \frac{4}{3}\sum_{i=1}^s \|\BV\|_F^2 + 2\sigma_{\max}(\Dh,\Dx)  \|\BV\|_F\sqrt{2s(K+N)\log L} + 8s\sigma_{\max}^2(\Dh, \Dx)(K+N) \log L.
\end{equation*}
By using Lemma~\ref{lem:DxDh}, we have $\|\Dh_i\|^2 \leq (7.5\delta_i^2 + 2.88\delta_0^2)d_{i0}$ and $\sqrt{L}\|\BB(\Dh_i)\|_{\infty} \leq 6\mu\sqrt{d_{i0}}$. Substituting them into $\sigma^2_{\max}(\Dh,\Dx)$ gives
\begin{eqnarray*}
\sigma_{\max}^2(\Dh, \Dx) = \max_{1\leq l\leq L}\left(\sum_{i=1}^s |\bb_l^*\Dh_i|^2 \|\Dx_i\|^2\right) \leq \frac{36\mu^2}{L} \sum_{i=1}^s  (7.5\delta_i^2 + 2.88\delta_0^2)d_{i0}^2  \leq \frac{272\mu^2 \delta^2 d_0^2}{L}.
\end{eqnarray*}
For $\|\BV\|_F$, note that $\|\Dh_i\|^2\|\Dx_i\|^2 \leq \frac{1}{26}(\delta_i^2 + \delta_0^2)d_{i0}^2$ and thus
\begin{equation*}
\|\BV\|^2_F = \sum_{i=1}^s \|\Dh_i\|^2\|\Dx_i\|^2 \leq\frac{1}{26} \sum_{i=1}^s(\delta_i^2 + \delta_0^2)d_{i0}^2 \leq \frac{1}{26} \cdot 1.01\delta^2d_0^2 = \frac{\delta^2 d_0^2}{25}.
\end{equation*}

Then by $\delta  \leq \eps \leq  \frac{1}{15}$ and letting $L \geq C\mu^2 s(K + N)\log^2 L$ for a sufficiently large numerical constant $C$, there holds
\begin{equation}\label{eq:AVdelta}
\|\A(\BV)\|^2 \leq \frac{\delta^2 d_0^2}{16} \implies \|\A(\BV)\| \leq \frac{\delta d_0}{4}.
\end{equation}

\paragraph{Lower bound for $\|\A(\BU)\|$:}  By the triangle inequality, 
$\|\BU\|_F \geq \delta d_0 - \frac{1}{5} \delta d_0 \geq \frac{4}{5}\delta d_0$
if $\epsilon \leq \frac{1}{15}$ since $\|\BV\|_F \leq 0.2\delta d_0$. Since $\BU \in T$, by \prettyref{lem:ripu}, there holds
\begin{equation}\label{eq:AUdelta}
\|\A(\BU)\| \geq \sqrt{\frac{9}{10}}\|\BU\|_F \geq \frac{3}{4} \delta d_0.
\end{equation}
With the upper bound of $\A(\BV)$ in~\eqref{eq:AVdelta}, the lower bound of $\A(\BU)$ in \eqref{eq:AUdelta}, and~\eqref{eq:I_0}, we get
$\Real(I_{01}) \geq \frac{\delta^2 d_0^2}{8}.$

Now let us give an upper bound for $\Real(I_{02})$, 
\begin{eqnarray*}
\| I_{02} \| & \leq & \|\A^*(\be)\| \|\BU + 2\BV\|_* = \|\A^*(\be)\| \sum_{i=1}^s\|\underbrace{\BU_i + 2\BV_i}_{\text{rank-2}}\|_* \\
& \leq & \sqrt{2}\|\A^*(\be)\| \sum_{i=1}^s \|\BU_i + 2\BV_i\|_F \\
& \leq & \sqrt{2s}\|\A^*(\be)\| \|\BU + 2\BV\|_F  \leq 2\sqrt{s}\delta d_0 \|\A^*(\be)\|
\end{eqnarray*}
where $\|\cdot\|$ and $\|\cdot\|_*$ are a pair of dual norms and 
\begin{equation*}
\|\BU + 2\BV\|_F \leq \|\BU + \BV\|_F + \|\BV\|_F \leq \delta d_0+ 0.2\delta d_0 \leq 1.2\delta d_0.
\end{equation*}
Combining the estimation of $\Real(I_{01})$ and $\Real(I_{02})$ above leads to 
\begin{equation*}
\Real( \lag \nabla F_{\bh}, \Dh\rag + \lag \nabla F_{\bx}, \Dx\rag)
\geq  \frac{\delta^2 d_0^2}{8} - 2\sqrt{s}\delta d_0 \|\A^*(\be)\|.
\end{equation*}
\end{proof}

\begin{lemma}{\bf (Local Regularity for $G(\bh,\bx)$}\label{lem:regG}
For any $(\bh, \bx) \in \Kd \bigcap \Keps$ with $\epsilon \leq \frac{1}{15}$ and $\frac{9}{10}d_{0} \leq d \leq \frac{11}{10}d_{0}$, $\frac{9}{10}d_{i0} \leq d_i \leq \frac{11}{10}d_{i0}$,   the following inequality holds uniformly
\begin{equation}
\label{eq:G_regularity}
\Real\lp\lag \nabla G_{\bh_i}, \Dh_i \rag + \lag \nabla G_{\bx_i}, \Dx_i \rag\rp  \geq 2\delta_0\sqrt{ \rho G_i(\bh_i, \bx_i)} = \frac{\delta}{5}\sqrt{ \rho G_i(\bh_i, \bx_i)},
\end{equation}
where $\rho \geq d^2 + 2\|\be\|^2.$ Immediately, we have
\begin{equation}
\label{eq:G_regularity_demix}
\Real\lp\lag \nabla G_{\bh}, \Dh \rag + \lag \nabla G_{\bx}, \Dx \rag\rp =\sum_{i=1}^s\Real\lp\lag \nabla G_{\bh_i}, \Dh_i \rag + \lag \nabla G_{\bx_i}, \Dx_i \rag\rp  \geq {\frac{\delta}{5}}\sqrt{ \rho G(\bh, \bx)}.
\end{equation}
\end{lemma}
\begin{remark}
For the local regularity condition for $G(\bh, \bx)$, we use the results from~\cite{LLSW16} when $s=1$. This is because each component $G_i(\bh,\bx)$ only depends on $(\bh_i,\bx_i)$ by definition and thus the lower bound of
$\Real\lp\lag \nabla G_{\bh_i}, \Dh_i \rag + \lag \nabla G_{\bx_i}, \Dx_i \rag\rp$ is completely determined by $(\bh_i,\bx_i)$ and $\delta_0$, and is independent of $s$.

\end{remark}
\begin{proof}
For each $i:1\leq i\leq s$, $\nabla G_{\bh_i}$ (or $\nabla G_{\bx_i}$) only depends on $\bh_i$ (or $\bx_i$) and there holds
\begin{equation*}
\Real\lp\lag \nabla G_{\bh_i}, \Dh_i \rag + \lag \nabla G_{\bx_i}, \Dx_i \rag\rp  \geq 2\delta_0\sqrt{ \rho G_i(\bh_i, \bx_i)} = \frac{\delta}{5}\sqrt{ \rho G_i(\bh_i, \bx_i)},
\end{equation*}
which follows exactly from Lemma 5.17 in~\cite{LLSW16}. For~\eqref{eq:G_regularity_demix}, by definition of $\nabla G_{\bh}$ and $\nabla G_{\bx}$ in~\eqref{def:grad}, 
\begin{align*}
\Real\lp\lag \nabla G_{\bh}, \Dh \rag + \lag \nabla G_{\bx}, \Dx \rag\rp & = \sum_{i=1}^s\Real\lp\lag \nabla G_{\bh_i}, \Dh_i \rag + \lag \nabla G_{\bx_i}, \Dx_i \rag\rp  \\
&\geq  \frac{\delta}{5} \sum_{i=1}^s\sqrt{ \rho G_i(\bh_i, \bx_i)} \geq \frac{\delta}{5}\sqrt{\rho G(\bh,\bx)}
\end{align*}
where $G(\bh,\bx) = \sum_{i=1}^s G_i(\bh_i,\bx_i).$

\end{proof}

\begin{lemma}{\bf (Proof of the Local Regularity Condition)}
\label{lem:reg}
Conditioned on~\eqref{eq:rip}, for the objective function $\tF(\bh,\bx)$ in~\eqref{def:FG}, there exists a positive constant $\omega$ such that
\begin{equation}\label{eq:regF}
\|\nabla \tF(\bh, \bx)\|^2 \geq \omega \left[ \tF(\bh, \bx) - c \right]_+
\end{equation}
with $c =  \|\be\|^2 + 2000s \|\A^*(\be)\|^2$ and $\omega = \frac{d_0}{7000}$ for all $(\bh, \bx) \in \Kint$. 
Here we set  $\rho \geq d^2 + 2\|\be\|^2.$
\end{lemma}

\begin{proof}
Following from Lemma~\ref{lem:regF} and Lemma~\ref{lem:regG}, we have
\begin{eqnarray*}
\Real( \lag \nabla F_{\bh}, \Dh\rag + \lag \nabla F_{\bx}, \Dx\rag)
& \geq & \frac{\delta^2 d_0^2}{8} - 2\sqrt{s}\delta d_0 \|\A^*(\be)\| \\
 \Real( \lag \nabla G_{\bh}, \Dh \rag + \lag \nabla G_{\bx}, \Dx \rag) 
 & \geq & \frac{\delta d}{5} \sqrt{ G(\bh, \bx)} \geq \frac{9\delta d_0}{50}\sqrt{G(\bh, \bx)}
\end{eqnarray*}
for all $(\bh, \bx) \in \Kint$ where $\rho \geq d^2 + 2\|\be\|^2  \geq d^2$ and $\frac{9}{10}d_0 \leq d \leq \frac{11}{10}d_0$.
Adding them together gives $\Real\left (\lag  \nabla \tF_{\bh}, \Dh \rag + \lag \nabla \tF_{\bx}, \Dx \rag\right)$ on the left side. Moreover, Cauchy-Schwarz inequality implies
\begin{equation*}
\Real\left (\lag  \nabla \tF_{\bh}, \Dh \rag + \lag \nabla \tF_{\bx}, \Dx \rag\right) \leq 4\delta\sqrt{d_0} \| \nabla \tF(\bh, \bx)\|
\end{equation*}
where both $\|\Dh\|^2$ and $\|\Dx\|^2$ are bounded by $8\delta^2d_0$ in  Lemma~\ref{lem:DxDh} since
\begin{equation*}
\|\Dh\|^2 = \sum_{i=1}^s \|\Dh_i\|^2 \leq \sum_{i=1}^s (7.5\delta_i^2 + 2.88\delta_0^2) d_{i0} \leq 8\delta^2 d_0.
\end{equation*}
Therefore, 
\begin{equation}
\frac{\delta^2 d_0^2}{8} + \frac{  9\delta d_0\sqrt{ G(\bh, \bx)}}{50} - 2\sqrt{s}\delta d_0 \|\A^*(\be) \|
\leq 4 \delta \sqrt{d_0}  \| \nabla \tF(\bh, \bx)\|. \label{eq:nabla-tF}
\end{equation}

Dividing both sides of~\eqref{eq:nabla-tF} by $\delta d_0$, we obtain
\begin{eqnarray*}
\frac{4}{\sqrt{d_0}} \|\nabla \tF(\bh, \bx)\| & \geq & 
 \frac{\delta d_0}{12}  + \frac{9}{50}\sqrt{G(\bh, \bx)} + \frac{\delta d_0 }{24} - 2\sqrt{s}\|\A^*(\be)\| \\
 & \geq & \frac{1}{6\sqrt{6}}[\sqrt{F_0(\bh,\bx)} + \sqrt{G(\bh, \bx)}] + \frac{\delta d_0}{24} - 2\sqrt{s}\|\A^*(\be)\| 
\end{eqnarray*}
where the Local RIP condition~\eqref{eq:rip} implies $F_0(\bh, \bx) \leq \frac{3}{2}\delta^2 d_0^2$ and hence $\frac{\delta d_0}{12} \geq \frac{1}{6\sqrt{6}}\sqrt{F_0(\bh, \bx)}$, where $F_0(\bh,\bx)$ is defined in~\eqref{def:F0}.  

Note that~\eqref{eq:cross} gives
\begin{equation}\label{eq:AEHX}
\sqrt{2\left[ \Real(\lag \A^*(\be), \BX - \BX_0\rag) \right]_+} \leq \sqrt{ 2\sqrt{2s} \|\A^*(\be)\| \delta d_0} \leq \frac{\sqrt{6}\delta d_0}{4} + \frac{4\sqrt{s}}{\sqrt{6}}\|\A^*(\be)\|.
\end{equation}
By~\eqref{eq:AEHX} and $\tF(\bh, \bx) - \|\be\|^2 \leq F_0(\bh, \bx) + 2 [\Real(\lag \A^*(\be), \BX - \BX_0\rag)]_+ + G(\bh, \bx)$, there holds
\begin{eqnarray*}
\frac{4}{\sqrt{d_0}} \|\nabla \tF(\bh, \bx)\| 
& \geq & \frac{1}{6\sqrt{6}} \Big[ \left(\sqrt{F_0(\bh, \bx)} +\sqrt{2\left[ \Real(\lag \A^*(\be), \BX-\BX_0\rag) \right]_+} + \sqrt{G(\bh, \bx)}\right) \\
&& + \frac{\delta d_0}{24} - \frac{1}{6\sqrt{6}} \left( \frac{\sqrt{6}\delta d_0}{4} + \frac{4\sqrt{s}}{\sqrt{6}}\|\A^*(\be)\|\right) - 2\sqrt{s}\|\A^*(\be)\|\\
& \geq & \frac{1}{6\sqrt{6}} \left[ \sqrt{ \left[\tF(\bh, \bx) - \|\be\|^2\right]_+} -  \sqrt{1000s}\|\A^*(\be)\|\right].
\end{eqnarray*}

For any nonnegative real numbers $a$ and $b$, we have
$[\sqrt{(x - a)_+} - b ]_+ + b \geq \sqrt{(x - a)_+} $
and it implies
\begin{equation*}
( x - a)_+ \leq 2 ( [\sqrt{(x - a)_+} - b ]_+^2 + b^2) \Longrightarrow [\sqrt{(x - a)_+} - b ]_+^2  \geq \frac{(x - a)_+}{2} - b^2.
\end{equation*}
Therefore, by setting $ a = \|\be\|^2$ and $b = \sqrt{1000s}\|\A^*(\be)\|$, there holds
\begin{eqnarray*}
\|\nabla \tF(\bh, \bx)\|^2 
& \geq & \frac{d_0}{3500} \left[ \frac{\tF(\bh, \bx) - \|\be\|^2 }{2} -  1000s \|\A^*(\be)\|^2 \right]_+ \\
& \geq & \frac{d_0}{7000} \left[ \tF(\bh, \bx) - (\|\be\|^2 + 2000s \|\A^*(\be)\|^2) \right]_+.
\end{eqnarray*}
\end{proof}
\subsection{Local smoothness}
\label{s:smooth}

\begin{lemma}
Conditioned on~\eqref{eq:rip},~\eqref{eq:robust} and~\eqref{eq:A-UPBD}, for any $\bz : = (\bh, \bx)\in\CC^{(K+N)s}$ and $\bw : = (\bu, \bv)\in\CC^{(K+N)s}$ such that $\bz$ and $\bz+\bw \in \Keps \cap \KF$, there holds
\begin{equation*}
\| \nabla\tF(\bz + \bw) - \nabla\tF(\bz) \| \leq C_L \|\bw\|,
\end{equation*}
with
\begin{equation*}
C_L \leq  \left(10\|\A\|^2d_0 + \frac{2\rho}{\min d_{i0}} \left( 5 + \frac{2L}{\mu^2} \right)\right)
\end{equation*}
where $\rho \geq d^2 + 2\|\be\|^2$ 
and $\|\A\| \leq \sqrt{s(N\log(NL/2) + (\gamma+\log s)\log L)}$ holds with probability at least $1 - L^{-\gamma}$ from Lemma~\ref{lem:A-UPBD}. 

In particular, $L = \mathcal{O}((\mu^2 + \sigma^2)s(K + N)\log^2 L)$
and $\|\be\|^2 = \mathcal{O}(\sigma^2d_0^2)$ follows from $\|\be\|^2 \sim \frac{\sigma^2d_0^2}{2L} \chi^2_{2L}$ and~\eqref{ineq:bern}. Therefore, $C_L$ can be simplified to
\begin{equation*}
C_L = \mathcal{O}(d_0s\kappa(1 + \sigma^2)(K + N)\log^2 L )
\end{equation*}
by choosing $\rho \approx d^2 + 2\|\be\|^2.$
\end{lemma}

\begin{proof}
By \prettyref{lem:betamu}, we know that both $\vct{z}=(\vct{h}, \vct{x})$ and $\vct{z}+\vct{w}=(\vct{h}+\vct{u}, \vct{x}+\vct{v}) \in \Kint$. 
Note that 
\[
\nabla \tF = (\nabla \tF_{\bh}, \nabla \tF_{\bx}) = (\nabla F_{\bh} + \nabla G_{\bh}, \nabla F_{\bx} + \nabla G_{\bx}),
\]
where~\eqref{eq:WFh},~\eqref{eq:WFx},~\eqref{eq:WGh} and~\eqref{eq:WGx} give $\nabla F_{\bh},\nabla F_{\bx},\nabla G_{\bh}$ and $\nabla G_{\bx}$.
It suffices to find out the Lipschitz constants for all of those four functions.
\paragraph{Step 1:} We first estimate the Lipschitz constant for $\nabla F_{\bh}$ and the result can be applied to $\nabla F_{\bx}$ due to symmetry.
\begin{align*}
\nabla F_{\bh}(\bz + \bw) - \nabla F_{\bh}(\bz) 
&= \A^*\A( \CH(\bh+\bu,\bx +\bv)) (\bx + \bv) - \left[ \A^*\A(\CH(\bh,\bx))\bx + \A^*(\by) \bv\right] \\
& = \A^*( \A(\CH(\bh+\bu, \bx+ \bv) - \CH(\bh,\bx)) )(\bx + \bv) \\
& \quad + \A^*\A( \CH(\bh,\bx) - \CH(\bh_0,\bx_0) )\bv - \A^*(\be) \bv \\
& = \A^*( \A(   \CH(\bh+\bu,\bv) + \CH(\bu, \bx) ) )(\bx + \bv) \\
& \quad + \A^*\A( \CH(\bh,\bx) - \CH(\bh_0,\bx_0) )\bv - \A^*(\be) \bv.
\end{align*}
Note that $\|\CH(\bh, \bx)\|_F\leq \sqrt{\sum_{i=1}^s \|\bh_i\|^2\|\bx_i\|^2} \leq \|\bh\|\|\bx\|$ and $\bz, \bz+\bw \in \Kd$ directly implies 
\begin{equation*}
\|\CH(\bu,\bx) + \CH(\bh+\bu,\bv)\|_F \leq \|\vct{u}\| \|\vct{x}\| + \|\vct{h}+\vct{u}\|\|\vct{v}\| \leq 2\sqrt{d_0} (\|\bu\| + \|\bv\|)
\end{equation*}
where $\|\bh + \bu\| \leq 2\sqrt{d_0}.$
Moreover,~\eqref{eq:rip} implies
\begin{equation*}
\|\CH(\bh,\bx) - \CH(\bh_0,\bx_0)\|_F \leq \epsilon d_0
\end{equation*}
since $\bz\in\Kint.$
Combined with $\|\A^*(\be)\| \leq \eps d_0$ in~\eqref{eq:robust} and $\|\bx + \bv\| \leq 2\sqrt{d_0}$, we have 
\begin{eqnarray}
\|\nabla F_{\bh}(\bz + \bw) - \nabla F_{\bh}(\bz) \| 
& \leq & 4d_0 \|\A\|^2(\|\bu\| + \|\bv\|) + \eps d_0 \|\A\|^2 \|\bv\| + \eps d_0 \|\bv\|  \nonumber \\
& \leq & 5d_0 \|\A\|^2 ( \|\bu\| + \|\bv\|). 
\label{eq:LipFh}
\end{eqnarray}
Due to the symmetry between $\nabla F_{\bh}$ and $\nabla F_{\bx}$, we have,
\begin{equation}
\label{eq:LipFx}
\|\nabla F_{\bx}(\bz + \bw) - \nabla F_{\bx}(\bz) \| \leq 5d_0\|\A\|^2 ( \|\bu\| + \|\bv\|).\end{equation}
In other words, 
\begin{equation*}
\|  \nabla F(\bz + \bw) - \nabla F(\bz) \| \leq 5\sqrt{2}d_0 \|\A\|^2(\|\bu\| + \|\bv\|) \leq 10d_0\|\A\|^2\|\bw\|
\end{equation*}
where $\|\bu\| + \|\bv\| \leq \sqrt{2}\|\bw\|.$
 
\paragraph{Step 2:} We estimate the upper bound of $\|\nabla G_{\bx_i}(\bz_i + \bw_i) - \nabla G_{\bx_i}(\bz_i)\|$. Implied by Lemma~5.19 in~\cite{LLSW16}, 
we have
\begin{align}
\| \nabla G_{\bx_i}(\bz_i + \bw_i) - \nabla G_{\bx_i}(\bz_i) \| \leq \frac{5d_{i0} \rho}{d_i^2} \|\bv_i\|.   \label{eq:LipGx}
\end{align}

\paragraph{Step 3:} We estimate the upper bound of $\|\nabla G_{\bh_i}(\bz + \bw) - \nabla G_{\bh_i}(\bz)\|$. Denote 
\begin{align*}
\nabla G_{\bh_i}(\bz + \bw) - \nabla G_{\bh_i}(\bz) &= \underbrace{\frac{\rho}{2d_i}\left[G'_0\left(\frac{\|\bh_i + \bu_i\|^2}{2d_i}\right) (\bh_i + \bu_i) - G'_0\left(\frac{\|\bh_i\|^2}{2d_i}\right) \bh_i\right] }_{\vct{j}_1}
\\
& \underbrace{+ \frac{\rho L}{8d_i\mu^2 }\sum_{l=1}^L \left[G'_0\left(\frac{L|\bb_l^*(\bh_i + \bu_i)|^2}{8d_i\mu^2}\right) \bb_l^*(\bh_i + \bu_i) - G'_0\left(\frac{L|\bb_l^*\bh_i|^2}{8d_i\mu^2}\right) \bb_l^*\bh_i \right]\bb_l}_{\vct{j}_2}.
\end{align*}
Following the same estimation of $\vct{j}_1$ and $\vct{j}_2$ in Lemma 5.19 of~\cite{LLSW16}, we have
\begin{equation}
\label{eq:LipG}
\|\vct{j}_1\| \leq \frac{5d_{i0} \rho}{d_i^2} \|\bu_i\|, \quad
\|\vct{j}_2\| 
\leq \frac{3\rho Ld_{i0}}{2d_i^2\mu^2}\|\bu_i\|.
\end{equation}
Therefore, combining~\prettyref{eq:LipGx} and \prettyref{eq:LipG} gives
\begin{align*}
\|\nabla G(\bz + \bw) - \nabla G(\bz)\| & = \sqrt{\sum_{i=1}^s \left(\|  \nabla G_{\bh_i}(\bz + \bw) - \nabla G_{\bh_i}(\bz) \|^2 + \|  \nabla G_{\bx_i}(\bz + \bw) - \nabla G_{\bx_i}(\bz) \|^2\right)} \\
& \leq \max\left\{\frac{5d_{i0} \rho}{d_i^2} +  \frac{3\rho Ld_{i0}}{2d_i^2\mu^2}\right\} \sqrt{\sum_{i=1}^s\|\bu_i\|^2} 
+ \max\left\{\frac{5d_{i0} \rho}{d_i^2}\right\} \sqrt{\sum_{i=1}^s\|\bv_i\|^2} \\
& \leq \max\left\{\frac{5d_{i0} \rho}{d_i^2} +  \frac{3\rho Ld_{i0}}{2d_i^2\mu^2}\right\} \|\bu\|
+ \max\left\{\frac{5d_{i0} \rho}{d_i^2}\right\} \|\bv\| \\
& \leq  \frac{2\rho}{\min d_{i0}} \left( 5 + \frac{2L}{\mu^2} \right)\|\bw\|. 
\end{align*}
In summary, the Lipschitz constant $C_L$ of $\tF(\bz)$ has an upper bound as follows:
\begin{align*}
\|\nabla \tF(\bz + \bw) - \nabla \tF(\bz)\| & \leq  \|\nabla F(\bz + \bw) - \nabla F(\bz)\| + \|\nabla G(\bz + \bw) - \nabla G(\bz)\| \\
& \leq \left(10\|\A\|^2d_0 + \frac{2\rho}{\min d_{i0}} \left( 5 + \frac{2L}{\mu^2} \right)\right) \|\bw\|.
\end{align*}
\end{proof}


\subsection{Robustness condition and spectral initialization}
\label{s:init}
In this section, we will prove the robustness condition~\eqref{eq:robust} and also Theorem~\ref{thm:init}. To prove~\eqref{eq:robust}, it suffices to show the following lemma, which is a more general version of~\eqref{eq:robust}.
\begin{lemma}\label{lem:denoise}
Consider a sequence of Gaussian independent random variable $\bc = (c_1, \cdots, c_L)\in\CC^L$ where $c_l\sim \CN(0, \frac{\lambda_i^2}{L})$ with $\lambda_i \leq \lambda$. Moreover, we assume $\A_i$ in~\eqref{def:Ai} is independent of $\bc$. Then there holds
\begin{equation*}
\|\A^*(\bc) = \|\max_{1\leq i\leq s}\|\A_i^*(\bc )\| \leq  \xi
\end{equation*}
with probability at least $1 - L^{-\gamma}$
if $L \geq C_{\gamma + \log(s)}( \frac{\lambda}{\xi} +\frac{\lambda^2}{\xi^2} )\max\{ K,N \}\log L/\xi^2.$
\end{lemma}
\begin{proof}
It suffices to show that $\max_{1\leq i\leq s}\|\A_i^*(\bc)\| \leq \xi$. For each fixed $i:1\leq i\leq s$, 
\begin{equation*}
\A_i^*(\bc) = \sum_{l=1}^L c_l\bb_l\ba_{il}^*
\end{equation*}
The key is to apply the matrix Bernstein inequality~\eqref{thm:bern} and we need to estimate $\|\CZ_l\|_{\psi_1}$, and the variance of $\sum_{l=1}^L \CZ_l.$
For each $l$, $\| c_l \bb_l\ba_{il}^*\|_{\psi_1} \leq \frac{\lambda \sqrt{KN}}{L}$ follows from~\eqref{lemma:psi}. Moreover, the variance of $\A_i^*(\bc)$ is bounded by $\frac{\lambda^2 \max\{K,N\}}{L}$ since
\begin{align*}
\E [ \A_i^*(\bc) (\A_i^*(\bc) )^* ] & = \sum_{l=1}^L \E(|c_l|^2 \|\ba_{il}\|^2)\bb_l\bb_l^* = \frac{N}{L} \sum_{l=1}^L \lambda_l^2\bb_l\bb_l^* \preceq \frac{\lambda^2 N}{L}, \\
\E [ (\A_i^*(\bc) )^* (\A_i^*(\bc))  ] & = \sum_{l=1}^L \|\bb_l\|^2 \E(|c_l|^2 \ba_{il}\ba_{il}^*) = \frac{K}{L^2}  \sum_{l=1}^L \lambda_i^2 \I_N \preceq \frac{\lambda^2 K}{L}.
\end{align*}
Letting $t = \gamma\log L$ and applying~\eqref{thm:bern} leads to
\begin{equation*}
\|\A_i^*(\bc)\| \leq C_0\max\left\{ \frac{\lambda\sqrt{KN}\log^2 L}{L},\sqrt{\frac{C_{\gamma}\lambda^2\max\{K,N\}\log L}{L}}\right\} \leq \xi.
\end{equation*}
Therefore, by taking the union bound over $1\leq i\leq s$, 
\begin{equation*}
\|\A_i^*(\bc)\| \leq \xi
\end{equation*}
with probability at least $1 - L^{-\gamma}$ if $L\geq C_{\gamma+\log(s)} (\frac{\lambda}{\xi} +\frac{\lambda^2}{\xi^2} )\max\{K,N\}\log^2L$.
\end{proof}

The robustness condition is an immediate result of Lemma~\ref{lem:denoise} by setting $\xi = \frac{\eps d_0}{10\sqrt{2}s\kappa}$ and $\lambda = \sigma d_0.$
\begin{corollary}\label{cor:Ae}{\bf [Robustness Condition]}
For $\be \sim \CN(\bzero, \frac{\sigma^2d_0^2}{L}\I_L) $
\begin{equation*}
\|\A_i^*(\be)\| \leq \frac{\eps d_0}{10\sqrt{2} s\kappa}, \quad \forall 1\leq i\leq s
\end{equation*}
with  probability at least $1 - L^{-\gamma}$ if $L \geq C_{\gamma}(\frac{ s^2\kappa^2 \sigma^2}{\eps^2} + \frac{s\kappa\sigma }{\eps})\max\{K, N\} \log L$.
\end{corollary}

\begin{lemma}\label{lem:Ay-hx}
For $\be \sim \CN(\bzero, \frac{\sigma^2d_0^2}{L}\I_L) $, there holds
\begin{equation}\label{eq:Ay-hx}
\| \A_i^*(\by) - \bh_{i0}\bx_{i0}^* \| \leq \xi d_{i0}, \quad \forall 1\leq i\leq s
\end{equation}
with probability at least $1 - L^{-\gamma}$ if $L \geq C_{\gamma+ \log (s)}s\kappa^2 (\mu^2_h + \sigma^2) \max\{K,  N\} \log L /\xi^2.$ 
\end{lemma}
\begin{remark}
The success of the initialization algorithm completely relies on the lemma above. As mentioned in Section~\ref{s:thm}, $\E(\A_i^*(\by)) = \bh_{i0}\bx_{i0}^*$ and Lemma~\ref{eq:Ay-hx} confirms that $\A_i^*(\by)$ is close to $\bh_{i0}\bx_{i0}^*$ in operator norm and hence the spectral method is able to give us a reliable initialization.
\end{remark}

\begin{proof}
Note that
\begin{eqnarray*}
\A_i^*(\by) & = & \A_i^*\A_i(\bh_{i0}\bx_{i0}^*) + \A_i^*(\bw_i)
\end{eqnarray*}
where
\begin{equation}\label{eq:AA}
\bw_i = \by - \A_i(\bh_{i0}\bx_{i0}^*) = \sum_{j\neq i} \A_j(\bh_{j0}\bx_{j0}^*) + \be
\end{equation}
is independent of $\A_i.$ The proof consists of two parts: 1. show that $\|\A_i^*\A_i(\bh_{i0}\bx_{i0}^*) - \bh_{i0}\bx_{i0}^*\| \leq \frac{\xi d_{i0}}{2}$; 2. prove that $\|\A_i^*(\bw_i)\| \leq \frac{\xi d_{i0}}{2}$.

\paragraph{Part I:}
Following from the definition of $\A_i$ and $\A_i^*$ in~\eqref{def:Ai}, 
\begin{equation*}
\A_i^*\A_i(\bh_{i0}\bx_{i0}^*) - \bh_{i0}\bx_{i0}^* = \sum_{l=1}^L \underbrace{\bb_l\bb_l^*\bh_{i0}\bx_{i0}^*(\ba_{il}\ba_{il}^* - \I_N)}_{\text{defined as } \CZ_l }.
\end{equation*}
where $\BB^*\BB = \I_K.$ The sub-exponential norm of $\CZ_l$ is bounded by
\begin{equation*}
\|\CZ_l\|_{\psi_1} \leq \max_{1\leq l\leq L}\|\bb_l\| |\bb_l^*\bh_{i0}| \| (\ba_{il}\ba_{il}^* - \I_N) \bx_{i0}\|_{\psi_1} \leq \frac{\mu\sqrt{KN}d_{i0}}{L}
\end{equation*}
where $\|\bb_l\| = \sqrt{\frac{K}{L}}$, $\max_l|\bb^*_l\bh_{i0}|^2 \leq \frac{\mu^2 d_{i0}}{L}$ and $\| (\ba_{il}\ba_{il}^* - \I_N) \bx_{i0}\|_{\psi_1} \leq \sqrt{Nd_{i0}}$ follows from~\eqref{lem:11JB}.

We proceed to estimate the variance of $\sum_{l=1}^L \CZ_l$ by using~\eqref{lem:9JB} and~\eqref{lem:12JB}:
\begin{eqnarray*}
\left\| \sum_{l=1}^L\E(\CZ_l\CZ_l^*)\right\| & = & \left\| \sum |\bb_l^*\bh_{i0}|^2 \bx_{i0}^*\E(\ba_{il} \ba_{il}^* - \I_N)^2\bx_{i0}\bb_l\bb_l^*\right\|  \leq \frac{\mu^2N d_{i0}^2}{L}, \\
\left\| \sum_{l=1}^L\E(\CZ_l^*\CZ_l)\right\| & = & \frac{K}{L}\left\| \sum_{l=1}^L |\bb_l^*\bh_{i0}|^2 \E\left[ (\ba_{il}\ba_{il}^* - \I_N)\bx_{i0}\bx_{i0}^*(\ba_{il}\ba_{il}^* - \I_N)\right] \right\|  \leq \frac{Kd_{i0}^2}{L}.
\end{eqnarray*}
Therefore, the variance of $\sum_{l=1}^L\CZ_l$ is bounded by $\frac{\max\{K, \mu^2_hN\}d_{i0}^2}{L}$.
By applying matrix Bernstein inequality~\eqref{thm:bern} and taking the union bound over all $i$, we prove that
\begin{equation*}
\|\A_i^*\A_i(\bh_{i0}\bx_{i0}^*) - \bh_{i0}\bx_{i0}^*\| \leq \frac{\xi d_{i0}}{2}, \quad \forall 1\leq i\leq s
\end{equation*}
holds with probability at least $1 -L^{-\gamma+1}$ if $L \geq C_{\gamma+\log(s)} \max\{K,\mu_h^2N\}\log L/\xi^2.$

\paragraph{Part II:} For each $1\leq l\leq L$, the $l$-th entry of $\bw_i$ in~\eqref{eq:AA}, i.e.,
$(\bw_i )_l = \sum_{j\neq i} \bb_{l}^*\bh_{j0}\bx_{j0}^*\ba_{jl} + e_l$,
is independent of $\bb^*_l\bh_{i0}\bx_{i0}^*\ba_{il}$ and obeys $\mathcal{C}\mathcal{N}(0, \frac{\sigma_{il}^2}{L})$. Here
\begin{eqnarray*}
\sigma_{il}^2 & = & L\E|(\bw_i)_l|^2 =  L\sum_{j\neq i} |\bb_{l}^*\bh_{j0}|^2 \| \bx_{j0} \|^2  + \sigma^2\|\BX_0\|_F^2 \\
& \leq & \mu_h^2 \sum_{j\neq i}\|\bh_{j0}\|^2 \|\bx_{j0}\|^2 + \sigma^2\|\BX_0\|_F^2  \leq (\mu_h^2 + \sigma^2) \|\BX_0\|_F^2.
\end{eqnarray*}
This gives $\max_{i,l} \sigma_{il}^2\leq (\mu^2_h + \sigma^2) \|\BX_0\|_F^2.$
Thanks to the  independence between $\bw_i$ and $\A_i$, applying Lemma~\ref{lem:denoise} results in
\begin{equation}\label{eq:AA2}
\|\A_i^*(\bw_i)\| \leq \frac{\xi d_{i0}}{2}
\end{equation}
with probability $1 - L^{-\gamma + 1}$ if
$L \geq C\max\left( \frac{(\mu_h^2 + \sigma^2) \|\BX_0\|_F^2 }{\xi^2d_{i0}^2},  \frac{\sqrt{\mu^2_h + \sigma^2}\|\BX_0\|_F }{\xi d_{i0}} \right)\max\{K,N\}\log L.$

\vskip0.5cm

Therefore, combining~\eqref{eq:AA} with~\eqref{eq:AA2}, we get
\begin{equation*}
\|\A_i^*(\by) - \bh_{i0}\bx_{i0}^*\| \leq \|\A_i^*\A_i(\bh_{i0}\bx_{i0}^*) - \bh_{i0}\bx_{i0}^*\|  + \|\A_i^*(\bw_i)\| \leq \xi d_{i0}
\end{equation*}
for all $1\leq i\leq s$ with probability at least $1 - L^{-\gamma+1}$ if
\begin{equation*}
L \geq C_{\gamma+\log (s)}(\mu_h^2 + \sigma^2)s \kappa^2  \max\{K,N\}\log L/\xi^2
\end{equation*}
where $\|\BX_0\|_F/d_{i0} \leq \sqrt{s}\kappa.$
\end{proof}

Before moving to the proof of Theorem~\ref{thm:init}, we  introduce a property about the projection onto a closed convex set. 
\begin{lemma}[Theorem 2.8 in~\cite{RM11AP}]\label{lem:KMC}
Let $Q := \{ \bw\in\CC^K | \sqrt{L}\|\BB\bw\|_{\infty} \leq 2\sqrt{d}\mu \}$ 
be a closed nonempty convex set. There holds
\begin{equation*}
\Real( \lag \bz - \PP_Q(\bz) , \bw - \PP_Q(\bz) \rag ) \leq 0, \quad \forall \, \bw \in Q, \bz\in \CC^K 
\end{equation*}
where $\PP_{Q}(\bz)$ is the projection of $\bz$ onto $Q$.
\end{lemma}
With this lemma, we can easily see 
\begin{equation}\label{eq:nonexp}
\|\bz - \bw\|^2 = \|\bz - \PP_Q(\bz)\|^2 + \|\PP_Q(\bz) - \bw\|^2 + 2\Real(\lag \bz - \PP_Q(\bz), \PP_Q(\bz) - \bw \rag) \geq \|\PP_Q(\bz) - \bw\|^2
\end{equation}
for all $\bz\in\CC^K$ and $\bw\in Q$. It means that projection onto nonempty closed convex set is non-expansive. 
Now we present the proof of Theorem~\ref{thm:init}.

\begin{proof}[\bf Proof of Theorem~\ref{thm:init}]

By choosing $L \geq C_{\gamma+\log (s)}(\mu_h^2 + \sigma^2)s^2 \kappa^4  \max\{K,N\}\log L/\eps^2$, we have 
\begin{equation}\label{eq:Ay-hx2}
\|\A_i^*(\by) - \bh_{i0}\bx_{i0}^*\| \leq \xi d_{i0}, \quad \forall 1\leq i\leq s
\end{equation}
where $\xi = \frac{\eps}{10\sqrt{2s}\kappa}$.

By applying the triangle inequality to~\eqref{eq:Ay-hx2}, it is easy to see that
\begin{equation}\label{eq:hd}
(1 - \xi)d_{i0} \leq d_i \leq (1 + \xi)d_{i0}, \quad |d_{i} - d_{i0}| \leq \xi d_{i0} \leq \frac{\eps d_{i0}}{10\sqrt{2s}\kappa} < \frac{d_{i0}}{10}, 
\end{equation}
which gives $\frac{9}{10}d_{i0} \leq d_i \leq \frac{11}{10}d_{i0}$ where $d_i = \|\A_i^*(\by)\|$ according to Algorithm~\ref{Initial}. 

\paragraph{Part I: Proof of $(\bu^{(0)},\bv^{(0)})\in \frac{1}{\sqrt{3}}\Kd \cap \frac{1}{\sqrt{3}}\Kmu$}  Note that $\bv_i^{(0)} = \sqrt{d_i}\|\hat{\bx}_{i0}\| = \sqrt{d_i}$ where $\hat{\bx}_{i0}$ is the leading right singular vector of $\A_i^*(\by).$ Therefore, 
\begin{equation*}
\| \bv_i^{(0)} \| = \sqrt{d_i} \|\hbx_{i0}\| =  \sqrt{d_i} \leq \sqrt{(1 + \xi)d_{i0}} \leq \frac{2}{\sqrt{3}}\sqrt{d_{i0}}, \quad \forall 1\leq i\leq s
\end{equation*}
which implies $\{\bv_i^{(0)}\}_{i=1}^s \in \frac{1}{\sqrt{3}}\Kd.$ 

Now we will prove that $\bu_i^{(0)}\in \frac{1}{\sqrt{3}}\Kd\cap\frac{1}{\sqrt{3}}\Kmu$ by Lemma~\ref{lem:KMC}.
By Algorithm~\ref{Initial}, $\bu_i^{(0)}$ is the minimizer to the function $f(\bz) = \frac{1}{2} \| \bz - \sqrt{d_i} \hbh_{i0} \|^2$ over $Q_i := \{ \bz | \sqrt{L}\|\BB\bz\|_{\infty} \leq 2\sqrt{d_i}\mu\}.$ Obviously, by definition, $\bu_i^{(0)}$ is the projection of $\sqrt{d_i} \hbh_{i0}$ onto $Q_i$. Note that $\bu_i^{(0)}\in Q_i$ implies $\sqrt{L}\|\BB\bu_i^{(0)}\|_{\infty} \leq 2\sqrt{d_i}\mu\leq 2\sqrt{(1+\xi)d_{i0}}\mu \leq \frac{4\sqrt{d_{i0}}\mu}{\sqrt{3}}$ and hence $\bu_i^{(0)}\in \frac{1}{\sqrt{3}} \Kmu.$

Moreover, due to~\eqref{eq:nonexp}, there holds 
\begin{equation}\label{eq:KMC}
\|\sqrt{d_i}\hbh_{i0} - \bw\|^2  \geq  \|\bu_i^{(0)} - \bw \|^2, \quad \forall \bw\in Q_i
\end{equation}
In particular, let $\bw = \bzero\in Q_i$ and immediately we have 
\begin{equation*}
\|\bu_i^{(0)}\|^2 \leq d_i \leq \frac{4}{3} \Longrightarrow \bu_i^{(0)}\in \frac{1}{\sqrt{3}}\Kmu.
\end{equation*}
In other words, $\{(\bu_i^{(0)}, \bv_i^{(0)})\}_{i=1}^s \in \frac{1}{\sqrt{3}}\Kd\cap \frac{1}{\sqrt{3}}\Kmu$. 

\paragraph{Part II: Proof of $(\bu^{(0)},\bv^{(0)})\in \MN_{\frac{2\eps}{5\sqrt{s}\kappa}}$}
We will show $\|\bu_i^{(0)}(\bv_i^{(0)})^* - \bh_{i0}\bx_{i0}^*\|_F \leq 4\xi d_{i0}$ for $1\leq i\leq s$ so that $\frac{\|\bu_i^{(0)}(\bv_i^{(0)})^* - \bh_{i0}\bx_{i0}^*\|_F}{d_{i0}} \leq \frac{2\eps}{5\sqrt{s}\kappa}$.

\vskip0.25cm

First note that $\sigma_j(\A_i^*(\by)) \leq \xi d_{i0}$ for all $j\geq 2$, which follows from Weyl's inequality~\cite{Stewart90} for singular values where $\sigma_j(\A_i^*(\by))$ denotes the $j$-th largest singular value of $\A_i^*(\by)$. Hence there holds
\begin{equation}\label{eq:dhx-hx}
\| d_i \hbh_{i0}\hbx_{i0}^* - \bh_{i0}\bx_{i0}^* \| \leq \|\A_i^*(\by) - d_i \hbh_{i0}\hbx_{i0}^* \| + \|\A_i^*(\by) - \bh_{i0}\bx_{i0}^* \| \leq 2\xi d_{i0}.
\end{equation}
On the other hand, for any $i$,
\begin{align*}
\left\| \left(\I_K - \frac{\bh_{i0}\bh_{i0}^*}{d_{i0}}\right)\hbh_{i0} \right\| 
&=   \left\| \left(\I_K - \frac{\bh_{i0}\bh_{i0}^*}{d_{i0}}\right)
\hbh_{i0}\hbx_{i0}^*\hbx_{i0}\hbh_{i0}^* \right\| \\
&=   \left\| \left(\I_K - \frac{\bh_{i0}\bh_{i0}^*}{d_{i0}}\right)\left[ 
\frac{1}{d_{i0}}(( \A_i^*(\by) - d_i \hbh_{i0}\hbx_{i0}^*) + \hbh_{i0}\hbx_{i0}^* - \frac{\bh_{i0}\bx_{i0}^*}{d_{i0}} \right] \hbx_{i0}\hbh_{i0}^* \right\| \\
&=  \frac{1}{d_{i0}} \| 
 \A_i^*(\by)  - \bh_{i0}\bx_{i0}^* \| + \left|\frac{d_i}{d_{i0}}-1\right| \leq 2\xi 
\end{align*} 
where $ (\I_K - \frac{\bh_{i0}\bh_{i0}^*}{d_{i0}}) \bh_{i0}\bx_{i0}^* = \bzero$ and $(\A_i^*(\by) - d_i \hbh_{i0}\hbx_{i0}^*)\hbx_{i0}\hbh_{i0}^* = \bzero$.
Therefore, we have
\begin{equation}\label{eq:h0-h}
\left\|  \hbh_{i0} -  \frac{\bh_{i0}^*\hbh_{i0}}{d_{i0}} \bh_{i0}  \right\| \leq 2\xi, 
\quad \|\sqrt{d_i} \hbh_{i0} - t_{i0} \bh_{i0} \| \leq 2\sqrt{d_i}\xi,
\end{equation}
where $t_{i0} = \frac{\sqrt{d_i}\bh_{i0}^*\hbh_{i0}}{d_{i0}}$ and $|t_{i0}| \leq \sqrt{d_i/d_{i0}} <\sqrt{2}$. 
If we  substitute $\bw$ by $t_{i0} \bh_{i0}\in Q_i$ into~\eqref{eq:KMC}, 
\begin{equation}\label{eq:h0-h-2}
\|\sqrt{d_i}\hbh_{i0} - t_{i0} \bh_{i0}\| \geq \| \bu_i^{(0)} - t_{i0} \bh_{i0}\|.
\end{equation}
where $t_{i0} \bh_{i0}\in Q_i$ follows from $\sqrt{L} |t_{i0}|\|\BB\bh_{i0}\|_{\infty} \leq |t_{i0}| \sqrt{d_{i0}}\mu_h \leq \sqrt{2d_{i0}}\mu$.

\vskip0.5cm
Now we are ready to estimate $\|\bu^{(0)}_i(\bv_i^{(0)})^* - \bh_{i0}\bx_{i0}^* \|_F$ as follows, 
\begin{align*}
\|\bu^{(0)}_i(\bv_i^{(0)})^* - \bh_{i0}\bx_{i0}^* \|_F
 & \leq 
\|\bu^{(0)}_i(\bv_i^{(0)})^* - t_{i0}\bh_{i0}(\bv_i^{(0)})^* \|_F + \|t_{i0}\bh_{i0}(\bv_i^{(0)})^* - \bh_{i0}\bx_{i0}^* \|_F \\
& \leq \underbrace{\|\bu_i^{(0)} - t_{i0}\bh_{i0}\| \|\bv_i^{(0)}\|}_{I_1} + 
\underbrace{\left\| \frac{d_i}{d_{i0}} \bh_{i0} \bh^*_{i0} \hbh_{i0}  \hbx_{i0}^* - \bh_{i0}\bx_{i0}^* \right\|_F}_{I_2}.
\end{align*}
Here $I_1\leq 2\xi d_i$ because $\|\bv_i^{(0)}\| = \sqrt{d_i}$ and $\|\bu_i^{(0)} - t_{i0}\bh_{i0}\| \leq 2\sqrt{d_i}\xi$ follows from~\eqref{eq:h0-h} and~\eqref{eq:h0-h-2}. For $I_2$, there holds
\begin{equation*}
I_2  = \left\| \frac{\bh_{i0}\bh_{i0}^*}{d_{i0}}
\left(d_i \hbh_{i0}  \hbx_{i0}^* - \bh_{i0}\bx_{i0}^*\right) \right\|_F \leq \|d_i \hbh_{i0}  \hbx_{i0}^* - \bh_{i0}\bx_{i0}^*\|_F \leq 2\sqrt{2}\xi d_{i0},
\end{equation*}
which follows from~\eqref{eq:dhx-hx}.
Therefore, 
\begin{align*}
\|\bu^{(0)}_i(\bv_i^{(0)})^* - \bh_{i0}\bx_{i0}^* \|_F
& \leq 2\xi d_i  + 2 \sqrt{2}\xi d_{i0} \leq 5\xi d_{i0}\leq \frac{2\eps d_{i0}}{5\sqrt{s}\kappa }.
\end{align*}
\end{proof}

\section*{Appendix}
\subsection*{Descent Lemma}
\begin{lemma}[Lemma 6.1 in~\cite{LLSW16}]\label{lem:DSL}
If $f(\bz, \bar{\bz})$ is a continuously differentiable real-valued function with two complex variables $\bz$ and $\bar{\bz}$, (for simplicity, we just denote $f(\bz, \bar{\bz})$ by $f(\bz)$ and keep in the mind that $f(\bz)$ only assumes real values) for $\bz := (\bh, \bx) \in \Keps\cap\KF$. 
Suppose that there exists a constant $C_L$ such that
\begin{equation*}
\|\nabla f(\bz + t \Delta \bz) - \nabla f(\bz)\| \leq C_L t\|\Delta\bz\|, \quad \forall 0\leq t\leq 1,
\end{equation*}
for all $\bz\in \Keps\cap\KF$ and $\Delta \bz$ such that  $\bz + t\Delta \bz \in \Keps\cap\KF$ and $0\leq t\leq 1$. Then
\begin{equation*}
f(\bz + \Delta  \bz) \leq f(\bz) + 2\Real( (\Delta \bz)^T \overline{\nabla} f(\bz)) + C_L\|\Delta \bz\|^2
\end{equation*}
where $\overline{\nabla} f(\bz) := \frac{\pa f(\bz, \bar{\bz})}{\pa \bz}$ is the complex conjugate of $\nabla f(\bz) = \frac{\pa f(\bz, \bar{\bz})}{\partial \bar{\bz}}$.
\end{lemma}

\subsection*{Concentration inequality}
We define the matrix $\psi_1$-norm via
\begin{equation*}
\|\BZ\|_{\psi_1} := \inf_{u \geq 0} \{ \E[ \exp(\|\BZ\|/u)] \leq 2 \}.
\end{equation*}

\begin{theorem}\label{thm:bern1}{\bf~\cite{KolVal11}}
Consider a finite sequence of $\CZ_l$ of independent centered random matrices with dimension $M_1\times M_2$. Assume that $R : = \max_{1\leq l\leq L}\|\CZ_l\|_{\psi_1}$ and introduce the random matrix 
\begin{equation}\label{S}
\mathcal{S} = \sum_{l=1}^L \CZ_l.
\end{equation}
Compute the variance parameter
\begin{equation}\label{sigmasq}
\sigma_0^2 = \max\Big\{ \left\| \sum_{l=1}^L \E(\CZ_l\CZ_l^*)\right\|, \left\| \sum_{l=1}^L \E(\CZ_l^* \CZ_l)\Big\| \right\},
\end{equation}
then for all $t \geq 0$
\begin{equation}\label{thm:bern}
\|\mathcal{S}\| \leq C_0 \max\{ \sigma_0 \sqrt{t + \log(M_1 + M_2)}, R\log\left( \frac{\sqrt{L}R}{\sigma_0}\right)(t + \log(M_1 + M_2)) \}
\end{equation}
with probability at least $1 - e^{-t}$ where $C_0$ is an absolute constant.

\begin{lemma}[Lemma 10-13 in~\cite{RR12}, Lemma 12.4 in~\cite{LS17b}] \label{lem:multiple1}
Let $\bu\in\CC^n \sim \CN(\bzero, \I_n)$, then $\|\bu\|^2 \sim \frac{1}{2}\chi^2_{2n}$ and
\begin{equation}
\label{lem:8JB}
\| \|\bu\|^2 \|_{\psi_1} = \| \lag\bu, \bu\rag \|_{\psi_1}  \leq C n
\end{equation}
and
\begin{equation}
\label{lem:9JB}
\E (\bu\bu^* - \I_n)^2 = n\I_n.
\end{equation}
Let $\bq\in\CC^n$ be any deterministic vector, then the following properties hold

\begin{equation}
\label{lem:11JB}
\| (\bu\bu^* - \I)\bq\|_{\psi_1} \leq C\sqrt{n}\|\bq\|,
\end{equation}
\begin{equation}
\label{lem:12JB}
\E[ (\bu\bu^* - \I)\bq\bq^* (\bu\bu^* - \I)] = \|\bq\|^2 \I_n.
\end{equation}
Let $\bv\sim \CN(\bzero, \I_m)$  be a complex Gaussian random vector in $\CC^m$, independent of $\bu$, then 
\begin{equation}\label{lemma:psi}
\left\| \|\bu\| \cdot \|\bv\|\right\|_{\psi_1} \leq C\sqrt{mn}.
\end{equation}
\end{lemma}

\end{theorem}

\section*{Acknowledgement}
S.Ling would like to thank Felix Krahmer and Dominik St\"{o}ger for the discussion about~\cite{SJK16}, and also thank Ju Sun for pointing out the connection between convolutional dictionary learning and this work.


\end{document}